\documentclass[14 pt]{article}

\usepackage{amsfonts,amssymb}
\usepackage{color}
\usepackage{amsthm, amssymb, latexsym, amsmath}
\usepackage[margin=1.25 in]{geometry}

\usepackage{graphics,geometry,epsfig}
\usepackage[active]{srcltx}
\usepackage{times}

\newtheorem{theorem}{Theorem}[section]

\newtheorem{lemma}[theorem]{Lemma}

\newtheorem{remark}{Remark}

\newtheorem{thm}{Theorem}

\def\eq{\begin{equation}}
\def\eeq{\end{equation}}
\def\eqnn{\begin{eqnarray*}}
\def\eeqnn{\end{eqnarray*}}
\def\eqn{\begin{eqnarray}}
\def\eeqn{\end{eqnarray}}
\def\bal{\begin{align}}
\def\eal{\end{align}}
\def\prf{\begin{proof}}
\def\endprf{\end{proof}}
\newcommand{\R}{{\BR}}
\newcommand{\C}{{\mathbb C}}
\newcommand{\BC}{{\mathbb C}}
\newcommand{\BR}{{\mathbb R}}
\newcommand{\RR}{{\mathbb R}}

\newcommand{\cA}{{\cal{A}}}
\newcommand{\cB}{{\cal{B}}}
\newcommand{\cD}{{\cal{D}}}
\newcommand{\cF}{{\cal{F}}}
\newcommand{\cH}{{\cal{H}}}

\newcommand{\cM}{{\cal{M}}}
\newcommand{\cR}{{\cal{R}}}
\newcommand{\cS}{{\cal{S}}}

\newcommand{\cW}{\mathcal{W}}

\def\Fo{{\mathcal F}} 
\def\Wmap{{\mathcal W}}
\newcommand{\eps}{{\varepsilon}}

\newcommand{\vphi}{{\varphi}}
\newcommand{\e}{{\epsilon}}       
\newcommand{\vep}{{{\varepsilon}}}           
\newcommand{\al}{{\alpha}}
\newcommand{\del}{{\delta}}
\newcommand{\lam}{{\lambda}}           
\newcommand{\Lam}{{\Lambda}}
\newcommand{\Om}{\Omega}                
\newcommand{\om}{\omega}
\newcommand{\s}{{\sigma}}
\newcommand{\z}{{\zeta}}

\newcommand{\field}[1]{\mathbb{#1}}

\newcommand{\CC}{\field{C}}     
\newcommand{\fh}{\mathfrak{h}}  
\newcommand{\bchi}{{\overline{\chi}}}

\newcommand{\btau}{\bar{\tau}}
\newcommand{\uw}{{\underline w}}
\newcommand{\ux}{{\underline{x}}}
\newcommand{\uy}{{\underline{y}}}
\newcommand{\oP}{{\overline{P}}}

\newcommand{\p}{\partial}

\newcommand{\lan}{\langle}
\newcommand{\ran}{\rangle}
\newcommand{\la}{\langle}
\newcommand{\ra}{\rangle}
\newcommand{\rIm}{{\rm{Im\, }}}
\newcommand{\rRan}{{\rm{Ran\, }}}
\newcommand{\rRe}{{\rm{Re\, }}}
\renewcommand{\Im}{\mathrm{Im}}
\renewcommand{\Re}{\mathrm{Re}}               
\newcommand{\Ran}{\mathrm{Ran}}             

\newcommand{\Null}{{\rm Null\ }}

\newcommand{\cirS}{\mathop{\bigcirc\kern -.73em {\scriptstyle{\rm S}}}}

\newcommand{\supp}{\mathrm{supp}}

\newcommand{\cern}{\mathrm{Null}}

\newcommand{\one}{\mathbf{1}}
\newcommand{\bfone}{{\bf 1}}
\newcommand{\id}{{\bf 1}}
\newcommand{\1}{{\bf 1}}

\renewcommand{\d}{{\rm d}}

\newcommand{\QED}{\hspace*{\fill}\mbox{$\Box$}}

\newcommand{\adj}{\mathrm{ad}}

\newcommand{\lb}{\left(}
\newcommand{\rb}{\right)}

\newcommand{\const}{\mathrm{const}}

\newcommand{\at}{{p}}

\def\knull{\kappa}
\def\ks{\knull_\sigma}
\def\e{\epsilon}
\def\c{e} 

\def\vac{\Omega}

\renewcommand{\part}{{\rm part}}
\newcommand{\el}{{\mathrm{el}}}
\newcommand{\re}{{\mathrm{f}}}

\def\Phsig{\Phi_{\sigma}}
\def\Psig{\Psi_{\sigma}}

\newcommand{\hf}{H_f}
\def\H(\vep, \theta){H(\vep, \theta)}

\def\hsm{H} 
\def\Hn{H}

\def\Kn{K_\chi} 

\def\Hns{H_{\sigma}}
\def\Kns{K_\sigma}

\def\nablE{\nabla E_\sigma}
\def\Pf{P_f}

\def\Ptot{P_{tot}}
\def\Afs{A_\sigma}
\def\Egs{E_{\sigma}}
\def\Af{A_\chi} 
\def\Eg{E_{}} 
\def\E0{E_{0}}

\newcommand{\DETAILS}[1]{}

\newcommand{\DATUM}{February 26, 2011}              
\title{Renormalization Group and Problem of Radiation\\ 
Les Houches, August, 2010}
\author{Israel Michael Sigal
\thanks{Dept.~of Math.,
Univ. of Toronto, Toronto, Canada; Supported by NSERC Grant No. NA7901} \\
}

\begin{document}
\date{\DATUM}
\maketitle
\begin{abstract} The standard model of non-relativistic quantum electrodynamics describes non-relativistic quantum matter, such as atoms and molecules, coupled to the quantized electromagnetic field. Within this model, we review basic notions, results and techniques in the theory of radiation. We describe the key technique in this area - the spectral renormalization group. Our review is based on joint works with Volker Bach and J\"urg Fr\"ohlich and with Walid Abou Salem, Thomas Chen, J\'er\'emy Faupin and Marcel Griesemer. Brief discussion of related contributions is given at the end of these lectures. This review  will appear in  "Quantum Theory from Small to Large Scales", Lecture Notes of the Les Houches Summer Schools, volume 95, Oxford University Press, 2011.
\end{abstract}


\section{Overview}

We will describe some key results  in theory of radiation for the standard model of non-relativistic electrodynamics (QED). 
The non-relativistic QED 
was proposed in early days of Quantum Mechanics\footnote{It was used, as already known, by Fermi (\cite{fermi}) in 1932 in his review of theory of radiation.} and it describes  quantum-mechanical particle systems coupled to quantized electromagnetic field.  It arises from a standard quantization of the corresponding classical systems 
(with possible addition of internal - spin - degrees of freedom)\footnote{In fact, it is the most general quantum theory obtained as a quantization of a classical system.} and it gives a complete and consistent account of electrons and nuclei interacting with electro-magnetic radiation at low energies. 
In fact, it accounts for all physical phenomena in QED, apart from vacuum polarization.
Sample of issues it addresses are

\begin{itemize}
\item
Stability; 

\item Radiation; 

\item Renormalization of mass;

\item Anomalous magnetic moment;

\item One-particle states;

\item Scattering theory.
\end{itemize}
There was a remarkable progress in the last 10 or so years in rigorous understanding of the corresponding phenomena. In this brief review we will deal with results concerning the first two items. We translate 
them into mathematical terms:

\begin{itemize}
\item
Stability $ \Longleftrightarrow$ Existence of the \textit{ground state};

\item Radiation $ \Longleftrightarrow$ Formation of 
    \textit{resonances} out of  the excited states of particle systems, scattering theory.

\end{itemize}
One of the key notions here is that of the \textit{resonance}. It gives a clear-cut mathematical description of processes of emission and absorption of the electro-magnetic radiation.
\medskip

The key and unifying technique we will concentrate on is the spectral \textit{renormalization group}. It is easily combined with other techniques, e.g. complex deformations (for resonances), the Mourre estimate (for dynamics), analyticity, fiber integral decompositions and Ward identities (used so far for translationally invariant systems).  
It was also extended to analysis of existence and stability of thermal states.

\medskip
\noindent \textbf{Acknowledgements} \\
The author is grateful to Walid Abou Salem, Thomas Chen, J\'er\'emy Faupin, Marcel Griesemer and especially Volker Bach and J\"urg Fr\"ohlich for fruitful collaboration and for all they have taught him in the course of joint work. 

\bigskip

\section{Non-relativistic QED}\label{sec:nrqed}
\subsection{Schr\"odinger equation}\label{sec:model}
We consider a system consisting of $n$ charged particles
interacting between themselves and with external fields, which are coupled to quantized electromagnetic field. The 
 starting point of the non-relativistic QED is the state Hilbert space $\cH=\cH_{p}\otimes\cH_{f}$,
which is the tensor product of the state spaces of the particles,  $\cH_\at$, say, $\cH_\at=L^2(\R^{3n})$,
and of Bosonic Fock space $\cH_{f}$ of the quantized electromagnetic field,  and the standard quantum Hamiltonian $\hsm\equiv H_{g\chi}$ on $\cH=\cH_{p}\otimes\cH_{f}$, given (in the units in which the Planck constant divided by $2\pi$ and the speed of light 
 are equal to $1:\  \hbar=1$ and
$c=1$)  
by 
\begin{equation} \label{Hsm}
\hsm=\sum\limits_{j=1}^n{1\over 2m_j}
(i\nabla_{x_j}-g\Af(x_j))^2+V(x)+H_f
\end{equation}
(see \cite{fermi} and \cite{PauliFierz}). Here, $m_j$ and $x_j$,  $j=1, ..., n$, are the ('bare') particle masses and the particle positions,
$x=(x_1,\dots,x_n)$, $V(x)$ is the total potential affecting particles and $g>0$ is a coupling constant related to the particle charges, 
$\Af:=\check\chi *A$, where $A(y)$ is the \textit{quantized vector potential}, in the Coulomb gauge ($\textrm{div} A(y)=0$), describing the quantized electromagnetic field, 
 and $\chi$ is an \textit{ultraviolet cut-off}, 
\begin{equation}\label{A}
 \Af(y)=\int(e^{iky}a(k)+e^{-iky}a^*(k))\chi(k){d^3k\over \sqrt{|k|}}
\end{equation}
($a(k)$ and $a^*(k)$ are annihilation and creation operators
acting on the Fock space $\cH_{f}\equiv \cF$, see Supplement \ref{sec:crannihoprs} for the definitions),  and $\hf$ is the quantum Hamiltonian of the quantized electromagnetic field, describing the dynamics of the latter, it is given by
\begin{equation} \label{Hf}
\hf \ = \ \int d^3 k \; \om(k) a^*(k) \cdot  a(k),
\end{equation}
where $\om(k) \ = \ |k|$ is the dispersion law connecting the energy
of the field quantum with its wave vector $k$.
For simplicity we omitted the interaction of the spin with magnetic field.
(For a discussion of this Hamiltonian including units, the removal of the center-of-mass motion of the particle system and taking into account the spin of the particles, see Appendix \ref{sec:pot}. Note that our units are not dimensionless. We use this units since we want to keep track of the particle masses. To pass to the dimensionless units we would have to set $m_\el=1$ also.)
The Hamiltonian $H$ determines the dynamics via  the time-dependent Schr\"odinger equation
$$i\partial_t\psi= H\psi,$$
 where $\psi$ is a differentiable path in $\cH=\cH_{p}\otimes\cH_{f}$.

The ultraviolet cut-off, $\chi$, satisfies
$\chi(k)= 1 $ in a neighborhood of $k=0$ and is decaying at infinity on the scale  $\kappa$ and sufficiently fast. 
We assume that $V(x)$ is a generalized $n$-body potential, i.e. it satisfies the assumptions:
\begin{itemize}
\item [(V)] $V(x) = \sum_i W_i(\pi_i x)$, where $\pi_i$ are a linear maps from
$\mathbb{R}^{3n}$ to $\mathbb{R}^{m_i},\ m_i \le 3n $ and $ W_i$ are
Kato-Rellich potentials (i.e. $W_i(\pi_i x) \in
L^{p_i}(\mathbb{R}^{m_i}) + (L^\infty (\mathbb{R}^{3n}))_\eps$ with
$p_i=2$ for $m_i \le 3,\ p_i>2 $ for $m_i =4$ and $p_i \ge m_i/2$
for $m_i > 4$).
\end{itemize}
Under the assumption (V), the operator $H$ is self-adjoint and bounded below.


We assume for simplicity that our matter consists of electrons and the nuclei and that the nuclei are infinitely heavy and therefore are manifested through the interactions only (put differently, the molecules are treated in the Born - Oppenheimer approximation). In this case, the coupling constant $g$ is related to the electron charge $-e$ as $g:= \alpha^{3/2}$, where $\alpha =\frac{e^2}{4\pi \hbar c}\approx {1\over 137}$, the fine-structure constant, and  $m_j =m$.
 It is shown (see Section \ref{sec:liter} and a review in \cite{bach}) 
  for references and  discussion) that the physical electron mass, $m_\el $, is not the same as the parameter $m\equiv m_j$ (the 'bare' electron mass) entering \eqref{Hsm}, but depends on $m$ and $\kappa$. Inverting this relation, we can think of $m$ as a function of $m_{\el}$ and $\kappa$.  
If we fix the particle potential $V(x)$ (e.g. taking it to be the total Coulomb potential), and $m_{\el}$ and $e$, then the Hamiltonian \eqref{Hsm} depends on one free parameter, the bare electron mass $m$ (or the ultraviolet cut-off scale, $\kappa$). 
%
\DETAILS{If we fix the particle potential $V(x)$, 
then the Hamiltonian $\hsm\equiv H_{g\chi}$ depends on \textit{two free parameters}:
\begin{itemize}
\item The coupling constant $g$ (related to the electron charges);

\item The ultraviolet cut-off $\kappa$ (related to the electron renormalized mass. 
\end{itemize}}

\bigskip

\subsection{Stability and radiation}\label{sec:stabrad}

We begin with considering the matter system alone. As was mentioned above, its state space, $\cH_{p}$, is either $L^2(\mathbb{R}^{3n})$ or a subspace of
this space determined by a symmetry group of the particle system, and its Hamiltonian operator,  $H_p$, acting on 
$\cH_{p}$, is given by
\begin{equation} \label{Hp}
H_p:=\sum\limits_{j=1}^n {-1\over 2m_j} \Delta_{x_j}+V(x),
\end{equation}
where $\Delta_{x_j}$ is the Laplacian in the variable $x_j$ and, recall, $V(x)$ is the total potential of the particle system. Under the conditions on the potentials $V(x)$, given above, the operator $H_p$ is self-adjoint and bounded below. Typically, according to the HVZ theorem, its spectrum consists of isolated eigenvalues,  $\epsilon^{(p)}_0 < \epsilon^{(p)}_1 < ... <\Sigma^{(p)}$, and continuum $[\Sigma^{(p)}, \infty)$, starting at the ionization threshold $\Sigma^{(p)}$, as shown in the figure below.  

\DETAILS{\begin{figure}[h!]
	\centering 

  \includegraphics[width=2.5in]{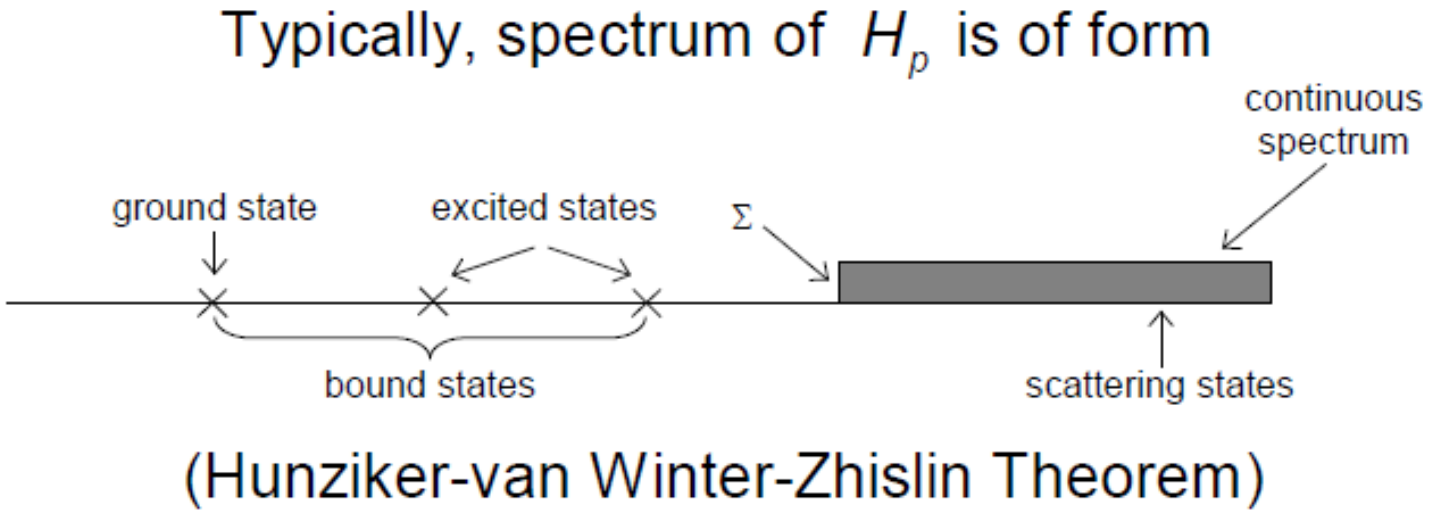} 
\end{figure} }


\bigskip

\bigskip

%
\includegraphics[height=4cm]{lh2.pdf}

\bigskip

\bigskip  
The eigenfunctions corresponding to the isolated eigenvalues are exponentially localized. Thus left on its own the particle system, either in its ground state or in one of the excited states, is stable and well localized in space. We expect that this picture changes dramatically when the total system (the universe) also includes the electromagnetic field, which at this level must be considered to be quantum. As was already indicated above what we expect is the following

\begin{itemize}
\item The stability of the system under consideration is
equivalent to the statement of existence of the ground state of
$H$, i.e. an eigenfunction with the smallest possible energy.

\bigskip

\item The physical phenomenon of radiation is expressed mathematically as
emergence of resonances out of
excited states of a particle system
due to coupling of this system to the quantum electro-magnetic
field.

\end{itemize}
Our goal is to develop the spectral theory of the Hamiltonian $H$  and relate to the properties of the relevant evolution. 
Namely, 
we would like to show that

\bigskip

1) The \textit{ground state} of the particle system is \textit{stable} when the coupling is turned on, while

\bigskip

2) The excited states, generically, are not. They turn into \textit{resonances}.



\bigskip

\subsection{Ultra-violet cut-off} \label{subsec:UV}

 We reintroduce the Planck constant, $\hbar$, speed of light, $c$, and electron mass, $m_\el$, for a moment. Assuming the ultra-violet cut-off $\chi(k)$
decays on the scale $\kappa$, in order to correctly describe the
phenomena of interest,  such as emission and absorption of
electromagnetic radiation,  i.e. for optical and rf modes, we have
to assume that the cut-off energy, $$\hbar c \kappa\ \gg\ \alpha^2 m_\el c^2,\ \mbox{ionization energy, characteristic energy of the particle motion}.$$
On the other hand, we should exclude the energies where the relativistic effects, such as
electron-positron pair creation, vacuum polarization and
relativistic recoil, take place, and therefore we assume
$$ \hbar c \kappa \ll m_\el c^2,\ \mbox{the rest energy of the
the electron}.$$
Combining the last two conditions we arrive
at $\alpha^2 m_\el c/\hbar \ll  \kappa \ll   m_\el c/\hbar$, or in our units,  
$$\alpha^2m_\el   \ll  \kappa \ll m_\el \,\,\ .$$ The Hamiltonian \eqref{Hsm} is obtained by the
rescaling $x \rightarrow \alpha^{-1} x$ and $k \rightarrow \alpha^2
k$ of the original QED Hamiltonian (see Appendix \ref{sec:pot}). After this rescaling, the new cut-off momentum scale, $\kappa'= \alpha^{-2} \kappa$, satisfies
$$m_\el \ll \kappa' \ll \alpha^{-2}m_\el,$$
which is easily accommodated by our estimates (e.g. we can have $\kappa'
=O(\alpha^{-1/3}m_\el)).$ 



\bigskip

\section{Resonances}

As was mentioned above, the mathematical language which describes the physical phenomenon of radiation is that of \textit{quantum resonances}. We expect that the latter emerge out of
excited states of a particle system
due to coupling of this system to the quantum electro-magnetic
field.

Quantum resonances manifest themselves in three different ways:

1) Eigenvalues of complexly deformed Hamiltonian;

2) Poles of the meromorphic continuation of the resolvent across the continuous spectrum;

3) Metastable states.

\bigskip

\subsection{Complex deformation}

To define resonances we use complex deformation method. In order to be able to apply this method 
we choose the ultraviolet cut-off,
$\chi(k)$, so that
\begin{itemize}
\item [] The function $\theta \rightarrow \chi(e^{-\theta} k)$  has an
analytic continuation from the real axis, $\mathbb{R}$, to the strip
$\{\theta \in \mathbb{C} | |\rIm\ \theta | < \pi/4 \}$ as a $L^2
\bigcap L^\infty (\mathbb{R}^3)$ function, \end{itemize} e.g.
$\chi(k)= e^{-|k|^2/\kappa^2}$. For the same purpose, we assume that the potential, $V(x)$,
satisfies the condition:
\begin{itemize}
\item [(DA)] The
the particle potential $V(x)$ is dilation analytic in the sense that
the operator-function
$\theta \rightarrow V(e^{\theta} x)$ $(-\Delta +1)^{-1}$ has an
analytic continuation from the real axis, $\mathbb{R}$, to the strip
$\{\theta \in \mathbb{C} | |\rIm\ \theta | < \theta_0 \}$ for some
$\theta_0 > 0$.
\end{itemize}


To define the resonances for the Hamiltonian $H$
we pass to the one-parameter (deformation) family
\begin{equation}\label{Htheta} 
H_{ \theta} := U_{\theta} H U_{\theta}^{-1},
\end{equation}
where $\theta$ is a real parameter and $U_{\theta},$ on
the total Hilbert space ${\mathcal H}:= {\mathcal H}_p \otimes
{\mathcal F}$, is the one-parameter group of unitary operators, whose action is rescaling particle positions and of
photon momenta:
$$x_j\rightarrow e^\theta x_j\  \mbox{and}\ k\rightarrow e^{-\theta} k.$$

One can show show that:

1) Under a certain analyticity condition on coupling functions, the family $H_{\theta}$  has an analytic
continuation in $\theta$ to the disc $D(0, \theta_0)$, as a type A
family in the sense of Kato;

2) The real eigenvalues of $H_{ \theta},\ \rIm \theta
>0,$ coincide with eigenvalues of $H$ and that complex
eigenvalues of $H_{ \theta},\ \rIm \theta >0,$ lie in the
complex half-plane $\mathbb{C}^-$;

3) The complex
eigenvalues of $H_{ \theta},\ \rIm \theta
>0,$ are locally independent of $\theta$.
The typical spectrum of $[H_{ \theta}\equiv H^{SM}_{ \theta}]|_{g=0},\ \rIm \theta
>0$ (here the superindex SM stands for the standard model) is shown in the figure below.
\bigskip

\bigskip

\bigskip


\includegraphics[height=4cm]{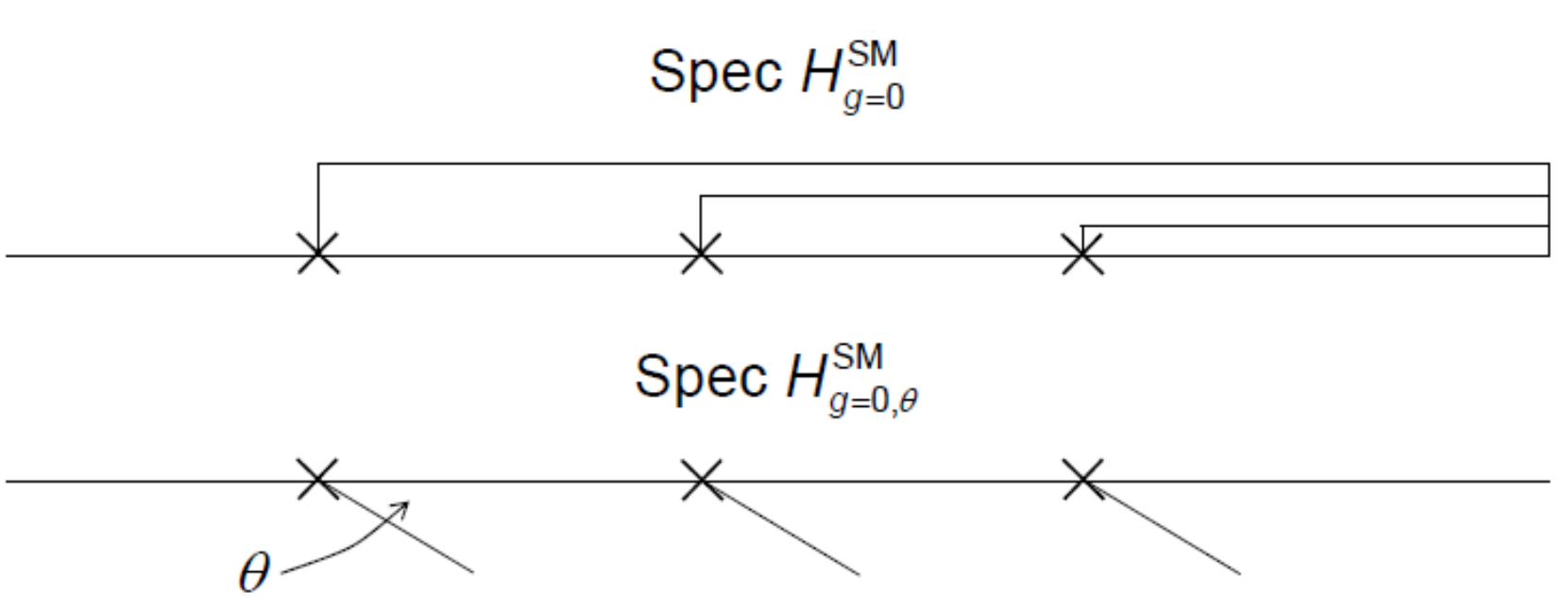}


\bigskip

\bigskip  

We call complex eigenvalues of $H_{ \theta},\ \rIm \theta
>0$ the \textit{resonances} of $H$.



 As an example of the above procedure we consider the complex deformation of the hydrogen atom and photon Hamiltonians $H_{hydr}:=-\frac{1}{2m}\Delta-\frac{\alpha}{|x|}$ and $H_f$: $$H_{hydr \theta}=e^{-2\theta}(-\frac{1}{2m}\Delta)-e^{-\theta}\frac{\alpha}{|x|},\ H_{f \theta}=e^{-\theta}H_f.$$ Let $e^{hydr}_j$ be the eigenvalues of the hydrogen atom.  Then the spectra of these deformations are
 $$\s(H_{hydr\theta})=\{e^{hydr}_j\}\cup e^{-2\Im\theta}[0,\infty),\ \s(H_{f\theta})=\{0\}\cup e^{-\Im\theta}[0,\infty).$$

\bigskip

\subsection{Resonances as poles}
Similarly to eigenvalues, we would like to characterize the resonances in terms of poles of matrix elements of the resolvent $(H -z)^{-1}$ of the Hamiltonian $H$. To this end we have to go beyond the spectral analysis of $H$. Let $\Psi_{\theta}=U_{\theta}\Psi$, etc., for $\theta\in
\mathbb{R}$ and $z\in \mathbb{C}^+$. Use the unitarity of
$U_{\theta}$ 
for real $\theta$, to obtain (the Combes argument)
\begin{equation} \label{meromcont}
\langle\Psi, \  (H -z)^{-1}\Phi\rangle  = \langle
\Psi_{\bar\theta} , ( H_{ \theta} -z)^{-1}\Phi_{\theta}
\rangle.
\end{equation}
Assume now that for a dense set of $\Psi$'s and $\Phi$'s (say, $\cD$, defined below), $\Psi_{\theta}$ and $\Phi_{\theta}$ have analytic continuations into a complex
neighbourhood of $\theta=0$ and continue the r.h.s of \eqref{meromcont} analytically first in $\theta$ into the upper half-plane and then in $z$ across the continuous spectrum. 
This meromorphic continuation has the following properties:
\begin{itemize}
\item The real eigenvalues of $H_{ \theta}$ give real poles of the
r.h.s. of \eqref{meromcont} and therefore they are the eigenvalues of
$H$.

\item The complex eigenvalues of $H_{\theta}$ 
are poles of the meromorphic continuation of the l.h.s. of \eqref{meromcont} across the
spectrum of $H$ onto the second Riemann sheet.
\end{itemize}
The poles manifest themselves physically as  bumps in the scattering cross-section or poles in the scattering matrix. 


%
%
\DETAILS{A simple example illustrating the use of complex deformation in order to continue analytically matrix elements of resolvents is an analytic continuation of the integral $\int_0^\infty \frac{f(\omega)}{\om-z}d\om$. Namely, we would like  to continue analytically this integral  across the semi-axis $(0, \infty)$ from the upper semi-plane $\C^+$ to the second Riemann sheet. 
Assuming that $f(e^{-\theta}\omega)$ is analytic in a neighbourhood of $\theta=0$, we can write such a continuation as $\int_0^\infty \frac{f(e^{-\theta}\omega)}{\om-e^{\theta}z}d\om,\ \Im\theta>0,\ \Im z<0.$}
%
%

The r.h.s. of \eqref{meromcont} has an analytic continuation into a complex neighbourhood of $\theta=0$, if $\Psi, \Phi\in \cD$, where
\begin{equation} \label{D}
\cD :=\bigcup_{n >0,a >0} \Ran \big(\chi_{N \le n}\chi_{|T| \le
a}\big). \end{equation}
Here $N  =  \int d^3 k  a^*(k) a(k)$ be the
photon number operator and $T$ be the self-adjoint
generator of the one-parameter group $U_{\theta},\ \theta \in
\mathbb{R}$.  (It is dense, since $N$ and  $T$ commute.)


\bigskip

\subsection{Resonance states as metastable states}\label{sec:meta}

While bound states are stationary solutions, one expects that resonances to lead to almost stationary, long-living solutions. Let $z_*,\ \rIm z_* \le0,$ be the ground state or
resonance eigenvalue. One expects that for an initial condition,
$\psi_0$, localized in a small energy interval around the ground state
or resonance energy, $\rRe z_*$,  the solution, $\psi=e^{-i H t}\psi_0$, of the
time-dependent Schr\"odinger equation, $i\partial_t\psi=
H\psi$,  is of the form
\begin{equation} \label{ResonDecay}
\psi = e^{-i z_* t}\phi_* +
O_{\textrm{loc}}(t^{-\alpha})+ O_{\textrm{res}}(g^{\beta}),
\end{equation}
for some $\alpha,\ \beta >  0$ (depending on $\psi_0$),
where
\begin{itemize}
\item $\phi_*$ is either the ground state or an excited state of the unperturbed system, depending on whether $z_*$ is the ground
state energy or a resonance eigenvalue;

\item The error term $O_{\textrm{loc}}(t^{-\alpha})$ satisfies $\|(\one+|T|)^{-\nu} O_{\textrm{loc}}(t^{-\alpha})\| \le C t^{-\alpha}$, where $T$ is the generator of the group $U_\theta$, with an appropriate $\nu>0$.

\end{itemize}

For the \textit{ground state}, \eqref{ResonDecay}, without the error term $O_{\textrm{res}}(g^{\beta})$, is called the local decay property (see Section \ref{sec:relat-res}). 
One way to prove it is to
use the formula  connecting the propagator and the resolvent:
\begin{align}\label{prop-intrepr}
& e^{-i H t}f(H)= \frac{1}{\pi} \int_{-\infty}^\infty d\lambda
f(\lambda)e^{-i\lambda t}\rIm(H-\lambda-i0)^{-1}.
\end{align}
Then one controls the boundary values of the resolvent on the spectrum (the corresponding result is called the limiting absorption principle, see Appendix \ref{sec:locdectransf})  and uses properties of the Fourier transform. 

\medskip

For the\textit{ resonances}, \eqref{ResonDecay} implies that $-\rIm z_*$ has the meaning of the decay probability per unit time, and $(-\rIm z_*)^{-1}$, as the life-time of the resonance.
To prove it,
one uses \eqref{meromcont} and the analyticity of its r.h.s. in $z$ and performs in \eqref{prop-intrepr} 
a suitable deformation of the contour of integration  to the second Riemann sheet to pick up the contribution of poles there. This works 
when the resonances are isolated. In the present case, they are not. This is a consequence
of the infrared problem.
Hence, determining the long-time behaviour of $e^{-i H t}\psi_0$ is a subtle problem in this case.

\bigskip

\subsection{Comparison with Quantum Mechanics}\label{sec:qm}
This situation is quite different from the one in Quantum Mechanics
(e.g. Stark effect or tunneling decay) where the resonances are
isolated eigenvalues of complexly deformed Hamiltonians. This makes
the proof of their existence and establishing their properties, e.g.
independence of $\theta$ (and, in fact, of the transformation group
$U_{\theta}$), relatively easy. In the non-relativistic QED (and
other massless theories),
giving meaning of the resonance poles and proving independence of
their location of $\theta$ is a rather involved matter. 

\bigskip

\subsection{Infrared problem}\label{sec:IRprob}

The resonances arise from the eigenvalues of the non-interacting
Hamiltonian $H_{g=0}$. The latter is of the form
\begin{equation}\label{H0} 
  H_{0} = H_\part \otimes \bfone_f + \bfone_\part \otimes H_f\, .
\end{equation}
The low energy spectrum of the
operator $H_0$ consists of branches $[\epsilon^{(p)}_i,
\infty)$ of absolutely continuous spectrum and of the eigenvalues
$\epsilon^{(p)}_i$'s, sitting at the continuous spectrum
'thresholds' $\epsilon^{(p)}_i$'s.
Here, recall, $\epsilon^{(p)}_0 < \epsilon^{(p)}_1 < ... <\Sigma^{(p)}$ are the isolated eigenvalues of the particle Hamiltonian $H_p$.
Let $\phi^{(p)}_i$ be the eigenfunctions of the particle system, while $\Om$ be the photon vacuum. The eigenvalues $\epsilon^{(p)}_i$'s correspond to the eigenfunctions $\phi^{(p)}_i\otimes \Om$ of $H_0$. The branches $[\epsilon^{(p)}_i, \infty)$ of absolutely continuous spectrum are associated with generalized eigenfunctions of the form  $\phi^{(p)}_i\otimes g_\lam$, where $g_\lam$ are the generalized eigenfunctions of $H_f:\ H_fg_\lam=\lam g_\lam,\ 0<\lam<\infty$.

The absence of gaps between the eigenvalues and thresholds is a
consequence of the fact that the photons  are massless. 
To address this problem we use  the spectral renormalization group (RG).  
The problem here is that the leading part of the perturbation in  $H$ is marginal. 
\DETAILS{either assuming the non-physical infrared behaviour of the
vector potential by replacing $|k|^{-1/2}$ in the vector potential
\eqref{A} by $|k|^{-1/2 +\varepsilon}$, with $\varepsilon
>0$, or by assuming presence of a strong confining external potential
so that $V(x) \ge c|x|^2$ for $x$ large.}

\bigskip


\bigskip

%
\section{Existence of the Ground and Resonance States}\label{sec:existgrres}

\subsection{Bifurcation of eigenvalues and resonances}\label{subsec:existgrres}
Stated informally what we show is

\begin{itemize}
\item  The ground state of $H |_{g=0}$ $\Rightarrow$ the ground state  $H$ ($\epsilon_{0} = \epsilon^{(p)}_{0} +O(g^2)$ and $\epsilon_{0} < \epsilon^{(p)}_{0}$); 


\item  The excited states of $H |_{g=0}$ $\Rightarrow$ (generically) the resonances of $H$ 
($\epsilon_{j,k} = \epsilon^{(p)}_{j} +O(g^2)$);



\item There is $ \Sigma>\inf\s(H)$ (the ionization threshold, $\Sigma+\Sigma^{(p)}+O(g^2)$) s.t. for energies $<\Sigma$ that particles are exponentially localized around the common center of mass.
\end{itemize}
For energies $>\Sigma$ the system either sheds off locally the excess of energy and descends into a localized state or breaks apart with some of the particles flying off to infinity.

To formulate this result more precisely, denote
$$\epsilon^{(p)}_{gap}(\nu):=\min \{|\epsilon^{(p)}_i
-\epsilon^{(p)}_j |\ |\ i\neq j,\ \epsilon^{(p)}_i, \epsilon^{(p)}_j
\le \nu \}.$$ 
%

%
{\bf{Theorem 4.1}} \textit{(Fate of particle bound states}).
Fix $e^{(p)}_0 < \nu < \inf \sigma_{ess} (H_p)$ and let
$g\ll \epsilon^{(p)}_{gap}(\nu)$. 
Then for $ g \ne 0,$
\begin{itemize}


\item  $H$ has a ground state, originating from a ground state
of $H |_{g=0}$ ($\epsilon_{0} = \epsilon^{(p)}_{0} +O(g^2),\ \epsilon_{0} < \epsilon^{(p)}_{0}$);

\item Generically, $H$ has no other bound state (besides the ground state);

\item  Eigenvalues, $\epsilon^{(p)}_{j} < \nu,\ j\ne 0$, of $H |_{g=0}\ \Longrightarrow$ 
resonance eigenvalues, $\epsilon_{j, k}$, of $H$;

\item  $\epsilon_{j,k} = \epsilon^{(p)}_{j} +O(g^2)$ and the total multiplicity of $\epsilon_{j, k}$ equals  the multiplicity of
$\epsilon^{(p)}_{j}$; 




\item The ground and resonance states are exponentially localized in the physical space: $$\|e^{\delta|x|}\psi\|<\infty,\ \forall \psi\in\Ran E_\Delta(H),\ \delta<\Sigma^{(p)}-\sup\Delta.$$
\end{itemize}


%
%

\bigskip

\bigskip

\includegraphics[height=5cm]{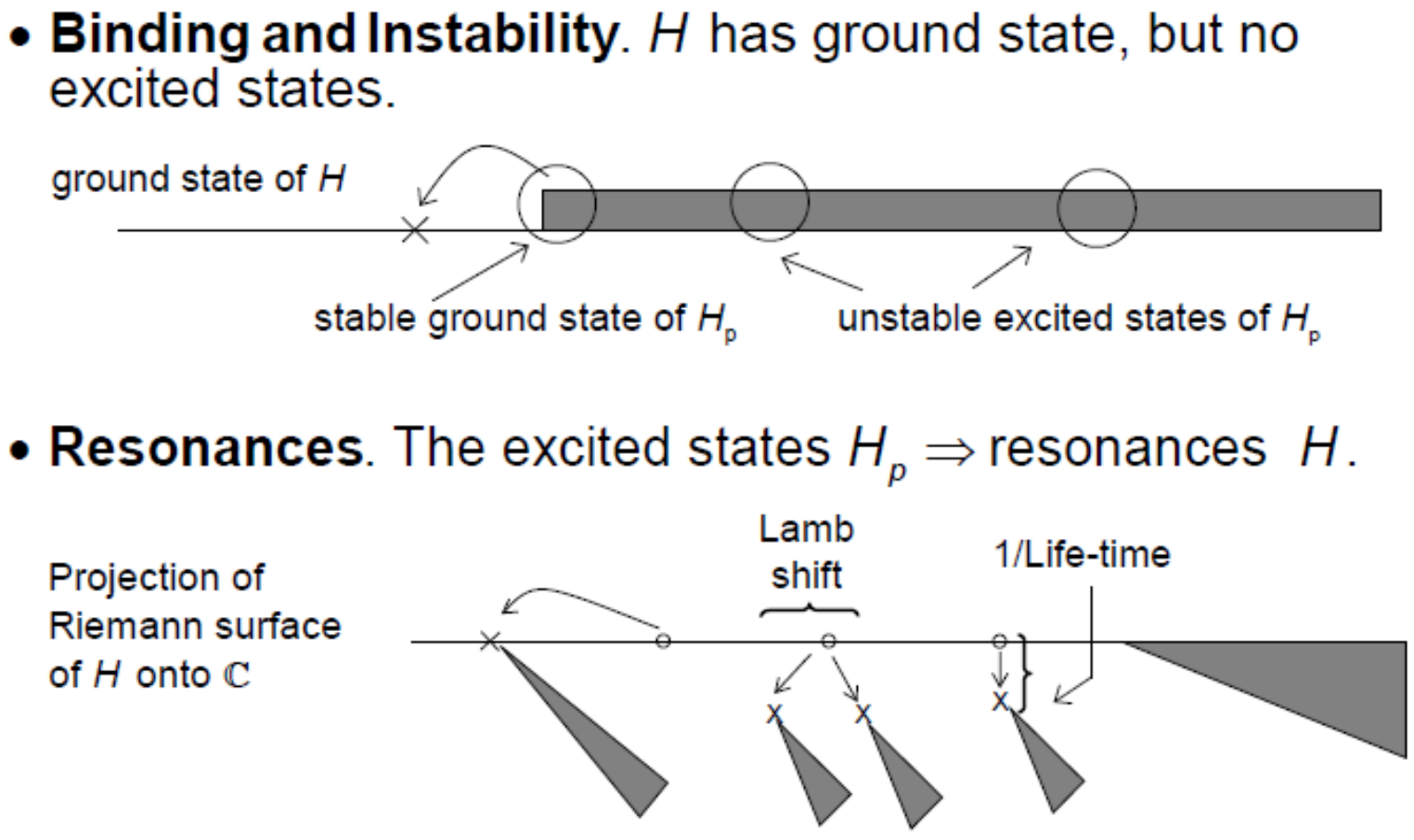}

\bigskip

\bigskip  


%
%


\noindent\textbf{Remark.} The relation $\epsilon_{0} < \epsilon^{(p)}_{0}$ is due to the fact that the electron surrounded by clouds of photons become heavier.
\subsection{Meromorphic continuation across spectrum}\label{subsec:meromcont}

\textbf{Theorem 4.2}. (\textit{Meromorphic continuation of the matrix elements of
the resolvent}) Assume $g\ll \epsilon^{(p)}_{gap}(\nu)$ and let $\epsilon_{0}:=\inf \sigma (H)$ be the ground state energy of $H$.
Then
\begin{itemize}
\item For a dense set (defined in \eqref{D} below) 
of vectors $\Psi$ and $\Phi$, the matrix elements
 $$F(z, \Psi, \Phi):=\langle \Psi, (H-z)^{-1}\Phi\rangle$$
have meromorphic continuations from $\mathbb{C}^+$
across the interval $(\epsilon_{0}, \nu)\subset \s_{ess}(H)$ 
into 
$$\{z \in \mathbb{C}^- |\ \epsilon_{0}
< \rRe z < \nu\}/\bigcup_{0 \le j \le j(\nu)}S_{j, k},$$
where $S_{j, k}$ are the wedges starting at the resonances 
 \begin{equation} \label{Sj}
S_{j, k}:=\{ z \in \mathbb{C} 
\mid \frac{1}{2}\rRe (e^{\theta} (z - \epsilon_{j, k})) \ge  
| \rIm (e^{\theta} (z - \epsilon_{j, k})) |\};
\end{equation}
\item This
continuation has poles at $\epsilon_{j, k}$: $\lim_{z
\rightarrow \epsilon_{j, k}}(\epsilon_{j, k}  -z) F(z, \Psi, \Phi)$ is finite
and $\ne 0$. 
\end{itemize}



\bigskip

\subsection{Discussion}\label{sec:disc}

%
%

\begin{itemize}
\item   Generically, excited states  turn into the resonances, not
bound states. 


\item  The second theorem implies the absolute continuity of the
spectrum and its proof gives also the limiting absorption principle for $H$ (see Appendix \ref{sec:locdectransf} for the definitions). 

\item  The proof of first theorem gives fast convergent
expressions in the coupling constant $g$ for the ground state energy
and resonances.

\item One can show analyticity of $\epsilon_{j,k}$ in the coupling constant $g$ (see Appendix \ref{sec:relat-res} for a result on the ground state energies).




\item  The meromorphic continuation in question is constructed in terms
of matrix elements of the resolvent of a complex deformation,
$H_{ \theta},\ \rIm \theta > 0,$ of the Hamiltonian
$H$.

\item A description of resonance poles is given in Section \ref{sec:relat-res}.
\end{itemize}

\bigskip

\subsection{Approach}\label{sec:Idea}

The main steps in our analysis of the spectral structure of the quantum Hamiltonian $H$ are:
\begin{itemize}
\item  Perform a new canonical transformation  (a generalized Pauli-Fierz transform)
\begin{equation*}
H \rightarrow H^{PF}:= e^{-ig F}H e^{ig F},
\end{equation*}
in order to bring $H$ to a more convenient for our analysis form;

\item Apply the spectral renormalization group (RG) on new -- momentum anisotropic -- Banach spaces. 
\end{itemize}

\bigskip

The main ideas of the spectral RG are as follows:
\begin{itemize}
\item Pass from a single operator $H^{PF}_\theta$ to a Banach space $\cal B$ of Hamiltonian-type
operators;
\item Construct a map, $\cR_{\rho}$, (RG transformation) on $\cB$,
 with the following properties:


(a) $\cR_{\rho}$ is 'isospectral';  

(b) $\cR_{\rho}$ removes the photon degrees of freedom related to
energies $\ge \rho$.
\item Relate the dynamics of semi-flow, $\cR_{\rho}^n, n \ge 1$, 
(called renormalization group) to spectral properties of individual
operators in $\cB$.

\end{itemize}

\bigskip

\bigskip

\section{Generalized Pauli-Fierz transformation} \label{sec-PF-transf}
%
We perform a canonical transformation  (generalized Pauli-Fierz transform) of $H$ in order to bring it to a form which is accessible to spectral renormalization group (it removes the marginal operators). For simplicity, consider one particle of mass $1$. We define the generalized Pauli-Fierz
transformation as:
\begin{equation} \label{HPF}
H^{PF}: = e^{-ig F} H e^{ig F},
\end{equation}
where 
$F(x)$ is the self-adjoint operator 
given by
%
\begin{equation}\label{F}
F(x)=\sum_\lambda
\int(\bar{f}_{x,\lambda}(k)a_{\lambda}(k)+f_{x,\lambda}(k)a_{\lambda}^*(k))
\frac{\chi(k)d^3k}{\sqrt{|k|}},
\end{equation}
with the coupling function $f_{x,\lambda}(k)$ chosen as
\begin{equation}\label{f}
f_{x,\lambda}(k):=
e^{-ikx}\frac{\varphi(|k|^{\frac{1}{2}}e_\lambda(k)
\cdot x)}{\sqrt{|k|}}, 
\end{equation}
with $\varphi \in C^2\ \mbox{bounded, with bounded derivatives and satisfying}\ \varphi'(0)=1.$ 
For the \textit{standard}  Pauli-Fierz transformation, we have $\varphi(s)=s$.

\bigskip


The Hamiltonian $H^{PF}$ is of the same form as $H$. 
Indeed,  using the commutator expansion\\
$e^{-ig F(x)} H_f e^{ig F(x)}= - i g [F,H_f] - g^2 [F, [F,
H_f]],$ we compute
\begin{equation} \label{H^PF}
H^{PF} = \frac{1}{2m} (p + g A_{\chi\varphi}(x))^2 + V_g(x) + H_f + gG(x), 
\end{equation}
where
\begin{equation*}
A_{\chi\varphi}(x)=\sum_\lambda
\int(\bar{\varphi}_{x,\lambda}(k)a_{\lambda}(k)
+\varphi_{x,\lambda}(k)a_{\lambda}^*(k)){\chi(k)d^3k\over \sqrt{|k|}},
\end{equation*}
with the new coupling function $\varphi_{\lambda, x}(k):=
e_{\lambda}(k)e^{-ikx}- \nabla_x f_{x,\lambda}(k)$ and
 $$
V_g(x):= V(x) + 2
g^2\sum_\lambda \int |k| |f_{x,\lambda}(k)|^2d^3k, $$
\begin{equation*}
G(x):=- i\sum_\lambda
\int|k|(\bar{f}_{x,\lambda}(k)a_{\lambda}(k)-f_{x,\lambda}(k)a_{\lambda}^*(k))
\frac{\chi(k) d^3k}{\sqrt{|k|}}.
\end{equation*}
The potential $V_g(x)$ is a small perturbation of $V(x)$ and the operator $G(x)$ is easy to control.  The new coupling function has better infrared behaviour for bounded $|x|$:  
\begin{equation}\label{chi-estim1}
|\varphi_{\lambda, x}(k)|
\le \const \min (1, \sqrt{|k|}\la x\ra).
\end{equation}

To prove the results above we first establish the spectral properties of the generalized Pauli-Fierz Hamiltonian $H^{PF}$ and then transfer the obtained information to the original Hamiltonian $H$.

%
%
%
\DETAILS{(The terms $gG$ and $V_g - V$ come from the commutator expansion
$e^{-ig F} H_f e^{ig F}$ $= - i g [F,H_f] - g^2 [F, [F,
H_f]]$.) Observe that the operator-family $A_1(x)$ is of the form

The fact that the operators $A_1$ and $G$ have better infra-red
behavior than the original vector potential $A$, is used in proving,
with a help of a renormalization group, the existence of the ground
state and resonances for the Hamiltonian $H_g^{SM}$.


We mention for further references that the operator \eqref{H^PFa} can be
written as
\begin{equation} \label{Hpf}
H_g^{PF} \ = \ H_{0}^{PF} \, + \, I_{g}^{PF} ,
\end{equation}
where $H^{PF}_{0}=H_0 + 2 g^2\sum_\lambda \int
|k||f_{x,\lambda}(k)|^2d^3k + g^2\sum_\lambda
\int{|\chi_\lambda(k)|^2\over |k|}d^3k$, with $H_0 := H_p+H_f$
and $I_{g}^{PF}$ is defined by this relation. Note that the operator
$I_{g}^{PF}$ contains linear and quadratic terms in the creation and
annihilation operators and that the operator $H^{PF}_{0}$ is of the
form $H^{PF}_{0}= H^{PF}_{p}+H_{f}$ where
\begin{equation} \label{Hpg}
H^{PF}_{p}:=H_{p}+2 g^2\sum_\lambda \int |k|
|f_{x,\lambda}(k)|^2d^3k + g^2\sum_\lambda
\int{|\chi_\lambda(k)|^2\over |k|}d^3k
\end{equation}
with $H_p$ given in \eqref{Hp}.

%

Since the operator $F(x)$ in \eqref{HPF} is self-adjoint, the
operators $H_g^{SM}$ and $H_g^{PF}$ have the same eigenvalues with
closely related eigenfunctions and the same essential spectra.}
%
%
%




\bigskip

\section{Renormalization Group Map}\label{sec:RG}


The \textit{renormalization map} is defined  on Hamiltonians acting on $\cH_f$ which as follows
\begin{equation}
\cR_{\rho }=\rho^{-1} S_{\rho}\circ F_{ \rho}, \label{RGmap}
\end{equation}
where $ \rho >0$, $S_\rho:
\cB[\cH] \to \cB[\cH]$ is the \textit{scaling transformation}:
%
%
%
\begin{equation} \label{Srho}
S_\rho(\one) \ := \ \one , \hspace{5mm} S_\rho (a^\#(k)) := \
\rho^{-3/2} \, a^\#( \rho^{-1} k) ,
\end{equation}
and  $F_{ \rho}$ is
the \textit{ Feshbach-Schur map}, or decimation, map,
\begin{equation} \label{Fesh}
 F_{\rho} (H ) \ := \ \chi_{\rho} (H  -
 H \bchi_{\rho}  (\bchi_{\rho} H \bchi_{\rho})^{-1} \bchi_{\rho} H) \chi_{\rho} ,
\end{equation}
where $\chi_{\rho}$ and $ \overline{\chi}_{\rho}$ is a pair of orthogonal projections, defined as
$$\chi_{\rho} =\chi_{H_{p\theta}=e_j}\otimes\chi_{H_f\le\rho}\ \quad \mbox{and}\ \quad \overline{\chi}_{\rho}:=\mathbf{1}- \chi_{\rho}.$$

\noindent\textbf{Remark.} For simplicity we defined the decimation map as the Feshbach-Schur map. Technically it is more convenient to use the \textit{smooth Feshbach-Schur map}, which we defined and discussed in Appendix \ref{sec:sfm}.   The smooth Feshbach-Schur map 
uses 'smooth' projections which form a partition of unity
$\chi_{\rho}^{2}+ \overline{\chi}_{\rho}^{2} = \mathbf{1}$, instead of true projections as defined above.

%

\bigskip

%
The map  $F_{ \rho}$ is \textit{isospectral} in the sense of the  following theorem:
\begin{theorem}\label{thm:isospF}
{\em \begin{itemize}
\item[(i)] $ \lambda \in \rho(H ) \Leftrightarrow 0 \in \rho(F_{\rho}
(H - \lambda))$; 
\item[(ii)]  
$H\psi = \lambda \psi\ \Longleftrightarrow$
$F_{\rho} (H - \lambda) \, \vphi = 0;$
\DETAILS{\item[(iii)] If $\vphi
\in \Ran\, \pi \setminus \{0\}$ solves $F_{\rho} (H - \lambda)
\, \vphi = 0$ then $\psi := Q_{\pi} (H - \lambda) \vphi \in \cH
\setminus \{0\}$ solves $H\psi = \lambda \psi$;}
\item[(iii)] 
$\dim \cern (H - \lambda) = \dim \cern F_{\rho} (H -
\lambda)$;
\item[(iv)] 
$(H - \lambda)^{-1}$ exists $ \Longleftrightarrow$ $F_{\rho} (H - \lambda)^{-1}$
exists.

%
\end{itemize}}
\end{theorem}
For the proof of this theorem as well as for the  relation between $\psi$ and $\varphi$ in (ii) and between $(H - \lambda)^{-1}$ and $F_{\rho} (H - \lambda)^{-1}$ in (iv) see Appendix \ref{sec:isosprel}.

\bigskip

\section{A Banach Space of Hamiltonians} \label{sec:Bansp}
\DETAILS{We construct a Banach space of Hamiltonians on which the
renormalization transformation will be defined. In order not to
complicate notation unnecessarily we will think about the creation-
and annihilation operators used below as scalar operators,
neglecting the helicity of photons.
We explain at the end of Supplement A how to reinterpret the
corresponding expression for the photon creation- and annihilation
operators.
%
%
 and $m,n \ge 0$. Given functions $w_{m,n}:
I\times B_1^{m+n} \rightarrow \mathbb{C}, m+n > 0$, we consider
monomials, $W_{m,n} \equiv W_{m,n}[w_{m,n}]$, in the creation and
annihilation operators defined as follows:}
%
%
%
We will study operators on the subspace $\rRan\chi_1$ of  the Fock space $\cF$. Such operators are said to be
in the \textit{generalized normal form} if they can be written as:
\begin{eqnarray} \label{H} 
H&=&\sum_{m+n\ge 0} W_{m+n},\\
W_{m+n}&=&\int_{B_1^{m+n}}  \prod_{i=1}^{m+n} d^3 k_i \; 
\prod_{i=1}^m a^*(k_i ) \, w_{m,n} \big( \hf ; k_{(m+n)} \big) \, \prod_{i=m+1}^{m+n} a(k_i ), \nonumber
\end{eqnarray}
where $B_1^r$ denotes the Cartesian product of $r$ unit balls in
$\RR^{3}$, $ k_{(m)} \: := \: (k_1, \ldots, k_m)$ and $w_{m,n}: I\times B_1^{m+n}\to \C,\ I:=[0,1]$.
We 
sometimes we display the dependence of $H$ and $W_{m,n}$ on the coupling functions $\uw:= (w_{m,n},\ m+n\ge 0)$ by writing $H[\uw ]$ and $W_{m,n}[\uw ]$.

We assume that 
the functions $w_{m,n}( r, , k_{(m+n)})$ are  
continuously differentiable in $r \in I$, symmetric w.~r.~t.\ the
variables $ (k_1, \ldots, k_m)$ and $ (k_{m+1}, \ldots, k_{m+n})$
and obey $\| w_{m,n} \|_{\mu,1} \ :=
\sum_{n=0}^{1} \|  \partial_r^n w_{m,n} \|_{\mu} \ < \ \infty ,
$
where $\mu \ge 0$ and
\begin{equation} \label{wmn}
\| w_{m,n} \|_{\mu } \ := 
\max_j \sup_{r \in I, k_{(m+n)} \in B_1^{m+n}} \big| | k_j|^{-\mu}\prod_{i=1}^{m+n} | k_i|^{1/2} 
w_{m,n}(r ; k_{(m+n)}) \big|.
\end{equation}
Here $k_j \in \mathbb{R}^3$ is the $j-$th  $3-$
vector in $k_{(m,n)}$ over which we take the supremum.
Note that these norms are anisotropic in the total momentum space.

\medskip

For $\mu \ge0$ and $0 < \xi < 1$ 
we define the Banach space 
\begin{equation} \label{Bmuxi}
\cB^{\mu\xi}:=\{H\ :\ \big\|  H \big\|_{\mu,\xi} \ := \ \sum_{m+n \geq 0}
\xi^{-(m+n)} \; \| w_{m,n} \|_{\mu, 1} \ < \ \infty \}.
\end{equation}
We mention some properties of these spaces 
\begin{itemize} 
\item For any $\mu \ge 0$ and $0 < \xi < 1$, the map $H : \uw \to H[\uw ]$,
given in \eqref{H}, is one-to-one.
\item If $ H$ is self-adjoint, then so are $W_{0,0}$ and $ \sum_{m+n \geq 1} \chi_1 W_{m,n}\chi_1$ (see \eqref{H}). 
\end{itemize}
\textbf{Remarks.} 1) Unlike the Banach spaces defined in \cite{bfs1, bfs2, bcfs1}, the Banach spaces are anisotropic in the momentum space. This is needed to overcome the problem of marginal operators which arise in the renormalization group approach. 

2) The self-adjointness statement follows from \cite{bcfs1}, Eq. (3.33).


\bigskip

\subsection{Basic bound} \label{sec:bnd}
The following bound shows that our Banach space norm control the operator norm and the terms with higher numbers of creation and annihilation operators make progressively smaller contributions:

\begin{theorem} 
{\em Let
$\chi_\rho\equiv\chi_{H_f\le\rho}$.  Then for all $\rho >0$ and $m+n \geq 1$}
\begin{equation}
\big\| \chi_\rho \, H_{m,n} \, \chi_\rho
 \big\|
\ \leq \ \frac{\rho^{m+n+\mu}}{\sqrt{m! \, n!} } \, \| w_{m,n}
\|_{\mu }.
\end{equation}
\end{theorem}

\noindent \textit{Sketch of proof.}
For simplicity we prove this inequality for $m=n=1$. Let $\phi\in \cF$ and $\Phi_k=a(k ) \chi_\rho\phi$.  We have
\begin{eqnarray} \nonumber
\la   \chi_\rho\phi, H_{1,1} \, \chi_\rho\phi\ra &=& \int_{B_1^{2}}  \prod_{i=1}^{2} d^3 k_i \; 
\la  \Phi_{k_1},  \, w_{1,1} \big( \hf ; k_1, k_2 \big) \,  \Phi_{k_2}\ra.
\end{eqnarray}
Now we write $\chi_\rho=\chi_\rho\chi_{2\rho}$ and pull  $\chi_{2\rho}$ toward $w_{1,1} ( \hf ; k_1, k_2)$ using the pull-trough formulae
$$a(k) \, f(\hf) \ = \ f( \hf + |k| ) \, a(k),\
f(\hf) \, a^*(k) \ = \ a^*(k) \, f( \hf +|k|)$$
(see Appendix \ref{sec:pullthrough}). This gives
\begin{eqnarray*}
&&|\la \phi, \chi_\rho \, H_{1,1} \, \chi_\rho\phi\ra|\\
&=& |\int_{|k_i|\le 2\rho, i=1,2}  \prod_{i=1}^{2} d^3 k_i \; \la \Phi_{k_1}, \chi_{2\rho-|k_1|} w_{1,1} \big( \hf ; k_1, k_2 \big) \chi_{2\rho-|k_2|}  \Phi_{k_2}\ra|\\
 &\le& \int_{|k_i|\le 2\rho, i=1,2}  \prod_{i=1}^{2} d^3 k_i \; \| \Phi_{k_1}\|\ \| w_{1,1} \big( \hf ; k_1, k_2 \big) \|\ \|  \Phi_{k_2}\| \\
 &\le& \bigg(\int_{|k_i|\le 2\rho, i=1,2}  \prod_{i=1}^{2} d^3 k_i \frac{\| w_{1,1} \big( \hf ; k_1, k_2 \big) \|^2}{|k_1||k_2|}\bigg)^{1/2}\int d^3 k\|\sqrt{|k|}  \Phi_{k}\|^2.
\end{eqnarray*}
Now, using $\| w_{1,1} \big( \hf ; k_1, k_2 \big) \|\le \| w_{1,1}  \|_\mu\frac{|k_1|^{\mu}+|k_2|^{\mu}}{|k_1|^\frac{1}{2}|k_2|^\frac{1}{2}}$ and $\int d^3 k\|\sqrt{|k|}  \Phi_{k}\|^2=\|\sqrt{H_f}  \chi_\rho\phi\|^2$, we find
$$|\la \phi, \chi_\rho \, H_{1,1} \, \chi_\rho\phi\ra|\lesssim \rho^{2+\mu} \| w_{1,1}\|_{\mu }.$$

\bigskip

\subsection{Unstable, neutral and stable components} \label{sec:comp}

We decompose $H\in \cB^{\mu\xi}$ into the components 
$E:=\la\Om, H\Om\ra,\ T:=H_{0,0}-\la\Om, H\Om\ra,\ 
W:=\sum_{m+n \geq 1} W_{m,n},$ 
so that
\begin{equation}\label{Hsplit}
  H=E\one + T + W.
 \end{equation}
If we assume $\sup_{r \in [0,\infty)}| T'(r) - 1 | \ll 1$, then we have $T \sim H_f$. These Hamiltonian components scale as follows
\begin{itemize}
\item 
$\rho^{-1} S_\rho \big( \hf \big) \ = \  \hf $ 
($\hf$ is a \emph{fixed point} of $\rho^{-1} S_\rho$);

\item  $\rho^{-1}  S_\rho(E \cdot \one) =
\rho^{-1} E \cdot \one$ ($E \cdot \one$ \emph{
expand} under $\rho^{-1} S_\rho$ at a rate $\rho^{-1}$);

\item $
 \| S_\rho(W_{m,n}) \|_{\mu} \le \rho^{\alpha} \, \| w_{m,n} \|_{\mu},\ \alpha:=m+n- 1 +\mu\del_{m+n=1}$ (
 $W_{mn}$ \textit{contract} under $\rho^{-1} S_\rho$, if $\mu> 0$).

\end{itemize}
Thus for $\mu >0$, $E,\ T,\ W$ behave, in the terminology of the renormalization group approach, as relevant, marginal, and irrelevant operators, respectively. For $\mu =0$, the operators $W_{mn},\ m+n=1,$ become marginal.

\bigskip

\section{Action of Renormalization Map} \label{sec:actionRG}

To control the components $E, T, W$ of $H$ we introduce, for $\alpha, \beta, \gamma >0$, the following polydisc:
\begin{eqnarray} 
\cD^{\mu}(\alpha,\beta,\gamma) && :=  \Big\{ H=E+T+W \in
\cB^{\mu\xi} \ | \ |E|\leq\alpha, \nonumber\\
 &&\sup_{r \in [0,\infty)}| T'(r) - 1 | \leq
\beta,\  \| W \|_{\mu, \xi}\leq\gamma \Big\}.\nonumber
\end{eqnarray}
(Strictly speaking we should write 
$\chi_{H_{p\theta}=e_j}\otimes \cD^{\mu}(\alpha,\beta,\gamma)$ instead of $\cD^{\mu}(\alpha,\beta,\gamma)$ (for various sets of parameters $\alpha,\beta,\gamma$) in the statement below.)
\begin{theorem} \label{thm:act-RGmap}
{\em Let $ 0<\rho<1/2,\ \alpha,\ \beta,\ \gamma \le \rho/8$ and $\mu_0=1/2$. Then there is $c>0$, s.t.

\begin{itemize}
\medskip
\item $\cR_\rho(H^{PF}_\theta)\in \cD^{\mu_0}(\al_0, \beta_0, \gamma_0),\ \al_0=c g^2 \rho^{\mu_0-2},  \beta_0=c g^2\rho^{\mu_0 -1}$,  $\gamma_0 =c g\rho^\mu_0$,

    provided $g \ll 1$;

\item $ \cD^{\mu}(\alpha,\beta,\gamma)\subset D(\cR_\rho)$, provided $\mu>0$;

\item
    $\cR_\rho :\cD^{\mu}(\alpha,\beta,\gamma)\rightarrow \cD^{\mu}(\alpha',\beta',\gamma'),$
    continuously, with 

  $\alpha'=\rho^{-1}\alpha+c\lb\gamma^2/2\rho\rb, \beta'=\beta+c\lb\gamma^2/2\rho\rb,\ \gamma'=c\rho^\mu\gamma.$
\medskip

\end{itemize}}
\end{theorem}
%
\DETAILS{\begin{theorem} {\em Let 
$\mu>0$. Then for an absolute constant $c$ and
for any $ 0<\rho<1/2,\ \alpha,\beta \le \frac{\rho}{8},$
$\gamma \le \frac{\rho}{2c}$, we have that
\begin{equation*}
\cR_\rho 
:\cD^{\mu,s}(\alpha,\beta,\gamma)\rightarrow
\cD^{\mu,s}(\alpha',\beta',\gamma'),
\end{equation*}
continuously, with 
\begin{equation*}
\alpha'=\alpha+c\lb\gamma^2/2\rho\rb,
\beta'=\beta+c\lb\gamma^2/2\rho\rb,\ \gamma'=30
c^2\rho^\mu\gamma.
\end{equation*}}
\end{theorem}}
%
%
\noindent \textit{Sketch of proof of the second and third properties.} 1) 
$\cD^{\mu}(\rho/8, 1/8, \rho/8)\subset D(\cR_\rho)$. Since $W:= H-E-T$ defines a bounded operator on $\cF$, we
only need to check the invertibility of $H_{\tau \chi_\rho}$ on
$\Ran \,\bchi_\rho$. The operator $E+T$ is
invertible on $\Ran \,\bchi_\rho$:
for all $r \in [3\rho/4, \infty)$
\begin{eqnarray*} 
\Re\ T(r) + \Re\ E & \geq & r \, - \, | T(r) - r | \, - \, |E|
\nonumber \\ & \geq & r \big( 1 \, - \, \sup_{r} | T'(r) - 1 | \big) \: - \:  |E| \nonumber \\
& \geq & \frac{3 \, \rho}{4} ( 1 - 1/8 ) \: - \: \frac{\rho}{8} \
\geq \ \frac{ \rho}{2} \nonumber \\
&\Rightarrow & E+T\ \mbox{is invertible and}\ \|(E+T)^{-1}\|
\le 2/\rho.
\end{eqnarray*}
Now, by the basic estimate, $\big\| W \|\leq
 \rho/8$ and therefore,
\begin{eqnarray*}  && \big\|\bchi_\rho W \bchi_\rho (E +
T)^{-1}\|\leq 1/4  \\
&\Rightarrow &  
E + T+\bchi_\rho W \bchi_\rho\  \mbox{is invertible on}\  \Ran \,\bchi_\rho \\
&\Rightarrow & \cD^{\mu}(\rho/8, 1/8, \rho/8)\subset D(F_\rho)=D(\cR_\rho).
\end{eqnarray*}



2) $\cR_\rho :\cD^{\mu}(\alpha,\beta,\gamma)\rightarrow \cD^{\mu}(\alpha',\beta',\gamma')$ (normal form of $\cR_{\rho}(H)$).  Recall that $\chi_\rho \equiv \chi_{H_f \le
\rho}$ and $\bchi_\rho:= 1 -\chi_\rho$. Let $H_0:=E+T$, so that $H=H_0+W$. We have shown above
\begin{equation*}
\big\|  H_0^{-1} \bchi_\rho \big\| \ \leq \ \frac{2}{\rho}
\hspace{5mm} \mbox{and} \hspace{5mm} \| W \| \ \leq \ \frac{\rho}{8}.
\end{equation*}
In the Feshbach-Schur map, $F_{\rho}$,
%
\begin{eqnarray*}
 F_{\rho} \big( H \big) \ = \  \chi_\rho \big(H_0+W
 -  W \, \bchi_\rho \big( \bchi_\rho (H_0+W)
\bchi_\rho \big)^{-1} \bchi_\rho \, W \,\big) \chi_\rho ,
\end{eqnarray*}
 we expand the resolvent $( \bchi_\rho (H_0+W)
\bchi_\rho \big)^{-1}$ 
in the norm convergent Neumann series
\begin{equation*}
 F_{\rho} \big( H \big) \ = \ \chi_\rho\big[H_0  + \sum_{s=0}^\infty
(-1)^{s} \,  W \big(  H_0^{-1}\bchi_\rho^2 \; W
\big)^{s}\big] \chi_\rho .
\end{equation*}
Next, we transform the right side to the generalized normal form using generalized Wick's theorem.

\bigskip

\textbf{Generalized Wick's theorem.} To write the product $W \big(  H_0^{-1}\bchi_\rho^2 \; W\big)^{s}$
in the generalized normal form we  pull the annihilation operators, $a$, to the right and the creation operators, $a^*$, to the left, apart from those which enter $H_f$. We use the rules (see Appendix \ref{sec:pullthrough}): 
 $$a(k)a^*(k') = a^*(k')a(k) +\delta(k-k'),$$
\begin{equation*}
a(k) \, f(\hf) \ = \ f( \hf + |k| ) \, a(k),\
f(\hf) \, a^*(k) \ = \ a^*(k) \, f( \hf +|k|)  .
\end{equation*}
Some of the creation and annihilation operators reach the extreme
left and right positions, while the remaining ones contract (see the figure below). The
terms with $m$ creation operators on the left
and $n$ annihilation operators on the right
contribute to the $(m,n)-$ formfactor, $w^{(s)}_{m,n}$, of the
operator $W \big(  H_0^{-1}\bchi_\rho^2 \; W\big)^{s}$. As the result we obtain the generalized normal form of $F_\rho(H)$:
\begin{equation*}
F_\rho(H) = \ \sum_{m+n\ge0}
W^{'}_{m,n}.
\end{equation*}
The term $W^{'}_{0,0}=\la W^{(s)}_{0,0}\ra_\Om +(W^{(s)}_{0,0}-\la W^{(s)}_{0,0}\ra_\Om)$ contributes the corrections to $E+T$.

%
%
\DETAILS{
\begin{center}
\psset{unit=1cm} \pspicture(-4,-5)(8,3.5)

\pscustom[linestyle=none,fillstyle=solid,fillcolor=lightgray]{
\psbezier(2,1)(3,0.5)(4,0)(4.5,-1.5)
\psbezier[liftpen=1](0,-4.5)(-0.25,-3)(-1,-2)(-2,-1)}
\psbezier[linewidth=0.5pt,linestyle=dashed](1,0.5)(2.2,-0.1)(3,-0.5)(3.5,-2.25)
\rput(0.5,-3.5){$\mathcal{M}_s$} \rput(0.5,0.8){$w H_f$}

\psset{linewidth=0.5pt} \psline(-3,-1.5)(3,1.5)
\rput(3.5,1.5){$\mathcal{M}_{fp}$}

\psline(0,-1)(0,2)\rput(0,2.3){$\mathcal{M}_u$}

\psline[linewidth=0.5pt,linestyle=dashed](1,0.5)(1,2.5)
\psbezier[linewidth=1pt]{->}(3,-0.7)(2,0.5)(1.2,0.5)(1.2,2)
\qdisk(1.4,0.97){2pt}\psline[linewidth=0.3pt]{<-}(1.5,1.1)(2.2,2)
\rput(3,2.2){$\mathcal{R}^{n}_{\rho}(H_\theta-\lambda)$}

\psline[linewidth=0.5pt,linestyle=dashed](3,0.2)(3,-0.7)
\qdisk(3,0.2){2pt}\qdisk(3,-0.7){2pt}\rput(3.4,0.2){$H_\theta$}
\psline[linewidth=0.3pt]{<-}(3.1,-0.7)(4.5,-0.5)\rput(5.2,-0.5){$H_\theta-\lambda$}
\endpspicture

Stable and unstable manifolds.
\end{center} }
%
%

\bigskip

\includegraphics[height=4cm]{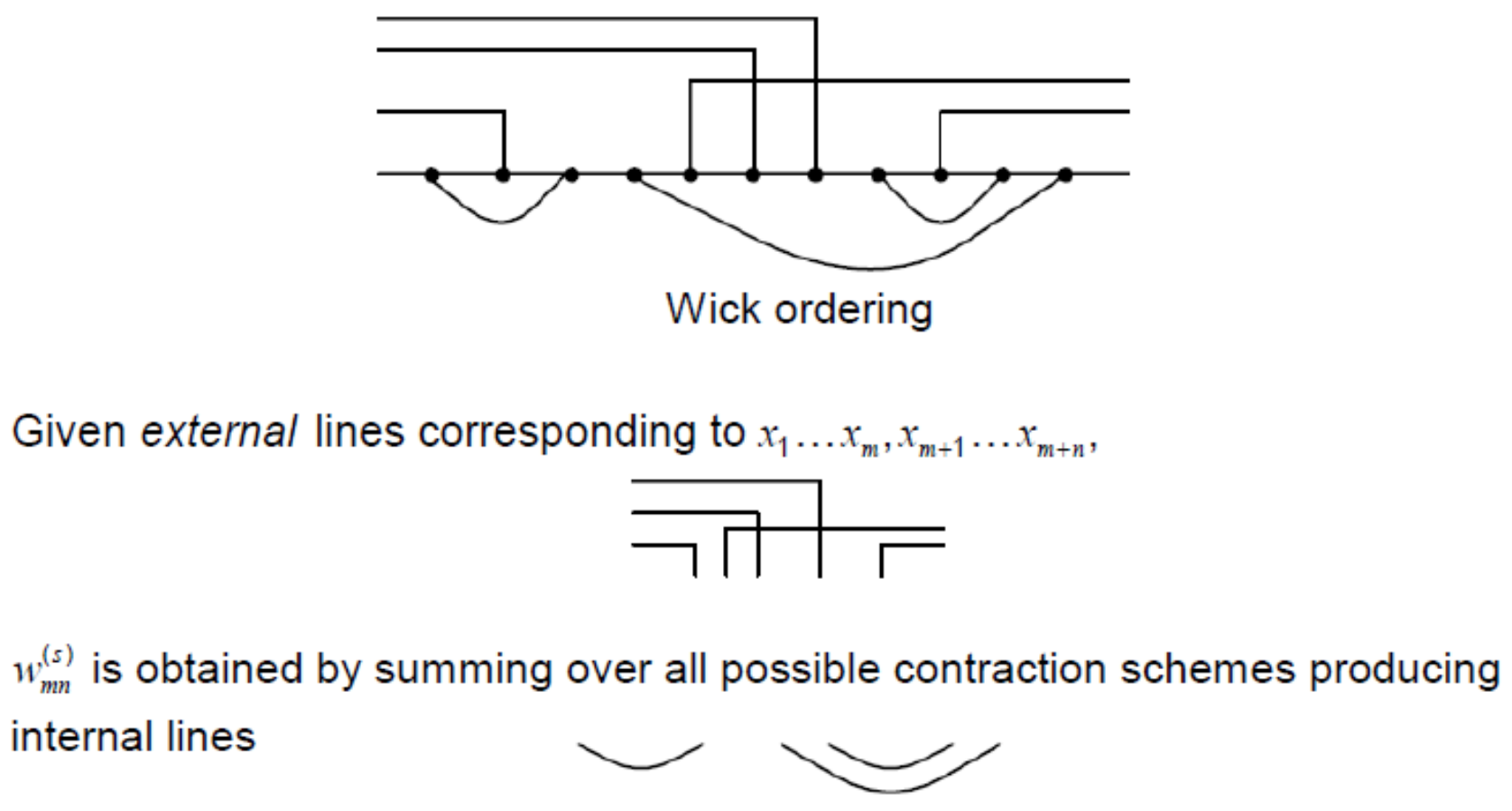}  
\bigskip

This is the standard way for proving the Wick theorem, taking into account the presence of $H_f-$ dependent factors. (See \cite{bfs1} for a different, more formal proof.)

\bigskip


\textbf{Estimating formfactors.} The problem here is that the number of terms generated by various contractions is 
 $O(s!)$. 
 Therefore a simple majoration of the series for  the $(m,n)-$ formfactor, $w^{(s)}_{m,n}$, of the operator $W \big(  H_0^{-1}\bchi_\rho^2 \; W\big)^{s}$ will diverge badly.

To overcome this we re-sum the series 
by, roughly, representing the sum over all contractions, for a given $m$ and $n$, as
$$w^{(s)}_{m,n}\sim \la\Omega, [W \big(  H_0^{-1}\bchi_\rho^2 \; W\big)^{s}]_{m,n} \mbox{}\Omega\ra,$$
where $[W \big(  H_0^{-1}\bchi_\rho^2 \; W\big)^{s}]_{m,n}$ is $W \big(  H_0^{-1}\bchi_\rho^2 \; W\big)^{s}$, with $m$ escaping creation operators and $n$ escaping annihilation operators deleted. 
Now the estimate of $w^{(s)}_{m,n}$ is straightforward and can be written (symbolically) as
$$\|w^{(s)}_{m,n}\|\lesssim \|\chi_{\rho'} \; W'\chi_{\rho'} \|^{s+1},$$
( for the operator norm, and similarly for $\cB^{\mu,\xi}-$norm.

\bigskip

\section{Renormalization Group}


To analyze spectral properties of individual
operators in $\cB^{\mu\xi}$ we use the discrete semi-flow, $\cR_{\rho}^n, n \ge 1$
(called renormalization group) , generated by the renormalization transformation, $\cR_{\rho}$.
By Theorem \ref{thm:act-RGmap}, in order to iterate $\cR_\rho$ we have to control the expanding direction: $\cR_{\rho}(\z\id)=\rho^{-1}\z\id.$
To control this direction, we adjust, inductively, at each step the constant component $\la H\ra_\Omega :=\la \Omega, H\Omega\ra$ of the initial Hamiltonian, $H$: 
$$|\la H\ra_\Omega - e_{n-1}| \leq \frac{1}{12}  \rho^{n+1},$$
\begin{equation*}
e_{n-1}\ \mbox{is the unique zero of the function}\ \lam\to \big\la\cR_\rho^{n-1}(H-\la H\ra_\Omega+\lambda)\big\ra_\Omega,
\end{equation*}
so that  $$H \in\ \quad \mbox{the domain of}\ \quad \cR_{\rho}^n.$$
This way one adjusts the initial conditions closer and closer to the stable manifold  $\mathcal{M}_s=\cap_n D(\cR_{\rho}^n)$.
\DETAILS{We have shown that $\cR_{\rho}$

\begin{itemize}
\item contracts, for $\mu>0$, in the co-dimension $2$: $\|\cR_{\rho}(H)-\cR_{\rho}(E+T)\|_\mu\le c\rho^\mu\|W\|_\mu$;

\item has fixed points $\z H_f,\ \z \in \C$:  $\cR_{\rho}(\z H_f)=\z H_f$; 

\item  expands along a subspace of  complex dimension $1$; $\cR_{\rho}(\z\id)=\rho^{-1}\z\id$.
\end{itemize}
To control the expanding direction, we adjust the parameter $\la H\ra_\Omega$ inductively,
$$|\la H\ra_\Omega - e_{n-1}| \leq \frac{1}{12}  \rho^{n+1},$$
\begin{equation*}
e_{n}\ \mbox{is the unique zero of the function}\ \big\la\cR_\rho^{n}(H-\la H\ra_\Omega+\lambda)\big\ra_\Omega,
\end{equation*}
so that  $H \in$ the domain of $\cR_{\rho}^n$.}
%
\begin{center}
\includegraphics[height=6.0cm]{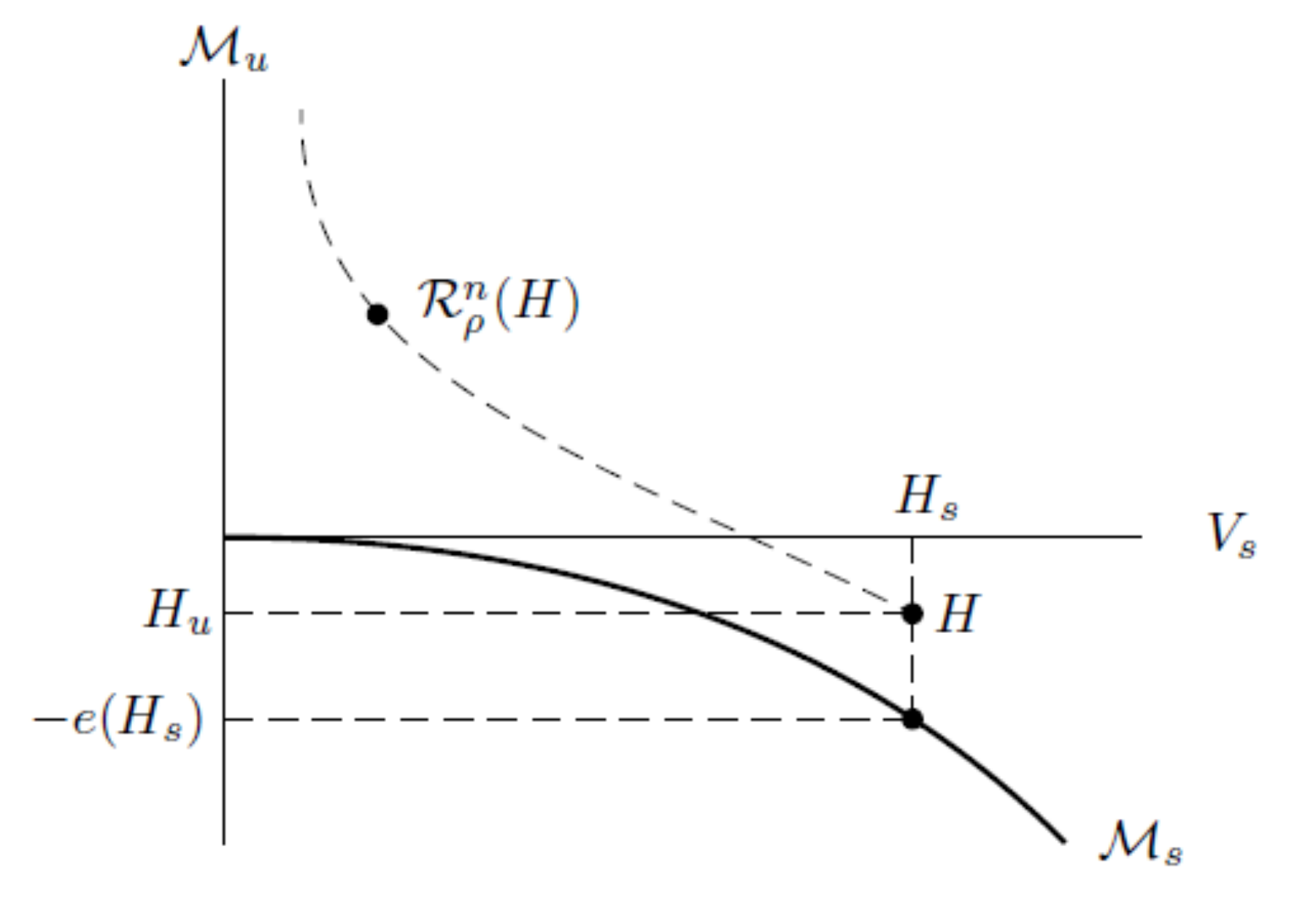}
\end{center}
%
%
\DETAILS{
\begin{center}
\psset{unit=1cm}
\begin{pspicture}(-1,-3)(8,3.5)

\psset{linewidth=0.5pt} \psline(0,-2)(0,3)
\psline(0,0)(6,0) 

\psbezier[linewidth=1pt](0,0)(2,0)(4,-0.5)(5.5,-2) 
\psbezier[linewidth=0.5pt,linestyle=dashed](4.5,-0.5)(2.5,0.5)(0.5,1)(0.5,2.8)
\qdisk(1,1.46){2pt}

\qdisk(4.5,-0.5){2pt} \qdisk(4.5,-1.19){2pt}
\psline[linewidth=0.5pt,linestyle=dashed](4.5,0)(4.5,-1.2) 
\psline[linewidth=0.5pt,linestyle=dashed](0,-0.5)(4.5,-0.5)
\psline[linewidth=0.5pt,linestyle=dashed](0,-1.19)(4.5,-1.19)
\rput(-0.7,-1.19){$-e(H_s)$}
\rput(1.8,1.5){$\mathcal{R}^{n}_{\rho}(H)$}

\rput(-0.3,-0.5){$H_u$} \rput(4.8,-0.5){$H$}\rput(4.6,0.25){$H_s$}

\rput(6,-2){$\mathcal{M}_s$} \rput(0,3.2){$\mathcal{M}_u$}

\rput(6.6,0){$V_s$} 
\end{pspicture}

$\mathcal{R}^{n}_{\rho}(H)$ converges to the unstable-fixed point manifold $\mathcal{M}_{ufp}:=\{\z_1\id+\z_2H_f\}$.

Here $V_s:=\{\sum_{m+n\ge 1}W_{mn}\},\ H_u:=\la H\ra_\Omega$ and $H_s:=H-\la H\ra_\Omega$.
\end{center}}
%
%

\bigskip

The procedure above leads to the following results:
\begin{itemize}
\item $\cR_{\rho}^n$  has the fixed-point manifold $\cM_{fp}:=\mathbb{C}H_f$;

\item $\cR_{\rho}^n$  has  an unstable manifold $\cM_{u}:=\mathbb{C}\one$;

\item  $\cR_{\rho}^n$  has a (complex) co-dimension $1$ stable
manifold $\cM_s$ for $\cM_{fp}$;

\item  $\cM_s$ is  foliated by (complex) co-dimension
$2$ stable manifolds for each fixed point.
\end{itemize}
%
%
\DETAILS{\begin{center}

\includegraphics[height=4.5cm]{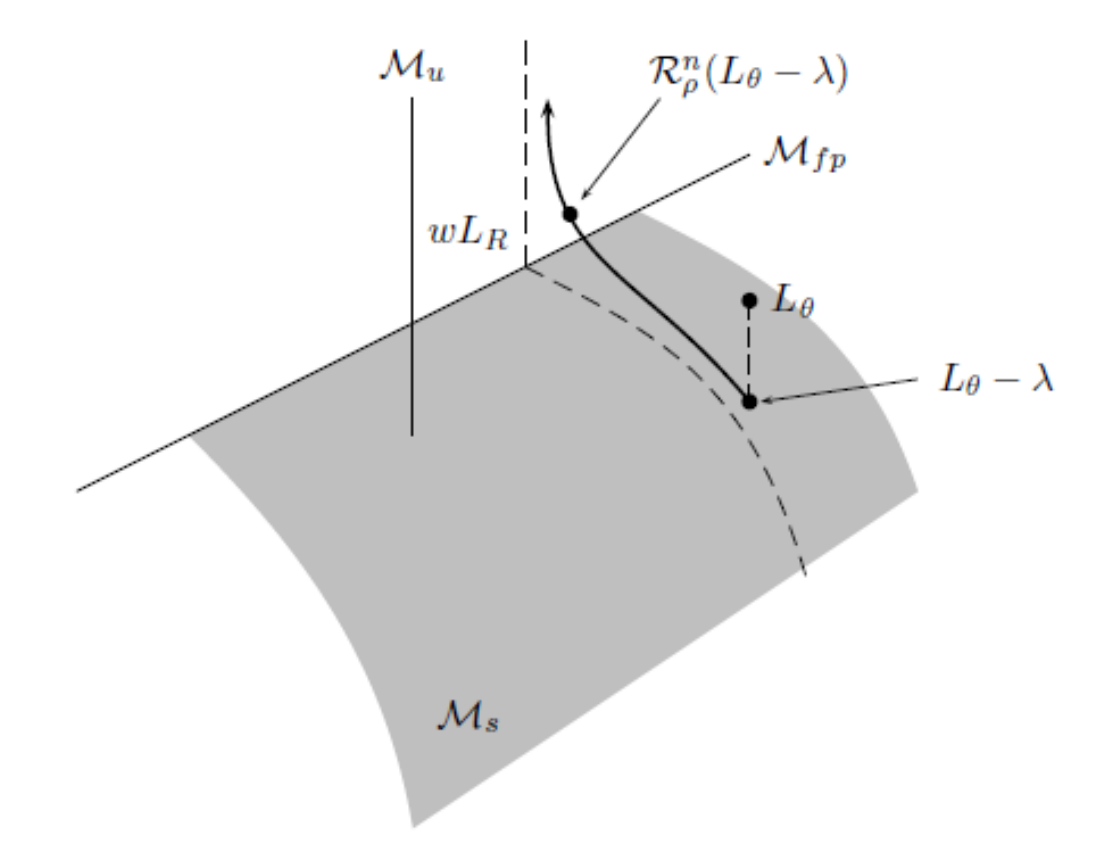}

Stable and unstable manifolds.
\end{center}}
\begin{center}
\includegraphics[height=6.0cm]{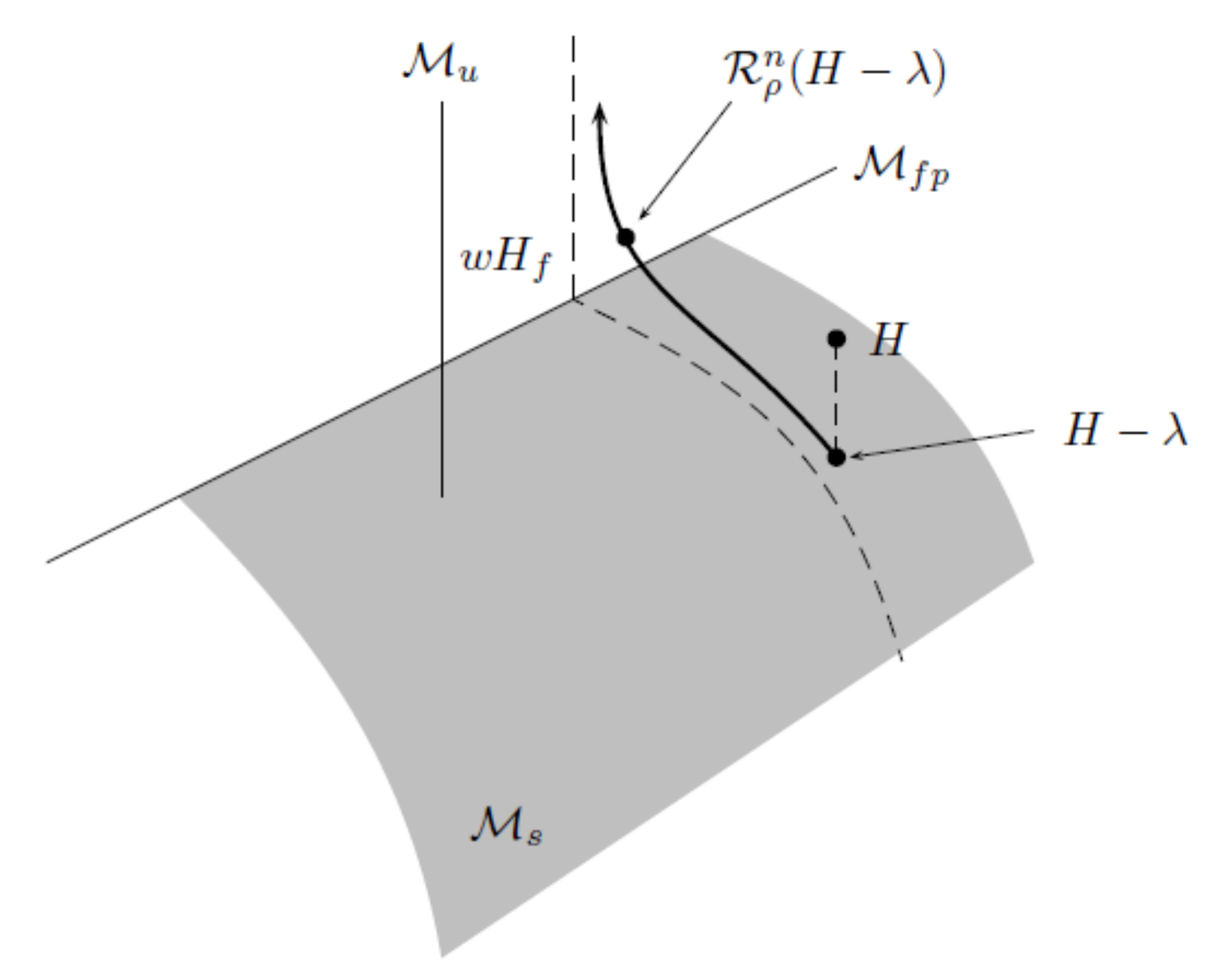}
\end{center}

%
%
\DETAILS{
\begin{center}
\psset{unit=1cm} \pspicture(-4,-5)(8,3.5)

\pscustom[linestyle=none,fillstyle=solid,fillcolor=lightgray]{
\psbezier(2,1)(3,0.5)(4,0)(4.5,-1.5)
\psbezier[liftpen=1](0,-4.5)(-0.25,-3)(-1,-2)(-2,-1)}
\psbezier[linewidth=0.5pt,linestyle=dashed](1,0.5)(2.2,-0.1)(3,-0.5)(3.5,-2.25)
\rput(0.5,-3.5){$\mathcal{M}_s$} \rput(0.5,0.8){$w H_f$}

\psset{linewidth=0.5pt} \psline(-3,-1.5)(3,1.5)
\rput(3.5,1.5){$\mathcal{M}_{fp}$}

\psline(0,-1)(0,2)\rput(0,2.3){$\mathcal{M}_u$}

\psline[linewidth=0.5pt,linestyle=dashed](1,0.5)(1,2.5)
\psbezier[linewidth=1pt]{->}(3,-0.7)(2,0.5)(1.2,0.5)(1.2,2)
\qdisk(1.4,0.97){2pt}\psline[linewidth=0.3pt]{<-}(1.5,1.1)(2.2,2)
\rput(3,2.2){$\mathcal{R}^{n}_{\rho}(H-\lambda)$}

\psline[linewidth=0.5pt,linestyle=dashed](3,0.2)(3,-0.7)
\qdisk(3,0.2){2pt}\qdisk(3,-0.7){2pt}\rput(3.4,0.2){$H$}
\psline[linewidth=0.3pt]{<-}(3.1,-0.7)(4.5,-0.5)\rput(5.2,-0.5){$H-\lambda$}
\endpspicture

Stable and unstable manifolds.
\end{center}}
%



\bigskip



Define the polydisc $\cD_s:=\cD^{\mu}(0, \beta_0, \gamma_0)$, with $\mu>0,\ \beta_0=c g^2\rho^{\mu_0 -1}$ and   $\gamma_0 =c g\rho^\mu_0$, the same as in Theorem \ref{thm:act-RGmap}. Using the above results, we draw the following conclusions about the spectrum of an operator $H\in \cD_s$. Find $\lam\in$ a neighbourhood of $0$, s.t. $H(\lambda):=H -\lam \in \cD(\cR_{\rho}^n)$ 
and define $H^{(n)}(\lambda):=\cR_{\rho}^n(H(\lambda) )$. We illustrate this as putting $H(\lambda)$ through the RG (black) box: 

$$H(\lambda) \Longrightarrow\ \textbf{RG BOX}\ \Longrightarrow\ H^{(n)}(\lambda).$$
Use that for $n$ sufficiently large, $H^{(n)}(\lambda)\approx \z H_f$, for some $\z \in \mathbb{C},\ \rRe\ \z >0$, to find wanted spectral information about $H^{(n)}(\lambda)$ in a neighbourhood of $0$. Then, use the isospectrality of $\cR_{\rho}$ to pass this information up the chain. 
These steps are described schematically as follows:
%


$\,$

 Spectral information about  $H^{(n)}(\lambda)$ $\Longrightarrow$ (by 'isospectrality' of  $\cR_{\rho}$)

 Spectral information about  $H^{(n-1)}(\lambda)$ 

$...$

%
%
$\Longrightarrow$ Spectral information about $H(\lambda)$.

\noindent We sum up this as
$$\textbf{Spec info}\ (H)\ \Longleftarrow\ \textbf{RG BOX}\ \Longleftarrow\ \textbf{Spec info}\ (H^{(n)}(\lambda)).$$
%
\DETAILS{Choose the parameter $\lambda$ inductively, so that  $H_{\theta}-\lambda \in D(\cR_{\rho}^n)$, 
i.e. $H_{\theta}-\lambda $ is in a $\rho^n-$neighborhood of the stable manifold $\cM_s$. Then

$H^{(n)}(\lambda):=\cR_{\rho}^n(H_{\theta}-\lambda )$ are close to the fixed point manifold $\cM_{fp}$

i.e. $ H^{(n)}(\lambda) \approx w H_f$, for some $\z\in \mathbb{C},\ \rRe\ \z >0$ and for  $n$ sufficiently large,

$\Longrightarrow$ Spectral information about $H^{(n)}(\lambda)$.
%

$\Longrightarrow$ Spectral information about  $H^{(n-1)}({\lambda })$ (by isospectrality of  $\cR_{\rho}$)

$...$}
%
%
%
If we start with the Hamiltonian $H_{\theta}$, then we are interested in, we derive this way the required spectral information about it.




\bigskip

\section{Related Results}\label{sec:relat-res}
 \textbf{Existence of the ionization threshold}.   There is $ \Sigma= \Sigma^{(p)}+O(g^2)>\inf\s(H)$ (the ionization threshold) s.t. for any energy interval in $\Delta\subset (\inf\s(H), \Sigma)$,
    $$\|e^{\delta|x|}\psi\|<\infty,\ \forall \psi\in\Ran E_\Delta(H),\ \delta<\Sigma-\sup\Delta.$$

\bigskip
\noindent\textbf{Analyticity.}
If the ground state energy $\e_g$ of $H$ is non-degenerate, then it is analytic in $g$ and has the following expansion
 \begin{equation}\label{GSE-exp}
\e_g=\sum\limits_{j=1}^\infty \e^{2j}_\al \al^{3j},
\end{equation}
where $\e^{2j}_\al$ are smooth function of $\al$ (recall that $g:= \alpha^{3/2}$). 

The proof relies on the renormalization group analysis together with the analyticity and the form-factor rotation symmetry transfer by the Feshbach-Schur map (see Appendix \ref{sec:anal-transf}) and on treating two sources of the dependence of $H$ on the coupling constant $g$ - in the prefactor of $\Af(y)$ and in its argument 
(see the line after \eqref{Hsm'}), differently, i.e. by considering the following family of  Hamiltonians 
  \begin{equation}\label{Hsm''}
\hsm:=\sum\limits_{j=1}^n{1\over 2m_j}
(i\nabla_{x_j}-gA_{\chi}(\al x_j))^2+ V(x)+H_f,
\end{equation}
and consider $g$ and $\al$ as independent variables. The powers of $g$ in \eqref{GSE-exp} come from $g$ in \eqref{Hsm''}.

\bigskip


\noindent\textbf{Analyticity of all parts of $H$.} 
Suppose that $\lambda\mapsto
H_\lam\equiv H(\underline{w}^{\lambda})$
is of the form \eqref{H} and is analytic in $\lambda\in S\subset\CC$ and that $H(\underline{w}^{\lambda})$ belongs to some polydisc
$\cD(\alpha,\beta,\gamma)$ for all $\lambda\in S$. Then
$\lambda\mapsto E_\lam:=w_{0,0}^{\lambda}(0),\  T_\lam :=w_{0,0}^{\lambda}(H_f)-w_{0,0}^{\lambda}(0),\  W_\lam:=H_{\lambda}-E_\lam-T_\lam $ 
are analytic in $\lambda\in S$. 

\bigskip

\noindent\textbf{Resonance poles.} Can we make sense of the resonance poles in the present context? 
Let $$ Q:= \{z \in \mathbb{C}^- |\ \epsilon_{0} < \rRe z < \nu \} / \bigcup_{j \le j(\nu), k} S_{j, k} .$$
\begin{theorem} {\em For each $\Psi$ and $\Phi$ from a dense set of
vectors (say, \eqref{D}), the meromorphic continuation, $F(z, \Psi, \Phi)$, of the
matrix element $\langle \Psi, (H-z)^{-1}\Phi\rangle$ is of the following form near the resonance
$\epsilon_j$ of $H$:
%
%
%
%
%
%
%
%
%
\begin{equation} \label{poles}
F(z, \Psi, \Phi)=(\epsilon_{j, k} -z)^{-1} p(\Psi, \Phi) + r(z, \Psi,
\Phi) .
\end{equation}
 Here $p$ and $r(z)$ are sesquilinear forms in $\Psi$ and $\Phi$, s.t.
 \begin{itemize}
\item  $r(z)$ is analytic in $Q$
and bounded on the intersection of a neighbourhood of
$\epsilon_{j, k}$ with $Q$ as $$|r(z, \Psi, \Phi)| \le C_{\Psi,
\Phi}|\epsilon_{j, k}-z|^{-\gamma}\ \mbox{for some}\ \gamma <1;$$
\item  $p \ne 0$ at least for one pair of vectors $\Psi$ and
$\Phi$ and $p = 0$
for a dense set of vectors $\Psi$ and $\Phi$ in a finite
co-dimension subspace.
 \end{itemize}}
\end{theorem}

\noindent\textbf{Local decay.}  For any compactly supported function $f(\lambda)$ with $\supp
f \subseteq (\inf H, \infty)/$(a neighbourhood of $\Sigma$), and for $\theta> \frac{1}{2},\ \nu< \theta-\frac{1}{2}$, we have that
\begin{equation}\label{loc-dec}
\|\la Y\ra^{-\theta}e^{-i H t}f(H)\la Y\ra^{-\theta}\|\le C
t^{-\nu}.
\end{equation}
Here
$\la Y\ra := (\one +Y^2)^{1/2}$, and $Y$ denotes the self-adjoint operator on Fock space $\cF$ of photon co-ordinate,
\begin{equation} \label{Y}
Y\ : = \ \; \int d^3k \; a^*(k) \:  i\nabla_k  \: a(k),
\end{equation}
 extended to the Hilbert space $\cH=\cH_{\at}\otimes\cF$.
(A self-adjoint operator $H$ obeying \eqref{loc-dec} is said to have the ($\Delta, \nu, Y, \theta$) - \textit{local decay} (LD) property.) \eqref{loc-dec} shows that for well-prepared initial conditions $\psi_0$, the probability of
finding photons within a ball of an arbitrary radius $R < {\infty} $,
centered say at the center-of-mass of the particle system, tends to 0, as time $t$
tends to ${\infty}$.

 Note that one can also show that as $t\to \infty$,  the photon coordinate  and wave vectors in the support of  the solution $e^{-iHt}\psi_0$ of the Schr\"odinger equation become more and more parallel. This follows from the local decay for the self-adjoint generator of dilatations on Fock space  $\cF$,
 \begin{equation} \label{B}
B \ : = \ \frac{i}{2}\; \int d^3k \; a^*(k) \: \big\{ k \cdot
\nabla_k + \nabla_k \cdot k \big\} \: a(k) .
\end{equation}
(In fact, one first proves the local decay property for $B$ and then transfers it to the photon co-ordinate operator $Y$.) 

\bigskip

\section{Conclusion}

Apart from the vacuum polarization, the non-relativistic QED provides a good qualitative description of the physical phenomena related to the interaction of quantized electrons and nuclei and the electro-magnetic field. (Though construction of the scattering theory is not yet completed and the correction to the gyromagnetic ratio is not established, it is fair to conjecture that while both are difficult problems, they should go through without a hitch.)

The quantitative results though are still missing. Does the free parameter, $m$ (or  $\kappa$), suffice to give a good fit with the experimental data say on the radiative corrections? Another important open question is the behaviour of the theory in the ultra-violet cut off. 

\DETAILS{
\section{Open Problems}

Compute With $m=m(m_\el, \kappa)$ in the quantum Hamiltonian
$$$$
Minimal and maximal velocity of photons;
$$$$
Asymptotic completeness; 
$$$$
Bohr photon frequency laws. }
%
\section{Comments on Literature}\label{sec:liter}
%

These lectures follow the papers \cite{sig, fgs2, AFFS}, which in turn extend  \cite{bfs1, bfs2, bcfs1}.
The papers \cite{sig, fgs2,  AFFS} use the smooth Feshbah-Schur map (\cite{bcfs1, GriesemerHasler1}), which is much more powerful (see Appendix \ref{sec:sfm}), while in these lectures we use, for simplicity, the the original, Feshbach-Schur map (see \cite{bfs1, bfs2}), which is simpler to formulate.

The self-adjointness of $H$ is not difficult and was proven in \cite{ bfs3} for sufficiently small coupling constant ($g$) and in  \cite{Hiroshima4} (see also \cite{Hiroshima2a}), for an arbitrary one.

Theorems 4.1 and 4.2 were proven in \cite{bfs1, bfs2, bfs3} for 'confined particles' (the exact conditions are somewhat technical) and in the present form in \cite{sig}.


The results of \cite{ bfs3} on existence 
of the ground state were considerably improved in  \cite{  Hiroshima1,
Hiroshima2, Hiroshima3, HiroshimaSpohn,  AraiHirokawa, LMS2}  (by compactness techniques) and
\cite{BachFroehlichPizzo1} (by multiscale techniques), with the sharpest result given in \cite{GriesemerLiebLoss,   LiebLoss3}. 
(The papers
\cite{bfs3, Hiroshima1, GriesemerLiebLoss,
LiebLoss1, LiebLoss2, LiebLoss3, HiroshimaSpohn, LMS2} include the interaction of the spin with magnetic field 
in the Hamiltonian.)



Related results:

\begin{enumerate}
\item[\nonumber] The asymptotic stability of the ground state (local decay, see Section \ref{sec:relat-res}):
\cite{bfss, fgs1, fgs3}.

\item[\nonumber] The survival probabilities of excited states (see \eqref{ResonDecay}): 
 \cite{bfs3, HaslerHerbstHuber, AFFS}.

\item[\nonumber] Atoms with dynamic nuclei: 
 \cite{faupin, AGG, LMS}.

\item[\nonumber] Analyticity of the ground state eigenvalues in parameters and asymptotic expansions (see Section \ref{sec:relat-res}): \cite{bfs3, BachFroehlichPizzo1, BachFroehlichPizzo3, GriesemerHasler2, HaslerHerbst1, HaslerHerbst2,HaslerHerbst3}. 


\item[\nonumber]  Existence of the ionization threshold (see Section \ref{sec:relat-res}): \cite{Griesemer}.

\item[\nonumber]  Resonance poles (see Section \ref{sec:relat-res}): \cite{
AFFS} (see also \cite{bfs3}).

\item[\nonumber] Self-energy and binding energy: \cite{HiroshimaSpohn1, LiebLoss1, LiebLoss2, LiebLoss3, H1, bcv, cvv, HVV,  H2, HS, HHS, CEH, bv, bcvv}.


\item[\nonumber] Electron mass renormalization: 
\cite{HS,  HiroshimaSpohn2, bcfs2, ch, FroehlichPizzo}. 

\item[\nonumber] One particle states: \cite{fr1, fr2, cfp1, cfp2, FroehlichPizzo}.

\item[\nonumber] Scattering amplitudes: \cite{BachFroehlichPizzo2}. 

\item[\nonumber]  Semi-relativistic Hamiltonians: \cite{BDG, MiyaoSpohn, MatteStockmeyer, KMS, Stockmeyer, HiroshimaSasaki}.

\end{enumerate}


Other aspects of non-relativistic QED:
\begin{enumerate}

\item[\nonumber]  Photo-electric effect:  \cite{bkz, GriesemerZenk}.


\item[\nonumber] Scattering theory:
\cite{FroehlichGriesemerSchlein1,FroehlichGriesemerSchlein2, FroehlichGriesemerSchlein3}.

\item[\nonumber] Stability of matter:
\cite{BugliaroFroehlichGraff, FeffermanFroehlichGraff, LiebLoss1}.

\end{enumerate} 

\bigskip
\noindent There is an extensive literature on related models, which we do not mention here: Nelson model describing a particle linearly coupled to a free massless scalar field (phonons), semi-relativistic models, based on Dirac equation, and quantum statistics (open systems, positive temperature) models. (In the latter case, 
one deals with Liouvillians, rather than Hamiltonians, on positive temperature Hilbert spaces. 
The main results above were proved simultaneously for the QED and Nelson models and extended, at least partially, to positive temperatures.) 

\appendix

\section{Hamiltonian of the Standard Model} 
\label{sec:pot}

In this appendix we demonstrate the origin the quantum Hamiltonian $H_g$ given in \eqref{Hsm}.  To be specific consider 
an atom or molecule with $n$ electrons interacting with radiation field. 
In this case the Hamiltonian of the system in our units is given by
 \begin{equation}\label{Hsm'}
H(\al)=\sum\limits_{j=1}^n{1\over 2m}
(i\nabla_{x_j}-\sqrt{\alpha}A_{\chi'}(x_j))^2+\al V(x)+H_f,
\end{equation}
%
%
where $\al V(x)$ is the total Coulomb potential of the particle system, $m$ is the electron bare mass,  $\alpha =\frac{e^2}{4\pi \hbar c}\approx {1\over 137}$ (the
fine-structure constant) 
and $A_{\chi'}(y)$ is the original vector potential with the ultraviolet cut-off $\chi'$.
Rescaling $x \rightarrow \alpha^{-1} x$ and $k \rightarrow
\alpha^2 k$, we arrive at the Hamiltonian \eqref{Hsm},
where $g:= \alpha^{3/2}$ 
and $A(y) = A_{\chi}(\alpha y)$, 
with $\chi( k):=\chi'(\alpha^2 k)$. 
%
After that 
we relax the restriction on $V(x)$
by allowing it to be a standard generalized $n$-body potential (see Subsection \ref{sec:model}). Note that though this is not displayed, $A(x)$ does depend on $g$. This however does not effect the analysis of the Hamiltonian $\hsm$. (If anything, this makes certain parts of it simpler, as derivatives of $A(x)$ bring down $g$.)

$$****$$
In order not to deal with the problem of center-of-mass motion
which is not essential in the present context, we assume that either
some of the particles (nuclei) are infinitely heavy (a molecule in the
Born-Oppenheimer approximation), or the system is
placed in a binding, external potential field. (In the case of the
Born-Oppenheimer molecule, the resulting $V(x)$ also depends on
the rescaled coordinates of the nuclei, but this does not effect our analysis except of making the complex deformation of the particle system more complicated (see \cite{HunzikerSigal}).) This means that the operator $H_p$ has isolated eigenvalues below its essential
spectrum.  The general case is considered below. 

In order to take into account the particle spin we change the state space the particle system to  $\cH_\at=\otimes_1^nL^2(\R^{3}, \C^2)$ (or the antisymmetric in identical particles  subspace thereof),
and  the standard quantum Hamiltonian on $\cH=\cH_{p}\otimes\cH_{f}$, is taken to be (see e.g. \cite{Cohen-TannoudjiDupont-RocGrynberg1, Cohen-TannoudjiDupont-RocGrynberg2, sakurai})
\begin{equation} \label{Hsm-spin}
H_{spin}=\sum\limits_{j=1}^n{1\over 2m}
[\sigma_j\cdot(i\nabla_{x_j}-g\Af(x_j))]^2+V(x)+H_f,
\end{equation}
where $\s_j:=(\s_{j1}, \s_{j2}, \s_{j3}),\ \s_{ji}$ are the Pauli matrices of the $j-$th particle and the identity operator on $\C^{2n}$ is omitted in the last two terms. It is easy to show that
\begin{equation} \label{spin}
[\sigma\cdot(i\nabla_{x}-g\Af(x))]^2=(i\nabla_{x}-g\Af(x))^2+ g\sigma \cdot B(x).
\end{equation}
where $B(x):=\textrm{curl} \Af(x)$ is the magnetic field. As a result the operator  \eqref{Hsm-spin} can be rewritten as
\begin{equation} \label{Hsm-spin}
H_{spin}=H\otimes \one+\one\otimes g\sum\limits_{j=1}^n{1\over 2m}
\sigma_j\cdot B(x_j).
\end{equation}
For the semi-relativistic Hamiltonian, the non-relativistic kinetic energy $\frac{1}{2m}|p|^2$ is replaced by the relativistic one, $\sqrt{|p|^2+m^2}$ or $\sqrt{(\sigma\cdot p)^2+m^2}$.

\bigskip

\section{Translationally Invariant Hamiltonians} \label{sec:transl-Ham}

If we do not assume that the nuclei are infinitely and there are no external forces acting on the system, then the Hamiltonian \eqref{Hsm} is translationally symmetric. This 
leads to conservation 
of the total momentum (a quantum version of the classical Noether theorem). Indeed,
 \index{translation ! invariance}
the system of particles interacting with the quantized electromagnetic fields 
is invariant under translations of the particle coordinates, $\ux\to  \ux+\uy$, where $\uy=(y, \dots, y)$ ($n-$ tuple) and the fields, $A(x)\to A(x-y)$, i.e. $\hsm$ commutes with the translations
\DETAILS{\begin{equation}\label{Ty}
	T_{y}: \Psi(\ux, A)\to   \Psi(\ux+\uy, t_{y}A), 
	\end{equation}
where 
$(t_{y}A)(x, A)=A(x-y)$.  (We say that $\hsm$ is \textit{translation invariant}.) Indeed,
we use \eqref{Ty} and the definitions of the operators $A(x)$ and $E(x)$, to obtain  $T_y A(x) = (t_{y}A)(x) T_y$ and $T_y E(x) = (t_{y}E)(x) T_y$, which, due to the definition of  $H_{e\chi}$, give   
\begin{equation*} T_y \hsm = \hsm T_y, \end{equation*}

Note that in the Fock space representation this gives: $T_{y}:   	 \oplus_n\Psi_n(\ux, k_1, \dots, k_n)$\\ $\to \oplus_n e^{ iy\cdot (k_1+\dots k_n)}	 \Psi_n(\ux+\uy, k_1, \dots, k_n)$ and therefore can be rewritten as}
 $T_{y}: \Psi(\ux)\to  e^{ iy\cdot P_\re}	 \Psi(\ux+\uy)$, where $ P_\re$ is the momentum
operator associated to the quantized radiation field,
\[\Pf=\sum_\lambda\int dk \, k \, a_\lambda^*(k)a_\lambda(k).\] 
It is straightforward to show that $T_y$ are unitary operators and that they satisfy the relations
$ T_{x+y} = T_xT_y,$
and therefore $y\to T_{y}$ is a unitary Abelian representation of $\R^3$.
Finally, we observe that the group $T_y$ is generated  by the total momentum operator, $\Ptot$, of the electrons and the photon field: $T_y= e^{ iy\cdot \Ptot}$. Here $\Ptot$ is the selfadjoint operator on $\cH$, given by 
\eqn
	\Ptot \, := \, 
\sum_ip_i\otimes\1_f \, + \, \1_{el}\otimes P_f
\eeqn
where, as above, $p_j:=-i\nabla_{x_j}$, the momentum of the $j-$th electron and $ P_f$ is the field momentum given above.  
Hence $[\hsm, \Ptot]=0$.

\medskip

Let $\cH$ be the direct integral
	$\cH = \int_{\R^3}^\oplus \cH_P dP, $ 
with the fibers $\cH_P := L^2(X)\otimes \cF$, where $X:=\{x\in {\mathbb R}^{3n}\ |\ \sum_i m_ix_i=0\}\simeq{\mathbb R}^{3(n-1)}$, (this means that $\cH = L^2(\R^3, dP; L^2(X)\otimes \cF)$) and 
		define $U : \cH_\el \otimes \cH_{\re} \to \cH$ on smooth functions with compact domain by the formula
	\begin{equation}\label{U}
		(U \Psi)(\ux', P) = \int_{\mathbb{R}^3}e^{i(P-P_\re)\cdot x_{cm}}\Psi(\ux'+\ux_{cm}) \d y, 
	\end{equation}
where $\ux'$ are the coordinates of the $N$ particles in the center-of-mass frame and $\ux_{cm}=(x_{cm}, \dots, x_{cm})$ ($n-$ tuple), with $ x_{cm}=\frac{1}{\sum_im_i}\sum_im_ix_i$, the center-of-mass coordinate, so that $\ux=\ux'+\ux_{cm}$.  Then $U$ extends uniquely to a unitary operator (see below). Its converse is written, for  $\Phi (\ux', P)\  \in L^2(X)\otimes \mathcal{F}$, as
 \begin{equation}\label{Uinv} 
 (U^{-1}\Phi)(\ux)=\int_{\mathbb{R}^3} e^{-\i x_{cm}\cdot(P-P_\re)}\Phi (\ux', P) \d P.
\end{equation}
The functions $\Phi (\ux', P)\ =(U \Psi)(\ux', P)$ 
are called fibers of $\Psi$. One can easily prove the following
\DETAILS{ For $k \in \R^3$, let $H_{ k}$ be the operator $H$ acting on $\cH_k$ with domain consisting of those $v \in \cH_k \cap H^2$ such that $\Psi$
	satisfies the boundary conditions $T_y \Psi(x) = \chi_k(t) \Psi(x)$. Then
	\begin{equation}
	\label{K-decomposition}
		UH U^{-1} = \int_{\R^3}^\oplus H_{ k} dk.
	\end{equation}
	and
	\begin{equation}
	\label{KKkspecrelat}
		\sigma(H) = \bigcup_{k\in\R^3} \sigma(H_{ k}).
	\end{equation}}
\DETAILS{\bigskip

We consider  the direct integral decomposition
\eqn
	\cH \, = \, \int^\oplus dp \, \cH_p \,,
\eeqn
with respect to $\Ptot$, where each fiber $\cH_p\cong {\rm Ker}(\Ptot-p)$ is isomorphic to $\Fo$.
It corresponds to the natural isomorphism
\eqn
	\cH \, \cong \, L^2(\R^3,\Fo) \,
\eeqn
which we will use in the sequel.


For $\psi\in\cH$, we define the Fourier transform with respect to the electron variable,
\eqn
	\widehat\psi(p) \, = \, \int dx \, e^{-i(p-\Pf)x} \psi(x)
\eeqn
with inverse transform
\eqn
	\check\phi
 \, = \, \int dp \, e^{i(p-\Pf)x}\phi(p) \,.
\eeqn}
\begin{lemma}
 {\em The  operations \eqref{U} 
and \eqref{Uinv} 
define unitary maps $L^2({\mathbb R}^{3n})\otimes \cF\rightarrow
 \cH$ and  $\cH\rightarrow L^2({\mathbb R}^{3n})\otimes \cF$, and are mutual inverses.}
\end{lemma}
\DETAILS{\prf  
By density, we may assume that $\Psi$ is a $C_0^\infty$ in $\ux$. 
Then, it follows from standard arguments in Fourier analysis that
\eqn\label{eq-FT-def-1}
(U^{-1}U\Psi)(\ux) &=&
	\int \d P \, e^{-i(P-\Pf)x_{cm}} \int dy \, e^{i(P-\Pf)\cdot y} \, \Psi(\ux' +y)
	\nonumber\\
	&=&
	\int dy \,  \int \d P \,   e^{-iP\cdot(x_{cm}-y)}e^{i\Pf\cdot(x_{cm}-y)} \, \Psi(\ux'+y)
	\nonumber\\
	&=&
	  \int dy \,  \delta(x_{cm}-y) \, e^{i\Pf\cdot(x_{cm}-y)} \, \Psi(\ux'+y)
	 \nonumber\\
	 &=&\Psi(\ux) \,.
\eeqn
On the other hand, for $\psi\in {\mathcal S}_{\Fo} $,
\eqn\label{eq-FTinv-def-1}
	(UU^{-1}\Phi)(\ux', P)
	&=&
	\int dy \, e^{i(P-\Pf)y} \int dq \, e^{i(q-\Pf)y} \, \Phi(\ux', q)
	\nonumber\\
	&=&
	 \int dq \, \int dy \, e^{i(p-q)y} \,  \Phi(\ux', q)
	 \nonumber\\
	 &=&
	 \int dq \, \delta(P-q) \, \Phi(\ux', q)
	 \nonumber\\
	 &=&\Phi(\ux', P) \,.
\eeqn
From the density of 
 $C_0^\infty$ in $\ux$ functions, we infer that \eqref{eq-FT-def-1}
and \eqref{eq-FTinv-def-1}  define bounded maps 
which are mutual inverses. Unitarity can be checked easily.
\endprf}
\DETAILS{We observe that
\eqn
	(\Ptot\psi)^{\widehat{\;}}(p) & = & \int dx \, e^{i(p-\Pf)x} (-i\nabla_x+\Pf)\psi(x)
	\nonumber\\
	&=& \int dx \, \big( \, (-i\nabla_x+\Pf)e^{i(p-\Pf)x}\,\big) \, \psi(x)
	\nonumber\\
	&=&p \, \widehat\psi(p) \,.
\eeqn
Hence, $\widehat\psi(p)\in\cH_p\cong {\rm Ker}(\Ptot-p)$, for any $p\in\R^3$.}
%
Since $\hsm$ commutes with $\Ptot$, it follows that it admits the fiber decomposition
\begin{equation}
U \hsm U^{-1} = \int_{\mathbb{R}^3}^{\oplus} \hsm(P) \d P,
\end{equation}
where the fiber operators $\Hn(P)$, $P \in \mathbb{R}^3$, are self-adjoint operators on $\mathcal{F}$. Using $a(k)e^{-\i y\cdot P_\re}=e^{-\i y\cdot (P_\re+k)}a(k)$ and  $a^*(k)e^{-i y\cdot P_\re}=e^{-i y\cdot (P_\re-k)}a^*(k)$, we find $\nabla_y  e^{i y\cdot(P-P_\re)}\Af(x'+y) e^{\i y\cdot(P-P_\re)} =0$ and therefore
\begin{equation}
\Af(x) e^{i y\cdot(P-P_\re)} = e^{i y\cdot(P-P_\re)}\Af(x-y).
\end{equation}
Using this and \eqref{Uinv},  we compute $\hsm  (U^{-1}\Phi)(x)=\int_{\mathbb{R}^3} e^{i x\cdot(P-P_\re)}\hsm(P)\Phi (P) d P$, where  $\hsm(P)$ are Hamiltonians on the space fibers $\cH_P := \cF$ given explicitly by
\begin{equation}
\hsm(P) =  \sum_j\frac{1}{2m_i}\big ( P - P_\re -i\nabla_{x_j'} -e_i\Af(x_j')\big)^2 + V_{\rm coul}(\ux')+H_\re
\end{equation}
where $x_i'=x_i - x_{cm}$ is the coordinate of the  $i-$th particle in the center-of-mass frame. Now, this hamiltonian can be investigated similarly to the one in \eqref{Hsm}. 
%
%
\DETAILS{For an observable $A$ on $ \mathcal{H}_{\mathrm{el}} \otimes \mathcal{H}_\re$, we consider the operator $UA U^{-1}$, acting on the space $\cH = \int_{\R^3}^\oplus \cH_P dP$. It can be written as
$$(UA U^{-1}\Phi)(P)=\int_{\mathbb{R}^3} A_{PQ}\Phi(Q)dQ,$$
for some operator-function $A_{PQ}$ on $\cF$. If $A\in \cA$, then $A_{PQ}\in {\frak W}(L^{2}_0)$.
We write $UA U^{-1}=\int A_{PQ} dPdQ.$ Note that the evolution, $\al^t(A)$, of $A$ is then given by
\begin{equation}
\al^t(A)= \int_{\mathbb{R}^3\times \R^3}dPdQ\alpha_{PQ}^t (A_{PQ}),
\end{equation}
where $\alpha_{PQ}^t (A_{PQ}):=e^{itH_P}A_{PQ} e^{-itH_Q}$. If an observable $A$ is translation invariant, then we can write it as a fiber integral, $A = \int_{\mathbb{R}^3}^{\oplus} A_P \d P.$ If $A$ is translation invariant, then so is the observable $\al^t(A)$ and
\begin{equation}
\al^t(A)=\int_{\mathbb{R}^3} \al^t_{P}(A_{P})dP,\ \mbox{where}\ \al^t_{P}(A_{P}):=e^{itH_P}A_{P} e^{-itH_P}.
\end{equation}
We write $\al^t=\int \al^t_{P}dP$.  An example of  a  translation invariant observable is the particle momentum  $P_\el=-i\nabla=i[H, x]$. In this case we have $P_\el = \int_{\mathbb{R}^3}^{\oplus} (P-P_\re) \d P = \int_{\mathbb{R}^3}^{\oplus} P \d P-P_\re.$}
\DETAILS{\bigskip
Next, we consider the Hamiltonian \eqref{eq-Hn-def-1}.
Clearly, $[\Hn,\Ptot]=0$, by translation invariance.
Then,
\eqn
	(\Hn \psi)^{\widehat{\;}}(p) \, = \, \Hn(p) \widehat\psi(p)
\eeqn
where $\Hn(p) =  \Hn\big|_{\cH_p}$ is the fiber Hamiltonian associated to total momentum $p$. Let $\Af:=\Af(0)  $. We compute
\eqn
	\Hn(p) =\frac12( \, p \,  - \,  \Pf \, - \, \c \Af  \, )^2
	\, + \,  H_f \,
\eeqn
We then have that
\eqn
	\Hn \, = \, \int e^{i(p-\Pf)x}\Hn(p)e^{-i(p-\Pf)x}dP_{\Ptot}(p)
\eeqn
which is the spectral representation of $\Hn$
with respect to the spectral measure associated to $\Ptot$.}
%

\bigskip


\section{Proof of Theorem \ref{thm:isospF}} \label{sec:isosprel}
In this appendix we \textit{omit the subindex} $\rho$ at $\chi_{\rho}$ and $\bchi_{\rho}$, and \textit{replace the subindex} $\rho$ in other operators by the subindex $\chi$. Moreover, we replace $H-\lam$ by $H$. Though $\chi$ and $\bchi$ we deal with are projections, we often keep the powers $\chi^2$ and $\bchi^2$, which occur often below, having in mind showing possible generalization to  $\chi$ and $\bchi$ which are 'almost (or smooth) projections' satisfying $\chi^2+\bchi^2=1$ (see Appendix \ref{sec:sfm}).

First we note that the relation between $\psi$ and $\varphi$ in Theorem \ref{thm:isospF} (ii) is $\vphi=\chi\psi,\ \psi=Q_{\chi}( H)\vphi,$
and between $H^{-1}$ and $F_{\chi}(H)^{-1}$ in (iv) is
\begin{equation} \label{resolvrel}
H^{-1}  =  Q_{\chi} ( H) \: F_{\chi}
H^{-1} \: Q_{\chi} ( H)^\# + \; \bchi \, H_{\bchi}^{-1} \bchi ,
\end{equation}
where  $H_{\bchi}:=\bchi_{\rho}H\bchi_{\chi}$ and $Q_{\chi} ( H)$ and $Q_{\chi} ( H)^\#$ are the operators, given by

\medskip
$Q_{\chi} (H)  :=  \chi \: - \: \bchi \, H_{\bchi}^{-1}
\bchi H \chi ,$

\medskip
$Q_{\chi} ^\#(H)  :=  \chi \: - \: \chi H \bchi \,
H_{\bchi}^{-1} \bchi .$

%
\DETAILS{\item[(iii)] If $\vphi
\in \Ran\, \pi \setminus \{0\}$ solves $F_{\rho} (H - \lambda)
\, \vphi = 0$ then $\psi := Q_{\pi} (H - \lambda) \vphi \in \cH
\setminus \{0\}$ solves $H\psi = \lambda \psi$;}
%

 \emph{Proof of Theorem \ref{thm:isospF}. }
Throughout the proof we use the notation $F := F_\chi(H)$,
$Q := Q_\chi(H)$, and $Q^\# := Q_\chi^\#(H)$. Note that (i) ($ 0 \in \rho(H ) \Leftrightarrow 0 \in \rho(F_{\chi}(H))$) follows from (iv) ($H^{-1}$ exists $ \Leftrightarrow$ $F_{\chi}(H)^{-1}$
exists) and (iv) follows from \eqref{resolvrel}, so we start with the latter.

\noindent
\textbf{Proof of \eqref{resolvrel}.}
The next two identities,
\begin{equation} \label{eq-II-11}
H \, Q     \ = \ \chi \, F
\hspace{8mm} \mbox{and} \hspace{8mm}
Q^\# \,  H \ = \ F \, \chi ,
\end{equation}
are of key importance in the proof. They both derive from a simple
computation, which we give only for the first equality in (\ref{eq-II-11}).  We observe the relations
\begin{equation} \label{Hchi} 
H \, \chi \ = \ \chi \, H_\chi \; + \; \bchi^2 \, H \chi ,
\hspace{8mm} \mbox{and} \hspace{8mm}
H \, \bchi \ = \ \bchi \, H_\bchi \; + \; \chi^2 \, H \bchi ,
\end{equation}
which follow from $\chi^2 + \bchi^2 = \one$. Now, using the definition of the operator $Q$ and the relations \eqref{Hchi}, we obtain
\begin{eqnarray} \label{eq-II-12}
H \, Q
& = &
\chi H_\chi \, + \, \bchi^2 H \chi \, - \,
\big( \bchi H_\bchi + \chi^2 H \bchi \big)
\, H_\bchi^{-1} \bchi H \chi.
\end{eqnarray}
Canceling the second term on the r.h.s. with the first term in the parentheses in the third term, we see that the r.h.s is equal $\chi \, F$, which gives  the first equality in (\ref{eq-II-11}).
%

Now, suppose first that the operator $F$ has bounded invertible and define
\begin{equation} \label{eq-II-13}
R \ := \ Q \: F^{-1} \: Q^\#
\; + \; \bchi \, H_\bchi^{-1} \bchi .
\end{equation}
Using (\ref{eq-II-11}) and \eqref{Hchi}, we obtain
\begin{eqnarray} \label{eq-II-14}
H \, R
& = &
H \, Q \, F^{-1} \, Q^\#
\; + \; \big( \bchi H_\bchi + \chi^2 H \bchi \big)
\, H_\bchi^{-1} \bchi
\\ \nonumber
& = &
\chi \, Q^\# \; + \; \bchi^2
\; + \; \chi^2 H \bchi \, H_\bchi^{-1} \bchi
\\ \nonumber
& = &
\chi^2 \; + \; \bchi^2 \ = \ \one ,
\end{eqnarray}
and, similarly, $R H = \one$. Thus $R = H^{-1}$,
and \eqref{resolvrel} holds true.

Conversely, suppose that $H$ is bounded invertible. 
Then, using the definition of $F$ and the relation $\chi^2 + \bchi^2=\one$, we obtain
\begin{eqnarray} \label{eq-II-16}
F \, \chi \, H^{-1} \, \chi 
& = &
\chi H \, \chi^2 \; H^{-1} \, \chi 
- \; \chi H \bchi \, H_\bchi^{-1} \bchi H \chi^2  \, H^{-1} \chi
\\ \nonumber
& = &
\chi H \, \chi^2 \; H^{-1} \, \chi -  \; \chi H \bchi \, H_\bchi^{-1} \bchi H   \, H^{-1} \chi + \; \chi H \bchi \, H_\bchi^{-1} \bchi H \bchi^2  \, H^{-1} \chi
\\ \nonumber
& = &
\chi H \, \chi^2 \; H^{-1} \, \chi + \; \chi H \bchi^2  \, H^{-1} \chi=\chi^2 .
\end{eqnarray}
Similarly, one checks that $\chi \, H^{-1} \, \chi F = \one$. Thus $F$ is invertible on $\Ran\, \chi$ 
with inverse
$F^{-1} = \chi \, H^{-1} \, \chi$. 

\noindent
\textbf{Proof of (ii) ($H\psi = \lambda \psi\ \Longleftrightarrow$
$F_{\rho} (H - \lambda) \, \vphi = 0$).} 
If $\psi \in \cH \setminus \{0\}$ solves $ H \psi = 0$
then (\ref{eq-II-11}) implies that
\begin{equation} \label{eq-II-16-2}
F  \chi \psi  \ = \ Q^\# \, H \, \psi \ = \ 0 .
\end{equation}
Furthermore, by \eqref{Hchi},  
$0 \ = \ \bchi \, H \, \psi \ = \
H_\bchi \, \bchi \psi \: + \: \bchi H \chi^2 \psi ,$ 
and hence
\begin{equation} \label{eq-II-18}
Q \, \chi \psi
\ = \
\chi^2 \psi \, - \, \bchi H_\bchi^{-1} \bchi H \chi^2 \psi
\ = \
\chi^2 \psi \, + \, \bchi^2 \psi \ = \ \psi .
\end{equation}
Therefore, $\psi \neq 0$ implies $\chi \psi \neq 0$.

If $\vphi \in \Ran\, \chi \setminus \{0\}$ solves
$F \vphi = 0$ then the definition of $Q$ implies that
\begin{equation} \label{eq-II-22}
\chi Q\vphi =\chi\vphi= \vphi,
\end{equation}
\DETAILS{(\ref{eq-II-11}) implies that\begin{equation} \label{eq-II-19}
H \, Q \vphi \ = \ \chi \, F \vphi \ = \ 0 .
\end{equation}
Since $T$ is invertible on $\Ran\, \bchi$, the trivial
identity $F = T + \chi W Q$ implies that
\begin{equation} \label{eq-II-20}
\bchi \ = \ T^{-1} \, \bchi \, T
\ = \ T^{-1} \, \bchi \, (F \, - \, \chi W Q) ,
\end{equation}
which, together with $\chi = \chi Q$, gives
\begin{equation} \label{eq-II-21}
\one \ = \ \bchi \, + \, \chi \ = \
T^{-1} \bchi F \; + \;
\big(\chi \, - \, T^{-1} \bchi \chi W \big) \, Q .
\end{equation}
Applying (\ref{eq-II-21}) to $\vphi$, we obtain that
\begin{equation} \label{eq-II-22}
\vphi \ = \ \big( \chi \, - \, T^{-1} \bchi \chi W \big) \,
Q \vphi ,
\end{equation}}
which implies that $Q \vphi \neq 0$ provided $\vphi \neq 0$.

\noindent
\textbf{Proof of (iii) ($\dim \cern (H - \lambda) = \dim \cern F_{\rho} (H -
\lambda)$).}
By (i), $\dim \cern H =0$ is equivalent to $\dim \cern F =0$,
assuming that $H\in D(F)$. We may
therefore assume that $\cern H \neq 0$ and $\cern F \neq 0$ are
both nontrivial.
Eq.~(\ref{eq-II-18}) shows that $\chi: \cern H \to \cern F$ is injective,
hence $\dim \cern H \leq \dim \cern F$, and
Eq.~(\ref{eq-II-22}) shows that $Q: \cern F \to \cern H$ is injective,
hence $\dim \cern H \geq \dim \cern F$.
This establishes (iv) and moreover that
$\chi: \cern H \to \cern F$ and $Q: \cern F \to \cern H$ are
actually bijections.
\QED

\section{Smooth Feshbach-Schur Map}\label{sec:sfm}
\subsection{Definition and isospectrality}\label{sec:sfm-def-iso}
We define the smooth Feshbach-Schur map and formulate its important isospectral property
Let $\chi$,  $\bchi$ be a partition of unity on a separable Hilbert
space $\cH$, i.e. $\chi$ and  $\bchi$ are positive operators on
$\cH$ whose norms are bounded by one, $0 \leq \chi, \bchi \leq
\mathbf{1}$, and $\chi^{2}+ \bchi^{2} = \mathbf{1}$. We assume that
$\chi$ and $\bchi$ are nonzero. Let $\tau$ be a (linear) projection
acting on closed operators on $\cH$ with the property that operators
in its image commute with $\chi$ an
$\tau(\textbf{1}) =\textbf{1}$.
Let $\overline{\tau}:= \mathbf{1} - \tau$ and define
\begin{equation} \label{Htauchi} 
H_{\tau,\chi^{\#}} \ \; :=  \tau(H) \: + \: \chi^{\#}
\overline{\tau}(H)\chi^{\#} ,
\end{equation}
where $\chi^{\#}$ stands for either $\chi$ or $\bchi$.

Given $\chi$ and $\tau$ as above, we denote by $D_{\tau,\chi}$ the
space of closed operators, $H$, on $\cH$ which belong to the domain
of $\tau$ and satisfy the following three conditions:

(i) $\tau$ and $\chi$ (and therefore also $\btau$ and $\bchi$) leave
the domain $D(H)$ of $H$ invariant:
\begin{equation}\label{domtauchi} 
D(\tau(H))=D(H)\ \mbox{and}\  \chi D(H)\subset D(H),
\end{equation}

(ii)
\begin{equation}\label{HtaubchiInvert} 
H_{\tau,\bchi}\ \mbox{is (bounded) invertible on}\
\Ran \, \bchi,
\end{equation}

(iii)
\begin{equation}\label{btauHbound} 
\overline{\tau}(H) \chi\ \mbox{and}\ \chi
\overline{\tau}(H)\ \mbox{extend to bounded operators on}\ \cH.
\end{equation}
(For more general conditions see \cite{bcfs1, 
 GriesemerHasler1}.)
%
%

The \textit{smooth Feshbach-Schur map (SFM)} maps operators 
 from $D_{\tau,\chi}$ into operators on $\cH$ by 
\begin{equation} \label{sfm}
 F_{\tau,\chi} (H) \ := \ H_0 \, + \, \chi W\chi \, -
\, \chi W \bchi H_{\tau,\bchi}^{-1} \bchi W \chi ,
\end{equation}
where $H_0 := \tau(H)$ and $W := \overline{\tau}(H)$. Note that $H_0$
and $W$ are closed operators on $\cH$ with coinciding domains, $
D(H_0)= D(W)=D(H)$, and $H = H_0 + W$. We remark that the domains of
$\chi W\chi$, $\bchi W\bchi$, $H_{\tau,\chi}$, and $H_{\tau,\bchi}$
all contain $D(H)$.

Define operators $Q_{\tau,\chi} (H)  :=  \chi \: - \: \bchi \, H_{\tau,\bchi}^{-1}
\bchi W \chi$ and $Q_{\tau,\chi} ^\#(H)  :=  \chi \: - \: \chi W \bchi \,
H_{\tau,\bchi}^{-1} \bchi$. The following result (\cite{bcfs1}) generalizes Theorem \ref{thm:isospF} above; its proof is similar to the one of that theorem:

\begin{theorem}[Isospectrality of SFM] \label{sfmisosp}
 {\em Let $0 \leq \chi \leq \one$ and $H\in D_{\tau,\chi}$ be an operator on a separable Hilbert space $\cH$. Then we
have the following results:
\begin{itemize}
\item[(i)]
$H$ is bounded invertible on $\cH$ if and only if
$F_{\tau,\chi} (H)$ is bounded invertible on $\Ran\, \chi$.
In this case
\begin{eqnarray} \label{Hinvrepr} 
H^{-1} & = & Q_{\tau,\chi} (H) \: F_{\tau,\chi} (H)^{-1} \: Q_{\tau,\chi} (H)^\#
\; + \; \bchi \, H_\bchi^{-1} \bchi , \hspace{8mm}
\\ \label{FHinvrepr} 
F_{\tau,\chi} (H)^{-1} & = &
\chi \, H^{-1} \, \chi \; + \; \bchi \, \tau(H)^{-1} \bchi . \hspace{8mm}
\end{eqnarray}
\item[(ii)]
If $\psi \in \cH \setminus \{0\}$ solves $H \psi = 0$
then $\vphi := \chi \psi \in \Ran\, \chi \setminus \{0\}$
solves $F_{\tau,\chi} (H) \, \vphi = 0$.
\item[(iii)]
If $\vphi \in \Ran\, \chi \setminus \{0\}$ solves
$F_{\tau,\chi} (H) \, \vphi = 0$
then $\psi := Q_{\tau,\chi} (H) \vphi \in \cH \setminus \{0\}$ solves $H \psi = 0$.
\item[(iv)]
The multiplicity of the spectral value $\{0\}$ is conserved in the sense that
$\dim \Null H $ $= \dim \Null F_{\tau,\chi} (H)$.
\DETAILS{\item[(v)]
Assume, in addition, that $H = H^*$ and $\tau(H)=\tau(H)^*$ are self-adjoint,
and introduce the bounded operators
\begin{eqnarray} \label{eq-II-8.1a}
M & := & H_{\tau,\chi}^{-1} \, \bchi \, (H-\tau (H)) \, \chi
\hspace{8mm} \mbox{and}
\\ \label{eq-II-8.2}
N & := &
\big( \one \, + \, M^* M \big)^{-1/2} .
\end{eqnarray}
Then, for any $\psi \in \cH$,
\begin{eqnarray} \label{eq-II-8.3}
\lefteqn{
\lim_{\eps \searrow 0} \; \rIm
\big\la \psi , \: (H - i \eps)^{-1} \, \psi \big\ra
\ = \
} \\ \nonumber & &
\lim_{\eps \searrow 0} \; \rIm
\Big\la N \, Q_{\tau,\chi} (H)^* \, \psi , \:
\big( N \, F_{\tau,\chi} (H) \, N \, - \, i \eps \big)^{-1}
\: N \, Q_{\tau,\chi} (H)^* \, \psi \Big\ra
\end{eqnarray}
and
\begin{eqnarray} \label{eq-II-8.4}
\lefteqn{
\lim_{\eps \searrow 0} \; \rIm
\big\la \psi , \:
\big( N \, F_{\tau,\chi} (H) \, N \, - \, i \eps \big)^{-1}
\, \psi \big\ra
\ = \
} \\ \nonumber & &
\lim_{\eps \searrow 0} \; \rIm
\big\la \chi \, N^{-1} \, \psi , \:
( H  \, - \, i \eps \big)^{-1} \, \chi \, N^{-1}
\, \psi \Big\ra .
\end{eqnarray}}
\end{itemize}}
\end{theorem}
We also mention the following useful property of $ F_{\tau, \chi}$:
 \begin{equation} \label{Hsa}
H\ \mbox{is self-adjoint}\ \quad \Rightarrow \quad   F_{\tau, \chi}(H)\ \mbox{is self-adjoint}.
\end{equation}

\subsection{ Transfer of local decay}\label{sec:locdectransf}
We have shown above that the  smooth Feshbach-Schur map is isospectral. In fact, under certain additional conditions it preserves (or transfers) much stronger spectral property -
the limiting absorption principle (LAP) (\cite{fgs3}), which is defined as follows. 
Let $\Delta\subset \R$ be an interval, $\nu>0$ and $B$, a self-adjoint operator.
We say that a  $C^1$ family of self-adjoint operators, s.t. $H(\lambda) \in D_{\tau,\chi}$ has the ($\Delta, \nu, B, \theta$) \textit{limiting  absorption principle} (LAP) property iff
\begin{eqnarray} 
\lim_{\eps\rightarrow 0+}\la B\ra^{-\theta}\big(H(\lambda)-i\eps\big)^{-1}\la B\ra^{-\theta}\ \mbox{exists and}\ \in C^\nu(\Delta).\label{BHinvB'}
\end{eqnarray}
Usually LAP holds for $\nu< \theta-\frac{1}{2}$.
 One can show that the LAP implies the local decay property (see, e.g. \cite{RS}, vol III;  recall that the definition of the local decay property is given in Section \ref{sec:relat-res}).
\begin{theorem} \label{thm:LAPtransfer}
 {\em Let $\Delta\subset \R$ and $,\ \forall \lambda\in\Delta,\ H(\lambda)$ be a $C^1$ family of self-adjoint operators, s.t. $H(\lambda) \in D_{\tau,\chi}$.  
 Assume that there is a self-adjoint operator $B$ s.t.  
\begin{equation} \label{commbnds} 
\adj_B^j(A)\ \mbox{is bounded and differentiable in}\ \lam,\ \forall j\le 2,
\end{equation}
where $A$ is one of the operators
$\chi,\ \overline{\chi},\ \chi \bar\tau(H(\lambda)),\ \bar\tau(H(\lambda))\chi$, $\partial_\lambda^k (\bchi H_{\tau,\bchi}(\lambda)^{-1} \bchi),$ $k=0, 1$.
Then, for any $0\le \nu\le 1$ and $0 < \theta \le 1$ and in the operator norm, we have
\item 
\begin{eqnarray} 
&&\lim_{\eps\rightarrow 0+}\lefteqn{\la B\ra^{-\theta}\left(F_{\tau, \chi}(H(\lambda))-i\eps\right)^{-1}\la B\ra^{-\theta}\ \mbox{exists and}\ \in C^\nu(\Delta)} \label{BFinvB} 
\\ &&\Rightarrow
\lim_{\eps\rightarrow 0+}\la B\ra^{-\theta}\big(H(\lambda)-i\eps\big)^{-1}\la B\ra^{-\theta}\ \mbox{exists and}\ \in C^\nu(\Delta).\label{BHinvB}
\end{eqnarray}}
\end{theorem}
 This allows one to reduce the proof of the LAP for the original operator, $H-\lam$, to the proof of this property for a much simpler one, $\cR_\rho^n(H-\lam)$.

\bigskip

\subsection{ Transfer of analyticity}\label{sec:anal-transf}

\begin{theorem} 
 {\em Let $\Lam$ be an open set in $\C$ and
$H(\lambda)$, $\lambda\in\Lam$, a family of operators with a fixed domain, which belong to the domain of $F_{\tau\chi}$. Assume $H(\lambda)$ and $\tau(H(\lambda))$, with the same domain, are analytic in the sense of Kato (see e.g. \cite{RS}, vol IV). 
Then we have that
\begin{itemize}
\item $F_{\tau, \chi}(H(\lambda))$ is an analytic in $\lambda\in\Lam$ family of operators.
\end{itemize}}
\end{theorem}

\begin{proof} Since that $H(\lambda)$,  $\ H_0(\lambda):=\tau(H(\lambda))$ and $W(\lambda) \chi\ \mbox{and}\ \chi
W(\lambda)$, where $W(\lambda):=\overline{\tau}(H(\lambda))$, are analytic in
$\lambda \in \Lam,$ we see from the
definition of the smooth Feshbach-Schur map $F_{\tau\chi}$ in \eqref{sfm}
 that $F_{\tau\chi}(H(\lambda))$ is analytic in $\lambda \in
\Lam$, provided $\bchi_\rho H(\lambda)_{\tau,\bchi_\rho}^{-1}$ $ \bchi_\rho$
is analytic in $\Lam$. The analyticity of the latter family
follows by the Neumann series argument. 
\end{proof}
One can generalize the above result to $\Lam$'s which are open sets in a complex Banach space. Recall that a complex vector-function $f$ in an open set $\Lam$ in a
complex Banach space $\cW$ is said to be \textit{analytic} iff it is
locally bounded and G\^{a}teaux-differentiable. One can show that
$f$ is analytic iff $\forall \xi \in \cW,\ f(H+ \tau \xi)$ is
analytic in the complex variable $\tau$ for $|\tau|$ sufficiently
small (see \cite{Berger, HillePhillips}). Furthermore if $f$ is
analytic in $\Lam$ and $g$ is an analytic vector-function from an
open set $\Omega$ in $\mathbb{C}$ into  $\Lam$, then the composite
function $f\circ g$ is analytic on $\Omega$.


\bigskip


\bigskip

\section{Mass Renormalization} \label{sec:mass-renorm}

As the free electron is surrounded by virtual 'soft' photons its effective (inertial0 mass is greater than the value ('bare' mass) entering its Hamiltonian. One calls this  electron \textit{mass renormalization}. We begin with analyzing the definition of (inertial) mass in Classical Mechanics. Consider a classical particle with the Hamiltonian $h(x, k):=K(k)+V(x)$, where $K(k)$ is some function describing the kinetic energy of the particle. To find the particle mass in this case we have to determine the relation between the force and acceleration at very low velocities. The Hamilton equations give $\dot x=\p_k K$ and  $\dot k=F$, where $F=-\p_x V$ is the force acting on the particle. Assuming that $K$ has a minimum at $k=0$ and expanding $\p_k K (k)$ around $0$, differentiating the resulting relation $\dot x= K''(0)k$, where $K''(0)$ is the hessian of $K$ at $k=0$, w.r. to time and using the second Hamilton equation, we obtain $\ddot x= K''(0)F(x)$. This suggests to define the mass of the particle as $m=K''(0)^{-1}$, i.e. as the inverse of the Hessian of the energy, in the absence of external forces, as a function of of momentum at $0$. ($K(k)$ is called the dispersion relation.) We adopt this as a general definition: \textit{the (effective) mass of a particle interacting with fields is  the inverse of the Hessian of the energy of the total system as a function of of the total momentum at $0$.}

Now, we consider a single non-relativistic electron coupled to quantized electromagnetic field. Recall that the charge of electron is denoted by $-e$ and its \textit{bare} mass in our units is $m$. The corresponding Hamiltonian is
\begin{equation}\label{H-single}
	\Hn \, := \, \frac{1}{2m}(  i\nabla_x \otimes 1_f \,  - \, \c \Af(x) \, )^2
	\, + \, \1_{el} \otimes H_f ,
\end{equation}
acting on the space $L^2({\mathbb R}^{3})\otimes \cF\equiv \cH_\part \otimes \cH_f$. It is the generator for the dynamics of a single non-relativistic electron, and of
the electromagnetic radiation field, which interact via minimal coupling. Here recall $\Af(x)$ and $H_f$ are the quantized electromagnetic vector potential with ultraviolet cutoff and the field Hamiltonian and are defined in \eqref{A} and \eqref{Hf}

\bigskip

The system considered is \textit{translationally invariant} \index{translation ! invariance} in the sense that $\Hn$ commutes with the translations, $T_y$, 
 \DETAILS{In the Fock space representation this gives: $T_{y}:   	 \oplus_n\Psi_n(x, k_1, \dots, k_n)$\\ $\ra \oplus_n e^{ iy\cdot (k_1+\dots k_n)}	 \Psi_n(x+y, k_1, \dots, k_n)$ and therefore can be rewritten as $T_{y}: \Psi(x)\ra  e^{ iy\cdot P_\re}	 \Psi(x+y)$, where $ P_\re$ is the momentum
operator associated to the quantized radiation field,
\[\Pf=\sum_\lambda\int dk \, k \, a_\lambda^*(k)a_\lambda(k).\] 
As it is straightforward to show $T_y$ are unitary operators and that they satisfy the relations
$ T_{x+y} = T_xT_y,$
and therefore $y\ra T_{y}$ is a unitary Abelian representation of $\R^3$. We also note that $T_y$ commutes with  $H$,
\begin{equation*} T_y \Hn = \Hn T_y. \end{equation*}
Indeed, using \eqref{Ty} and the definitions of the operators $A(x)$ and $E(x)$, we arrive at  $T_y A(x) = (t_{y}A)(x) T_y$ and $T_y E(x) = (t_{y}E)(x) T_y$, which, due to the definition of  $\Hn$, gives  $T_y \Hn = \Hn T_y$.
Finally, we observe that the group $T_y$ is generated  by}
\begin{equation*} T_y \Hn = \Hn T_y, \end{equation*}
which in the present case take the form
 \begin{equation}\label{Ty'}
	T_{y}:  \Psi(\ux)\to  e^{ iy\cdot P_\re}	 \Psi(\ux+\uy), 
	\end{equation}
This as before leads to $\Hn$ commuting with the total momentum operator,
\eqn
	\Ptot \, := \, P_\el\otimes\1_f \, + \, \1_{el}\otimes \Pf,
\eeqn
of the electron and the photon field: $[H, \Ptot]=0$. Here
$P_\el:=-i\nabla_x$ and $ \Pf=\sum_\lambda\int dk \, k \, a_\lambda^*(k)a_\lambda(k)$ are electron and field momenta. 
Again as in Appendix \ref{sec:transl-Ham}, 
this leads to
%
%
\DETAILS{\bigskip

\subsection{Fiber decomposition with respect to total momentum}\label{subsec:fiber}
Let $\cH$ be the direct integral
	$\cH = \int_{\R^3}^\oplus \cH_P dP, $ 
with the fibers $\cH_P := \cF$ (this means that $\cH = L^2(\R^3, dP; \cF)$) and 
		define $U : \cH_\el \otimes \cH_{\re} \to \cH$ on smooth functions with compact domain by the formula
	\begin{equation}\label{U}
		(U \Psi)(P) = \int_{\mathbb{R}^3}e^{i(P-P_\re)\cdot y}\Psi(y) \d y. 
	\end{equation}
Then $U$ extends uniquely to a unitary operator (see below). Its converse is written, for  $\Phi (P)\  \in \mathcal{F}$, as
 \begin{equation}\label{Uinv} 
 (U^{-1}\Phi)(x)=\int_{\mathbb{R}^3} e^{-\i x\cdot(P-P_\re)}\Phi (P) \d P.
\end{equation}
The functions $\Phi (P)\ =\int_{\mathbb{R}^3} dye^{i(P-P_\re)\cdot y}\Psi(y)$ are called fibers of $\Psi$. We denote $\widehat\Psi(P) \, = \, (U \Psi)(P)$ and $	\check\Phi (x)
 \, = \, (U^{-1} \Phi)(x)$
\DETAILS{ For $k \in \R^3$, let $H_{ k}$ be the operator $H$ acting on $\cH_k$ with domain consisting of those $v \in \cH_k \cap H^2$ such that $\Psi$
	satisfies the boundary conditions $T_y \Psi(x) = \chi_k(t) \Psi(x)$. Then
	\begin{equation}
	\label{K-decomposition}
		UH U^{-1} = \int_{\R^3}^\oplus H_{ k} dk.
	\end{equation}
	and
	\begin{equation}
	\label{KKkspecrelat}
		\sigma(H) = \bigcup_{k\in\R^3} \sigma(H_{ k}).
	\end{equation}}
\DETAILS{\bigskip

We consider  the direct integral decomposition
\eqn
	\cH \, = \, \int^\oplus dp \, \cH_p \,,
\eeqn
with respect to $\Ptot$, where each fiber $\cH_p\cong {\rm Ker}(\Ptot-p)$ is isomorphic to $\Fo$.
It corresponds to the natural isomorphism
\eqn
	\cH \, \cong \, L^2(\R^3,\Fo) \,
\eeqn
which we will use in the sequel.


For $\psi\in\cH$, we define the Fourier transform with respect to the electron variable,
\eqn
	\widehat\psi(p) \, = \, \int dx \, e^{-i(p-\Pf)x} \psi(x)
\eeqn
with inverse transform
\eqn
	\check\phi
 \, = \, \int dp \, e^{i(p-\Pf)x}\phi(p) \,.
\eeqn}
\begin{lemma}
The  operations \eqref{U} 
and \eqref{Uinv} 
define unitary maps $L^2({\mathbb R}^{3})\otimes \cF\rightarrow\cH$ and  $\cH\rightarrow L^2({\mathbb R}^{3})\otimes \cF$, and are mutual inverses.
\end{lemma}
\prf
By density, we may assume that $\Psi$ is a generalized Schwartz class function,
\eqn
	{\mathcal S}_{\Fo} \, :=\,
	\big\{ \, \Psi\in L^2({\mathbb R}^{3})\otimes \cF \, \big| \, \sup_x\|x^\alpha\partial_x^\beta\Psi\|_{  \Fo} \, < \, \infty
	\; \; \; , \; \; \;  \forall \alpha,\beta\in\N_0^3 \, \big\} \,.
\eeqn
Then, it follows from standard arguments in Fourier analysis that
\eqn\label{eq-FT-def-1}
	(\widehat\psi)^\vee(x) &=&
	\int \d P \, e^{-i(P-\Pf)x} \int dx' \, e^{i(p-\Pf)x'} \, \Psi(x')
	\nonumber\\
	&=&
	\int dx' \,  \int \d P \,   e^{-ip(x-x')}e^{i\Pf(x-x')} \, \Psi(x')
	\nonumber\\
	&=&
	  \int dx' \,  \delta(x-x') \, e^{i\Pf(x-x')} \, \Psi(x')
	 \nonumber\\
	 &=&\Psi(x) \,.
\eeqn
On the other hand, for $\psi\in {\mathcal S}_{\Fo} $,
\eqn\label{eq-FTinv-def-1}
	(\check\Phi)^{\widehat{\;}}(P)
	&=&
	\int dx \, e^{i(P-\Pf)x} \int dq \, e^{i(q-\Pf)x} \, \Phi(q)
	\nonumber\\
	&=&
	 \int dq \, \int dx \, e^{i(p-q)x} \,  \Phi(q)
	 \nonumber\\
	 &=&
	 \int dq \, \delta(P-q) \, \Phi(q)
	 \nonumber\\
	 &=&\Phi(P) \,.
\eeqn
From the density of ${\mathcal S}_{\Fo}$ in $L^2({\mathbb R}^{3})\otimes \cF$, we infer that \eqref{eq-FT-def-1}
and \eqref{eq-FTinv-def-1}  define bounded maps 
which are mutual inverses. Unitarity can be checked easily.
\endprf
\DETAILS{We observe that
\eqn
	(\Ptot\psi)^{\widehat{\;}}(p) & = & \int dx \, e^{i(p-\Pf)x} (-i\nabla_x+\Pf)\psi(x)
	\nonumber\\
	&=& \int dx \, \big( \, (-i\nabla_x+\Pf)e^{i(p-\Pf)x}\,\big) \, \psi(x)
	\nonumber\\
	&=&p \, \widehat\psi(p) \,.
\eeqn
Hence, $\widehat\psi(p)\in\cH_p\cong {\rm Ker}(\Ptot-p)$, for any $p\in\R^3$.}
%
It follows that $\Hn$ admits}
 the fiber decomposition
\begin{equation}
U\Hn U^{-1} = \int_{\mathbb{R}^3}^{\oplus} \Hn(P) \d P,
\end{equation}
where the fiber operators $\Hn(P)$, $P \in \mathbb{R}^3$, are self-adjoint operators on $\mathcal{F}$. Using $a(k)e^{-\i x\cdot P_\re}=e^{-\i x\cdot (P_\re+k)}a(k)$ and  $a^*(k)e^{-\i x\cdot P_\re}=e^{-\i x\cdot (P_\re-k)}a^*(k)$, we find $\nabla_x  e^{\i x\cdot(P-P_\re)}\Af(x) e^{\i x\cdot(P-P_\re)} =0$ and therefore
\begin{equation}
\Af(x) e^{\i x\cdot(P-P_\re)} = e^{\i x\cdot(P-P_\re)}\Af(0).
\end{equation}
Using this and \eqref{Uinv},  we compute $\Hn  (U^{-1}\Phi)(x)=\int_{\mathbb{R}^3} e^{\i x\cdot(P-P_\re)}\Hn(P)\Phi (P) \d P$, where  $\Hn(P)$ are Hamiltonians on the fibers $\cH_P := \cF$ given explicitly by
\begin{equation}
\Hn(P) =  \frac{1}{2m}\big ( P - P_\re  -e\Af)^2 + H_\re
\end{equation}
where $\Af:=\Af(0)$. Explicitly, $\Af$ is given by
\begin{equation}\label{Achi}
\Af \, = \,
	\sum_{\lambda}\int  \d k \, \frac{\chi(|k|)}{|k|^{1/2}} \,
	  \e_\lambda(k) \,\{ \,   a_\lambda(k) \, + \, a^*_\lambda(k) \, \}. 
\end{equation}
%
%
\DETAILS{For an observable $A$ on $ \mathcal{H}_{\mathrm{el}} \otimes \mathcal{H}_\re$, we consider the operator $UA U^{-1}$, acting on the space $\cH = \int_{\R^3}^\oplus \cH_P dP$. It can be written as
$$(UA U^{-1}\Phi)(P)=\int_{\mathbb{R}^3} A_{PQ}\Phi(Q)dQ,$$
for some operator-function $A_{PQ}$ on $\cF$. If $A\in \cA$, then $A_{PQ}\in {\frak W}(L^{2}_0)$.
We write $UA U^{-1}=\int A_{PQ} dPdQ.$ Note that the evolution, $\al^t(A)$, of $A$ is then given by
\begin{equation}
\al^t(A)= \int_{\mathbb{R}^3\times \R^3}dPdQ\alpha_{PQ}^t (A_{PQ}),
\end{equation}
where $\alpha_{PQ}^t (A_{PQ}):=e^{itH_P}A_{PQ} e^{-itH_Q}$. If an observable $A$ is translation invariant, then we can write it as a fiber integral, $A = \int_{\mathbb{R}^3}^{\oplus} A_P \d P.$ If $A$ is translation invariant, then so is the observable $\al^t(A)$ and
\begin{equation}
\al^t(A)=\int_{\mathbb{R}^3} \al^t_{P}(A_{P})dP,\ \mbox{where}\ \al^t_{P}(A_{P}):=e^{itH_P}A_{P} e^{-itH_P}.
\end{equation}
We write $\al^t=\int \al^t_{P}dP$.  An example of  a  translation invariant observable is the particle momentum  $P_\el=-i\nabla=i[H, x]$. In this case we have $P_\el = \int_{\mathbb{R}^3}^{\oplus} (P-P_\re) \d P = \int_{\mathbb{R}^3}^{\oplus} P \d P-P_\re.$}
\DETAILS{\bigskip
Next, we consider the Hamiltonian \eqref{eq-Hn-def-1}.
Clearly, $[\Hn,\Ptot]=0$, by translation invariance.
Then,
\eqn
	(\Hn \psi)^{\widehat{\;}}(p) \, = \, \Hn(p) \widehat\psi(p)
\eeqn
where $\Hn(p) =  \Hn\big|_{\cH_p}$ is the fiber Hamiltonian associated to total momentum $p$. Let $\Af:=\Af(0)  $. We compute
\eqn
	\Hn(p) =\frac12( \, p \,  - \,  \Pf \, - \, \c \Af  \, )^2
	\, + \,  H_f \,
\eeqn
We then have that
\eqn
	\Hn \, = \, \int e^{i(p-\Pf)x}\Hn(p)e^{-i(p-\Pf)x}dP_{\Ptot}(p)
\eeqn
which is the spectral representation of $\Hn$
with respect to the spectral measure associated to $\Ptot$.}
%

\bigskip
Consider the infimum $E(P):=\inf\sigma (\Hn(P))$ of the spectrum of the fiber Hamiltonian $\Hn(P)$. Note that for $e=0$, $E(P)|_{e=0} \; =: \; E_0(P)$ is the ground state energy of  $H_0(P):=\Hn(P)|_{e=0} =  \frac{1}{2m}\big ( P - P_\re \big)^2 + H_\re$ with the ground state $\vac$ and is $E_0(P) \; = \frac{|P|^2}{2 m}$. The renormalized electron mass 
is defined as
\eqnn
	E(P) \; =  \, \frac{|P|^2}{2 m_{ren}} \, + \, O(|P|^3)
\eeqnn
where the left hand side is computed perturbatively up the second order in the coupling constant (charge).
Provided that $E(P)$ is spherically symmetric and $C^2$ at $P=0$,
and therefore, in particular,  $\partial_{|P|}  E(0) = 0$,
we define the renormalized electron mass at zero total momentum 
as
\eqnn
	m_{ren} \; := \; \frac{1}{ \partial_{|P|}^2 E(0) } \;.
\eeqnn
The kinematic meaning of this expression is as follows.
The ground state energy $E(P)$ can be considered as an effective Hamiltonian
of the electron in the ground state.
(The propagator $\exp(-itE(P))$ determines the propagation properties of a
wave packet formed of dressed one-particle states  with a wave function supported near $p = 0$
-- which exist as long as there is an infrared regularization.)
The first Hamilton equation gives the expression for the electron velocity as
\eqnn
	v \; = \; \partial_P E(P) \;.
\eeqnn
Expanding the right hand side in $P$ we find $v  = {\rm Hess} \, E(0)  P + O(P^2)$,
where
\eqn\label{eq:Hess-E-def-1}
		        \big({\rm Hess}\,E(P)\big)_{ij} \; = \;
        \Big(\delta_{ij}-\frac{P_i P_j}{|P|^2}\Big)
        \frac{\partial_{|P|}E(P)}{|P|}
        \, + \, \frac{P_i P_j}{|P|^2}
        \partial_{|P|}^2E(P) \;
\eeqn
is the Hessian of $E(P)$ at $P\in\R^3$ (given that $E(P)$ is spherically symmetric,
and $C^2$ in $|P|$ near $P=0$).
It follows from (\ref{eq:Hess-E-def-1}) and the fact $\partial_{|P|}  E(0) = 0$ that
${\rm Hess} \, E(0)  = \partial_{|P|}^2 E(0) \, \1$, so that
\eqnn
	v  \; = \; \partial_{|P|}^2 E(0,\sigma)  \, P \, + \, O(P^2) \;.
\eeqnn
This suggests  taking  $(\partial_{|P|}^2 E(0) )^{-1}$ as the renormalized
electron mass at $P=0$.
\DETAILS{For sufficiently small momenta $p$ we define the renormalized electron mass as
        $m_{ren}(p) \; := \; \frac{1}{\partial_{|p|}^2 E(p)} \;. $ 
We remark that a different notion of the renormalized electron mass
in non-relativistic QED can be introduced through the binding of an electron to a nucleus, see \cite{hasei,lilo1}.}
The following result is proven in \cite{bcfs2,ch,cfp1}: 
\begin{theorem}
 {\em For any $P$, s.t. $|P|<\frac13$, the infimum of the spectrum $\Eg(P)=\inf{\rm spec}(\Hn(P))$ is twice differentiable and satisfies $ 1 \; \leq \; m_{ren} \; \leq \; 1 \, + \, c \, g^2 \;$ for some $c>0$.}
\end{theorem}
\bigskip
\DETAILS{We obtain this result  by a formal perturbation theory in $e$. To this end we  compute $\Eg(P)$ (by perturbation theory in $e$). We write $\Hn(P)$ as $\Hn(P)=H_0(P)+W$ where $H_0(P)=  \big ( P - P_\re )^2 + H_\re$. Let $P$ be the orthogonal projection onto the
eigenspace $\Null(H_0 - \lam_0)$ spanned by the vacuum $\vac$, let
$\oP := \bfone - P$ and let $H_{0\oP}(P) := \oP H_0(P) \oP \restriction_{\Ran \oP}$ . Then a general result of perturbation theory (see Appendix \ref{subsec:pert}, \eqref{pert-exp})
gives
\begin{equation}\label{eq:pertexp}
  \Eg(P) = \E0(P) + \sum_{n=0}^\infty(-1)^n \lan \vac,  W' (\bar{R}_0 W')^n  \vac \ran,
\end{equation}
where $\E0(P)=|P|^2/2$, $W':=\oP(W+  \Eg(P) - \E0(P))\oP$ and $\bar{R}_0 = \oP ( H_{0 \oP}(P) - \E0(P) )^{-1} \oP$.
Using that $W=\, -e\frac{1}{m}( \, P \,  - \,  \Pf ) \, \cdot \Af +\frac{1}{2m}e^2\Af^2$ and therefore in particular $\lan \vac, W \vac \ran=0$. 
Now, we compute the second term on the r.h.s: 
\eqn  \label{E-ren-pert-lead-1}
        &&\lan W \vac, \bar{R}_0 W \vac \ran \;= \; \frac{e^2}{m^2}\lan P \, \cdot \Af \vac, \bar{R}_0 P \, \cdot \Af \vac \ran     \;\nonumber\\
         & = & \; \frac{e^2}{m^2}P_iP_j\sum_{\lambda}\int  dk \, \frac{\kappa_\sigma(|k|)}{|k|^{1/2}} \,
	  \e^i_\lambda(k) \,\sum_{\lambda'} \int  dk' \, \frac{\kappa_\sigma(|k'|)}{|k'|^{1/2}} \,
	  \e^j_{\lambda'}(k') \,   \;\nonumber\\
         & \times &  \lan  \,   \vac, a_\lambda(k) ( H_{0}(P) - \E0(P) )^{-1}   a^*_{\lambda'}(k') \vac \ran.
\eeqn
To compute the inner product on the r.h.s., we move the operators $a(k)$ to the extreme right and
the operators $a^*(k)$ to the extreme left.  In doing this
we use the following rules:
\begin{enumerate}
\item
  $a(k)$ is pulled through $a^* (k)$ according to the
  relation
  \[
    a(k) a^* (k') = a^* (k') a(k) + \delta (k-k')
  \]
\item
  $a(k)$ and $a^*(k)$ are pulled through $R_{0,\oP}$ according
  to the relations
  \begin{eqnarray*}
    a(k) ( H_{0} - \E0(P) )^{-1} &=& ( H_{0}(P) +\omega(k)- \E0(P) )^{-1} a(k),\\
    ( H_{0} - \E0(P) )^{-1}a^*(k) &=& a^*(k)( H_{0}(P) +\omega(k)- \E0(P) )^{-1}
  \end{eqnarray*}
  (see the equations \eqref{eq:pull} and of \eqref{eq:pull2} of Appendix \ref{sec:pullthrough}).
\end{enumerate}
Using that $a(k)\vac=0$ and $H_{0}(P)\vac=\frac{|P|^2}{2}\vac$, we obtain, to leading order in $\c$,
\eqn  \label{E-ren-pert}
E(P) \; = 
\frac{|P|^2}{2m}\Big(1   -    8\pi \, \c^2\frac{\kappa}{m} \Big)
\, + \, O(\c^4)
\eeqn
so that we infer \eqref{m-ren-pert}.
Note that (\ref{E-ren-pert}) and (\ref{m-ren-pert}) depend
on the ultraviolet cutoff $\Lambda$.
\[\mbox{Representation of \eqref{E-ren-pert-lead-1} and more generally terms in \eqref{eq:pertexp} in terms of Feynman diagrams.}\]}
A presentation of the leading order calculations 
can be found in \cite{HiroshimaSpohn2}.

\begin{remark}
 {\em The estimate  $ 1 \; \leq \; m_{ren} \; \leq \; 1 \, + \, c \, g^2 \;$ 
reflects the fact that the mass of the electron is increased
by interactions with the photon field.}
\end{remark}

\DETAILS{\begin{remark}
The existence of the ground state at $p=0$
(see also \cite{mol}), and
renormalization of the electron mass is an important ingredient for the
phenomenon of {\em enhanced binding}, \cite{cvv} and \cite{hispo1,hvv}. A
Schr\"odinger operator with a non-confining potential
can exhibit a bound state when the interaction of the
electron with the quantized electromagnetic field is included.
Binding to a shallow potential can be energetically more favorable for the electron
than forming an infraparticle through binding of a cloud of soft photons.
\end{remark}}

\bigskip


\section{One-particle States}\label{sec:One-particle states} 
First we note that for $e=0$, the one-particle states of the Hamiltonian  $H_0:=\Hn|_{e=0} =  |P_\el|^2 + H_\re$ are the generalized eigenfunctions
\begin{equation}\label{eq:2.4}
    e^{-iP\cdot x} \otimes  \,  \Omega,
\end{equation}
corresponding to the spectral points $E (P) = \frac{|P|^2}{2m}$. This corresponds to the true ground state $\vac$ of the fiber Hamiltonians $H_0(P):=\Hn(P)|_{e=0} =  \frac{1}{2m} ( P - P_\re )^2 + H_\re$. The generalization of such a state for the interacting system would be the ground state of $\Hn(P)$, if it existed. However, we have
\begin{thm}
$\Hn(P)$ has a ground state if and only if $P=0$.
\item 
\end{thm}
To define one particle states for the interacting model, we first introduce IR regularization,
$\ks\in C_0^\infty([0,\kappa];\R_+)$ is assumed to be a smooth cutoff function obeying 
$\lim_{x\rightarrow0}\frac{\ks(x)}{x^\sigma} \; = \; 1\;.$ 
 The corresponding Hamiltonian is
\eqn\label{eq-Hn-def-1}
	\Hns \, := \, \frac{1}{2m}\big( \, -i\nabla_x \otimes 1_f \,  + \, \c \Afs(x) \, \big)^2
	\, + \, \1_{el} \otimes H_f
\eeqn
with the quantized electromagnetic vector potential subjected, besides the ultra-violet
cutoff, also  infrared regularization,
\eqn
	\Afs(x) \, = \,
	\sum_{\lambda}\int  dk \, \frac{\kappa_\sigma(|k|)}{|k|^{1/2}} \,
	\{ \,  \e_\lambda(k) \,  e^{-ikx} \otimes a_\lambda(k) \, + \, h.c. \, \}.
\eeqn
Let $\cS:=\{P\in \R^3\  |\ |P|<1/3\}$.
\begin{thm}
For $P\in\cS$ and for any $\sigma>0$, the infimum of the spectrum $\Egs(P)=\inf{\rm spec}(\Hns(P))$   is a simple eigenvalue.\\
\end{thm}
\begin{remark}
The upper bound on $|P|$ of $\frac13$ is not optimal, but we note that,
For $E(P)$ to be an eigenvalue, $|P|$ cannot exceed a critical value $P_c< 1$ (corresponding
to the speed of light). As $|P|\rightarrow P_c$,
it is expected that the eigenvalue at $E(P)$ dissolves in
the continuous spectrum, while a resonance appears.
This is a manifestation of a phenomenon analogous to
Cherenkov radiation.
\end{remark}
Let $\Psig(P)\in\Fo$ denote the associated normalized fiber ground state, $\|\Psig(P)\|_\Fo=1$,
for $P\in\cS$,
\eqn
	\Hns(P)\Psig(P) \, = \, \Egs(P) \, \Psig(P) \,.
\eeqn
 The vector
$\Psi(P,\sigma)$ is an {\em infraparticle state}, describing a compound
particle comprising the electron together with a cloud of low-energy (soft)
photons whose expected number diverges as $\sigma\rightarrow0$, unless $p=0$. 
For $P\in \cS$,  we introduce the Weyl operators
\eqn
	W_{\nablE(P)}(x) 
:= e^{D(x)-D^*(x)}
\eeqn
where
\eqn
	D(x) \, := \, \sum_\lambda  \int dk \,  G_\lambda(k,p) \,  e^{-ikx} \, a_\lambda(k) \,,
\eeqn
with
\eqn
	G_\lambda(k,p)  \, := \,\alpha^{\frac12} \, \kappa_\sigma(|k|) \,
	\frac{\nablE(p)\cdot\e_\lambda(k)  }{|k|^{1/2}(|k|-\nablE(p)\cdot k)} \,.
\eeqn
(Here and in the sequel, we will use the abbreviated notation 
	$\nablE(p) \, \equiv \, \nabla_p \Eg(p) \,.$)
We observe that they commute with the total momentum operator,
\eqn
	[\Ptot \, , \, W_{\nablE(p)}(x)] \, = \, 0 \,.
\eeqn
To see this, we note that $[\Ptot,D(x)]=0=[\Ptot,D^*(x)]$.
Indeed, we have that
\eqn
	\lefteqn{
	\Ptot \, \sum_\lambda \int dk \,  G_\lambda(k,p) \,  e^{-ikx} \, a_\lambda(k) \, \psi(x)
	}
	\nonumber\\
	&=&
	\sum_\lambda \int dk \,  G_\lambda(k,p) \, (-i\nabla_x+\Pf) \, e^{-ikx} \, a_\lambda(k) \, \psi(x)
	\nonumber\\
	&=&
	\sum_\lambda \int dk \,  G_\lambda(k,p)  \, e^{-ikx} \, a_\lambda(k) \,  (-i\nabla_x+k+\Pf-k) \, \psi(x)
	\nonumber\\
	&=&
	\sum_\lambda \int dk \,  G_\lambda(k,p) \,  e^{-ikx} \, a_\lambda(k)   \, \Ptot \,  \psi(x) \,.
\eeqn
Accordingly, we infer that $W_{\nablE(p)}(x) =\exp[D(x)-D^*(x)]$ commutes with $\Ptot$.

Furthermore, we observe that
\eqn
	W_{\nablE(p)}(x)  \, e^{i(p-\Pf)x} \, = \,  e^{i(p-\Pf)x} \, W_{\nablE(p)}
\eeqn
holds. Here and in what follows, we will use the abbreviated notation
\eqn
	 W_{\nablE(p)} \, \equiv \, W_{\nablE(p)}(x=0) \,.
\eeqn
We define the maps
\eqn
	(\Wmap\phi)(x) & := & \int dp \,  W_{\nablE(p)}(x)  \, e^{i(p-\Pf)x} \, \widehat\phi(p)
	\nonumber\\
	&=& \int dp \,   e^{i(p-\Pf)x} \, W_{\nablE(p)}  \, \widehat\phi(p) \,.
\eeqn
Likewise,
\eqn
	(\Wmap^*\phi)(x)
	& := & \int dp \,  W_{\nablE(p)}^*(x)  \, e^{i(p-\Pf)x} \, \widehat\phi(p)  \,.
\eeqn
The associated Bogoliubov-transformed Hamiltonian is given by
\eqn
	\Kns \, := \, (\Wmap\Hns\Wmap^{*}) \,.
\eeqn
We also introduce the Bogoliubov-transformed fiber Hamiltonians
\eqn
	\Kns(p) \, := \, W_{\nablE(p)} \, \Hns(p) \, W_{\nablE(p)}^* \,.
\eeqn
Then,we observe that
\eqn
	\Kns & = & (\Wmap \Hns \Wmap^{*})(x)
	\nonumber\\
	&=& \int W_{\nablE(p)}(x) \, e^{i(p-\Pf)x}\Hn(p)e^{-i(p-\Pf)x}  \, dP_{\Ptot}(p) \, W_{\nablE(p)}^*(x)
	\nonumber\\
	&=& \int W_{\nablE(p)}(x) \, e^{i(p-\Pf)x}\Hn(p)e^{-i(p-\Pf)x}  \, W_{\nablE(p)}^*(x) \, dP_{\Ptot}(p)
	\nonumber\\
	& = & \int e^{i(p-\Pf)x}\Kn(p)e^{-i(p-\Pf)x}dP_{\Ptot}(p).
\eeqn
In particular, we have that
\eqn
	\Wmap (\Hns\psi) \, = \, \Kns(\Wmap\psi) \,,
\eeqn
as can be readily verified.
Defining
\eqn
	\Phsig(p) \, := \, W_{\nablE(p)} \, \Psig(p) \,,
\eeqn
we obtain
\eqn
	\Kns(p) \, \Phsig(p) \, = \, \Egs(p) \, \Phsig(p) \,.
\eeqn
The following result is proven in \cite{cfp1}. 
\begin{thm}
For any $P\in \cS$,
the ground state eigenvector $\Phsig(P)$ of $\Kns(P)$ converges strongly in $\Fo$:
 $\Phi(P):=\lim_{\sigma\rightarrow0}\Phsig(P)$ exists in $\Fo$.
\end{thm} 

\bigskip

\section{Pull-through formulae}\label{sec:pullthrough}
In this appendix we prove the very useful ``pull-through'' formulae
(see \cite{bfs1})
\begin{equation}\label{eq:pull}
  a(k) f(H_f) = f(H_f + \omega(k)) a(k)
\end{equation}
and
\begin{equation}\label{eq:pull2}
  f(H_f) a^*(k) = a^*(k)f(H_f+\omega(k)),
\end{equation}
valid for any piecewise continuous, bounded
function, $f$, on $\BR$.
First, using the commutation relations for $a(k),\ a^*(k)$,
one proves relations~(\ref{eq:pull})-~(\ref{eq:pull2})
for $f(H) = (H_f - z)^{-1}$,
$z \in \BC / \bar{\BR}^+$.
Then using the Stone-Weierstrass theorem,
one can extend~(\ref{eq:pull})-~(\ref{eq:pull2})
from functions of the form
$f(\lambda) = (\lambda - z)^{-1}$,
$z \in \BC \backslash \bar{\BR}^+$,
to the class of functions mentioned above.
\DETAILS{Alternatively, (\ref{eq:pull}) -~(\ref{eq:pull2})
follow from the relation
\[
  f(H_f) \prod_{j=1}^N a^*(k_j) \Omega
  = f(\sum_{i=1}^N \omega(k_i)) \prod_{j=1}^N a^*(k_j) \Omega.
\]
\begin{hw}
{\em Prove this last relation, and
derive~(\ref{eq:pull})-~(\ref{eq:pull2})
from it.}
\end{hw}}

\bigskip
%
%
%

\section{Supplement: Creation and Annihilation Operators}\label{sec:crannihoprs}

Let $ \fh$ be either $ L^2 (\RR^3, \mathbb{C}, d^3 k)$ or  $ L^2
(\RR^3, \mathbb{C}^2, d^3 k)$. In the first case we consider $ \fh$
to be the Hilbert space of one-particle states of a scalar boson or
 phonon, and in the second case,  of a photon. The variable
$k\in\RR^3$ is the wave vector or momentum of the particle. (Recall
that throughout these lectures, the propagation speed $c$, of photon or
photons and Planck's constant, $\hbar$, are set equal to 1.) The
Bosonic Fock space, $\cF$, over $\fh$ is defined by
\begin{equation} \label{fock}
\cF \ := \ \bigoplus_{n=0}^{\infty} \cS_n \, \fh^{\otimes n} , 
\end{equation}
where $\cS_n$ is the orthogonal projection onto the subspace of
totally symmetric $n$-particle wave functions contained in the
$n$-fold tensor product $\fh^{\otimes n}$ of $\fh$; and $\cS_0
\fh^{\otimes 0} := \CC $. The vector $\Om:= (1, 0, ... )$ 
is called the \emph{vacuum vector} in
$\cF$. Vectors $\Psi\in \cF$ can be identified with sequences
$(\psi_n)^{\infty}_{n=0}$ of $n$-particle wave functions,  which are
totally symmetric in their $n$ arguments, and $\psi_0\in\CC$. In the
first case these functions are of the form, $\psi_n(k_1, \ldots,
k_n)$, while in the second case, of the form $\psi_n(k_1, \lambda_1,
\ldots, k_n, \lambda_n)$, where $\lambda_j \in \{-1, 1\}$ are the
polarization variables.

\DETAILS{The Bosonic Fock space, $\cF$, over $L^2 (\R^3, \mathbb{C},
d^3 k)$ (or  $ L^2 (\mathbb{R}^3, \mathbb{C}^2, d^3 k)$) is defined by
\begin{equation} \label{eq-I.10}
\cF \ := \ \bigoplus_{n=0}^{\infty} \cS_n \, L^2 (\mathbb{R}^3, \C^2,
d^3 k)^{\otimes n} ,
\end{equation}
where $\cS_n$ is the orthogonal projection onto the subspace of
totally symmetric $n$-particle wave functions contained in the
$n$-fold tensor product $L^2 (\RR^3, \mathbb{C}, d^3 k)^{\otimes n}$
of $L^2 (\RR^3, \mathbb{C}, d^3 k)$; and $\cS_0 L^2 (\RR^3,
\mathbb{C}, d^3 k)^{\otimes 0} := \BC $. The vector $\Omega:=1
\bigoplus_{n=1}^{\infty}0$ is called the \emph{vacuum vector} in
$\cF$. Vectors $\Psi\in \cF$ can be identified with sequences
$(\psi_n)^{\infty}_{n=0}$ of $n$-particle wave functions,
$\psi_n(k_1, \ldots, k_n)$,  which are
totally symmetric in their $n$ arguments, and $\psi_0 \in \BC$.

and $a^*(k)$ and $a(k)$ denote the creation and annihilation
operators on  $\cF$. The families $a^*(k)$ and $a(k)$ are
operator-valued generalized, transverse vector fields: (below $a_{\lambda}^{\#}= a_{\lambda}$ or $a_{\lambda}^*$.)
$$a^\#(k):= \sum_{\lambda \in \{-1, 1\}}
e_{\lambda}(k) a^\#_{\lambda}(k),$$ where $e_{\lambda}(k)$ are
polarization vectors, i.e. orthonormal vectors in $\mathbb{R}^3$
satisfying $k \cdot e_{\lambda}(k) =0$, and $a^\#_{\lambda}(k)$ are
scalar creation and annihilation operators, satisfying the \emph{canonical
commutation relations}:
\begin{equation} \label{ccr}
\big[ a_{\lambda}^{\#}(k) \, , \, a_{\lambda'}^{\#}(k') \big] \ = \ 0 , \hspace{8mm}
\big[ a_{\lambda}(k) \, , \, a_{\lambda'}^*(k') \big] \ = \ \del_{\lam\lam'}\delta^3 (k-k').
\end{equation}}
%
%
In what follows we present some key definitions in the first case,
limiting ourselves to remarks at the end of this appendix on how
these definitions have to be modified for the second case. The
scalar product of two vectors $\Psi$ and $\Phi$ is given by
\begin{equation} \label{F-scalprod}
\langle \Psi \, , \; \Phi \rangle \ := \ \sum_{n=0}^{\infty}  \int
\prod^n_{j=1} d^3k_j \; \overline{\psi_n (k_1, \ldots, k_n)} \:
\varphi_n (k_1, \ldots, k_n) .
\end{equation}

Given a one particle dispersion relation $\omega(k)$, the energy of
a configuration of $n$ \emph{non-interacting} field particles with
wave vectors $k_1, \ldots,k_n$ is given by $\sum^{n}_{j=1}
\omega(k_j)$. We define the \emph{free-field Hamiltonian}, $H_f$,
giving the field dynamics, by
%
\begin{equation} \label{Hfn}
(H_f \Psi)_n(k_1,\ldots,k_n) \ = \ \Big( \sum_{j=1}^n \omega(k_j)
\Big) \: \psi_n (k_1, \ldots, k_n) ,
\end{equation}
for $n\ge1$ and $(H_f \Psi)_n =0$ for $n=0$. Here
$\Psi=(\psi_n)_{n=0}^{\infty}$ (to be sure that the r.h.s. makes
sense we can assume that $\psi_n=0$, except for finitely many $n$,
for which $\psi_n(k_1,\ldots,k_n)$ decrease rapidly at infinity).
Clearly that the operator  $H_f$ has the single eigenvalue  $0$ with
the eigenvector $\Omega$ and the rest of the spectrum absolutely
continuous.

With each function $\varphi \in L^2 (\RR^3, \mathbb{C}, d^3 k)$ one
associates an \emph{annihilation operator} $a(\varphi)$ defined as
follows. For $\Psi=(\psi_n)^{\infty}_{n=0}\in \cF$ with the property
that $\psi_n=0$, for all but finitely many $n$, the vector
$a(\varphi) \Psi$ is defined  by
\begin{equation} \label{a}
(a(\varphi) \Psi)_n (k_1, \ldots, k_n) \ := \ \sqrt{n+1 \,} \, \int
d^3 k \; \overline{\varphi(k)} \: \psi_{n+1}(k, k_1, \ldots, k_n).
\end{equation}
These equations define a closable operator $a(\varphi)$ whose
closure is also denoted by $a(\varphi)$. Eqn \eqref{eq-I.12} implies
the relation
\begin{equation} \label{aOm}
a(\varphi) \Omega \ = \ 0 ,
\end{equation}
The creation operator $a^*(\varphi)$ is defined to be the adjoint of
$a(\varphi)$ with respect to the scalar product defined in
Eq.~\eqref{F-scalprod}. Since $a(\varphi)$ is anti-linear, and
$a^*(\vphi)$ is linear in $\varphi$, we write formally
\begin{equation} \label{ak}
a(\varphi) \ = \ \int d^3 k \; \overline{\varphi(k)} \, a(k) ,
\hspace{8mm} a^*(\varphi) \ = \ \int d^3 k \; \varphi(k) \, a^*(k),
\end{equation}
where $a(k)$ and $a^*(k)$ are unbounded, operator-valued
distributions. The latter are well-known to obey the \emph{canonical
commutation relations} (CCR):
\begin{equation} \label{CCR}
\big[ a^{\#}(k) \, , \, a^{\#}(k') \big] \ = \ 0 , \hspace{8mm}
\big[ a(k) \, , \, a^*(k') \big] \ = \ \delta^3 (k-k') ,
\end{equation}
where $a^{\#}= a$ or $a^*$.

Now, using this one can rewrite the quantum Hamiltonian $H_f$ in
terms of the creation and annihilation operators, $a$ and $a^*$, as
\begin{equation} \label{Hfa}
H_f \ = \ \int d^3 k \; a^*(k)\; \omega(k) \; a(k) ,
\end{equation}
acting on the Fock space $ \cF$.

More generally, for any operator, $t$, on the one-particle space $
L^2 (\mathbb{R}^3, \mathbb{C}, d^3 k)$ we define the operator $T$ on
the Fock space $\cF$ by the following formal expression $T: = \int
a^*(k) t a(k) dk$, where the operator $t$ acts on the $k-$variable
($T$ is the second quantization of $t$). The precise meaning of the
latter expression can obtained by using a basis $\{\phi_j\}$ in the
space $ L^2 (\mathbb{R}^3, \mathbb{C}, d^3 k)$ to rewrite it as $T:
= \sum_{j} \int a^*(\phi_j) a(t^* \phi_j) dk$.

To modify the above definitions to the case of photons, one replaces
the variable $k$ by the pair $(k, \lambda)$ and adds to the
integrals in $k$ sums over $\lambda$. In particular, the creation-
and annihilation operators have now two variables: $a_
\lambda^\#(k)\equiv a^\#(k, \lambda)$; they satisfy the commutation
relations
\begin{equation}
\big[ a_{\lambda}^{\#}(k) \, , \, a_{\lambda'}^{\#}(k') \big] \ = \
0 , \hspace{8mm} \big[ a_{\lambda}(k) \, , \,
a_{\lambda'}^*(k') \big] \ = \ \delta_{\lambda, \lambda'} \delta^3
(k-k') .
\end{equation}
One can introduce the operator-valued transverse vector fields by
$$a^\#(k):= \sum_{\lambda \in \{-1, 1\}} e_{\lambda}(k) a_{\lambda}^\#(k),$$
where $e_{\lambda}(k) \equiv e(k, \lambda)$ are polarization
vectors, i.e., orthonormal vectors in $\mathbb{R}^3$ satisfying $k
\cdot e_{\lambda}(k) =0$. Then, in order to reinterpret the
expressions in this paper for photons, one either adds the variable
$\lambda$, as was mentioned above, or replaces, in appropriate
places, the usual product of scalar functions or scalar functions
and scalar operators by the dot product of vector-functions or
vector-functions and operator-valued vector-functions.



\bigskip
%
%
\DETAILS{
\section{Related Problems}\label{sec:relatprobl}
Similar techniques are or can be used to obtain

\begin{itemize}

\item The mass renormalization for electrons,

\item Local decay of scattering states,

\item Existence and stability of thermal states.
\end{itemize}
We will discuss this at the end of these lectures. }

\bigskip
\DETAILS{
where $X_{\theta} := U_{\theta}e^{-ig F} $ with $F$, the
self-adjoint operator \textbf{defined below}. The
transformation $H \rightarrow e^{-ig F} H_{g}^{SM} e^{ig
F}$ is a generalization of the well-known Pauli-Fierz
transformation. Note that the operator-family $X_{\theta}$ has the
following two properties needed in order to establish the desired
properties of the resonances:

(a) $U_{\theta}$ are unitary for $\theta \in \mathbb{R}$;

(b) $U_{\theta_1 +\theta_2}= U_{\theta_1}U_{\theta_2}$ where
$U_{\theta}$ are unitary for $\theta \in \mathbb{R}$. }

%
%

\bigskip

\DETAILS{
To prove the results mentioned above we apply the
spectral renormalization group (RG) method
(BachChenFroehlichSigal2003,
BachFroehlichSigal1998a,BachFroehlichSigal1998b, GriesemerHasler2,
FroehlichGriesemerSigal2008b) to the Hamiltonians $H_{g \theta=0}^{SM}= e^{-ig F}
H_{g}^{SM} e^{ig F}$ (the ground state case) and $H_{g \theta}^{SM},\ \rIm \theta >0,$
(the resonance case).
Note that the version of RG needed in this work uses new --
anisotropic -- Banach spaces of operators, on which the
renormalization group acts.
Using the RG technique we describe the spectrum of the operator
$H^{SM} _{g \theta}$ in $\{z \in \mathbb{C}^- |\ \epsilon_{0} < \rRe
z < \nu$\} from which we derive Theorems \ref{thm-main} and
\ref{thm-main2}. }


\begin{thebibliography}{10}

\bibitem{AFFS} W. Abou Salem,  J. Faupin, J. Fr\"ohlich,
I.M.Sigal, \newblock On theory of resonances in non-relativisitc QED, {\it Advances in Applied Mathematics} 43 (2009), pp. 201-230.

\bibitem{AGG} L. Amour, B. Gr\'ebert,  J.-C. Guillot, \newblock The dressed mobile atoms and ions,  {\em J. Math. Pures Appl.} (9)  86,  no. 3, 177--200,   2006.

\bibitem{Arai1999}
A. Arai, \newblock Mathematical analysis of a model in relativistic quantum
electrodynamics. Applications of renormalization group methods in
mathematical sciences (in Japanese) Kyoto, 1999.


\bibitem{AraiHirokawa}
A. Arai,  M. Hirokawa, Ground states of a general class of
quantum field Hamiltonians, {\em Rev. Math. Phys.} 12 , no. 8,
1085--1135, 2000.

\bibitem{AraiHirokawaHiroshima}
A. Arai, M. Hirokawa and Fumio Hiroshima, Regularities of ground states of quantum field models, arXiv.



\bibitem{bach}
V.~Bach, \newblock Mass renormalization in nonrelativisitic quantum electrodynamocs, 
in Quantum Theory from Small to Large Scales: Lecture Notes of the Les Houches Summer Schools, volume 95. Oxford University Press, 2011.


\bibitem{bcfs1}
V.~Bach, T. Chen, J.~Fr{\"{o}}hlich, and I.~M. Sigal,
\newblock Smooth Feshbach map and operator-theoretic
renormalization group methods,
\newblock {\em Journal of Functional Analysis}, 203, 44-92, 2003.

\bibitem{bcfs2}
V.~Bach, T. Chen, J.~Fr{\"{o}}hlich, and I.~M. Sigal,
\newblock The renormalized electron mass in non-relativistic quantum electrodynamics,
\newblock {\em Journal of Functional Analysis}, 243, 426--535, 2007.

\bibitem{BachFroehlichPizzo1}
V.~Bach, J.~Fr{\"{o}}hlich, and A.~Pizzo,
\newblock Infrared-Finite Algorithms in QED: The Groundstate of an atom interacting with the quantized radiation field,
\newblock {\em  Communications in Mathematical Physics}
264, Issue: 1, 145 - 165, 2006.


\bibitem{BachFroehlichPizzo2}
V.~Bach, J.~Fr\"ohlich, and A.~Pizzo, \newblock An
infrared-finite algorithm for Rayleigh scattering amplitudes, and
Bohr's frequency condition, \newblock {\em Comm. Math. Phys.} 274,
no. 2, 457--486, 2007.

\bibitem{BachFroehlichPizzo3}
V.~Bach, J.~Fr\"ohlich, and A.~Pizzo, \newblock Infrared-finite
algorithms in QED II. The expansion of the groundstate of an atom
interacting with the quantized radiation field, \newblock {\em Adv.~in Math.~}, 220 (4), 1023-1074, 2009. mp\_arc.


\bibitem{bfs95}
V.~Bach, J.~Fr\"ohlich, and I.~M. Sigal,
\newblock Mathematical theory of non-relativistic matter and radiation,
\newblock {\em Letters~in Math.~Physics}, 34:183--201, 1995.


\bibitem{bfs1}
V.~Bach, J.~Fr\"ohlich, and I.~M. Sigal,
\newblock Quantum electrodynamics of confined non-relativistic particles,
\newblock {\em Adv.~in Math.~}, 137:299--395, 1998.

\bibitem{bfs2}
V.~Bach, J.~Fr\"ohlich, and I.~M. Sigal,
\newblock Renormalization group analysis of spectral problems in quantum field
  theory,
\newblock {\em Adv.~in Math.~}, 137:205--298, 1998.

\bibitem{bfs3}
V.~Bach, J.~Fr\"ohlich, and I.~M. Sigal,
\newblock Spectral analysis for systems of atoms and molecules coupled to the
  quantized radiation field,
\newblock {\em Commun.~Math.~Phys.}, 207(2):249--290, 1999.


\bibitem{bfs4}
V.~Bach, J.~Fr\"ohlich, and I.~M. Sigal,
\newblock Return to equilibrium,
\newblock {\em J.~Math.~Phys.}, 41(6):3985--4060, 2000.

\bibitem{bfss}
V.~Bach, J.~Fr\"ohlich, I.~M. Sigal, and A.~Soffer,
\newblock Positive commutators and spectrum of {P}auli-{F}ierz {H}amiltonian of
  atoms and molecules,
\newblock {\em Commun.~Math.~Phys.}, 207(3):557--587, 1999.

\bibitem{bkz} V.~Bach, F. Klopp, and H. Zenk, \newblock Mathematical analysis of the photoelectric effect. \newblock {\em Adv. Theor. Math. Phys.}, 5(6):969–999, 2001.

\bibitem{bcv} J.-M. Barbaroux, T.~Chen,  and S. A. Vugalter, \newblock  Binding conditions for atomic N-electron systems in non-relativistic QED,
\newblock {\em   Ann. H. Poinc.}, 4 (6), 1101 - 1136, 2003.


\bibitem{bcvv} J.-M. Barbaroux, T.~Chen, V. Vougalter and S. A. Vugalter, \newblock  Quantitative estimates on the hydrogen ground state energy in non-relativistic QED,
\newblock {\em  Ann. H. Poincare}, 11 (8), 1487-1544, 2010.



\bibitem{BDG}
 J.-M. Barbaroux,  M. Dimassi,  J.C. Guillot,
\newblock Quantum electrodynamics of relativistic bound states
with cutoffs,
\newblock {\em J. Hyperbolic Differ. Equ.} \textbf{1}, (2004) 271--314.

\bibitem{bv} J.-M. Barbaroux  and S. A. Vugalter, \newblock  Non analyticity of the ground state energy of the Hamiltonian for Hydrogen atom in nonrelativistic QED (2010), Preprint.


\bibitem{Berger} M.~Berger,
\newblock {\em  Nonlinearity and functional analysis. Lectures on
nonlinear problems in mathematical analysis.} Pure and Applied
Mathematics. Academic Press,
New York-London, 1977.



\bibitem{BugliaroFroehlichGraff} L. Bugliaro, J. Fr\"ohlich and G.M. Graf, \newblock Stability of quantum electrodynamics with nonrelativistic
matter,\newblock {\em Phys.Rev. Lett.} 77 (1996), 3494-3497.



\bibitem{CEH}
I. Catto, P. Exner, Ch. Hainzl,  \newblock Enhanced binding
revisited for a spinless particle in nonrelativistic QED,
\newblock {\em J. Math. Phys.} 45, no. 11, 4174--4185, 2004.

\bibitem{CH}
I. Catto, Ch. Hainzl, \newblock Self-energy of one electron in
non-relativistic QED, \newblock {\em J. Funct. Anal.} 207, no. 1,
68--110, 2004.

\bibitem{ch}
T.~Chen,
\newblock Infrared renormalization in non-relativistic QED and scaling criticality,
\newblock {\em J. Funct. Anal.} 254, no. 10,
2555--2647, 2008.


\bibitem{cffs}
T. Chen, J. Faupin, J. Fr\"ohlich, I.M. Sigal,
Local decay in non-relativistic QED, {\em arXiv}. 

\bibitem{chfr} T. Chen, J. Fr\"ohlich,
 Coherent infrared representations in non-relativistic QED,
{\em 	Spectral Theory and Mathematical Physics: A Festschrift in Honor of Barry Simon's 60th Birthday,
Proc. Symp. Pure Math.}, AMS, 2007.

\bibitem{cfp1} T. Chen, J. Fr\"ohlich, A. Pizzo,
 \newblock Infraparticle scattering states in non-relativistic QED - I. The Bloch-Nordsieck paradigm,
{\em Commun. Math. Phys.}, {\bf 294} (3), 761-825 (2010).

\bibitem{cfp2} T. Chen, J. Fr\"ohlich, A. Pizzo,
\newblock Infraparticle scattering states in non-relativistic QED - II. Mass shell properties,
{\em J. Math. Phys.}, {\bf 50} (1), 012103  (2009). arXiv:0709.2812


\bibitem{cvv}  T.~Chen, V. Vougalter and S. A. Vugalter, \newblock The increase of binding energy and enhanced binding in non-relativistic QED,
\newblock {\em   J. Math. Phys.}, 44 (5), 1961-1970, 2003.

\bibitem{Cohen-TannoudjiDupont-RocGrynberg1}
C.~Cohen-Tannoudji, J.~Dupont-Roc, and G.~Grynberg,
\newblock {\em Photons and Atoms -- Introduction to Quantum Electrodynamics}.
\newblock John Wiley, New York, 1991.

\bibitem{Cohen-TannoudjiDupont-RocGrynberg2}
C.~Cohen-Tannoudji, J.~Dupont-Roc, and G.~Grynberg,
\newblock {\em Atom-Photon Interactions -- Basic Processes and Applications}.
\newblock John Wiley, New York, 1992.



\bibitem{faupin}
J.~Faupin,
\newblock Resonances of the confined hydrogen atom and the Lamb-Dicke
effect in non-relativisitc quantum electrodynamics,
\newblock  {\em Ann. Henri Poincar\'e}  9,  no. 4, 743--773,  2008.


\bibitem{FeffermanFroehlichGraff} C. Fefferman, \newblock J. Fr\"ohlich and G.M. Graf, Stability of ultraviolet cutoff quantum electrodynamics
with non-relativistic matter,\newblock {\em Commun. Math. Phys.} 190 (1997), 309–330.



\bibitem{fermi}
E.~Fermi,
\newblock Quantum theory of radiation,
\newblock {\em Rev.~Mod.~Phys.}, 4:87--132, 1932.

\bibitem{Feshbach1958}
H.~Feshbach,
\newblock Unified theory of nuclear reactions,
\newblock {\em Ann.~Phys.}, 5:357--390, 1958.
.


\bibitem{fr1} J. Fr\"ohlich, \newblock On the infrared problem in a model of scalar electrons
and massless, scalar bosons, {\em Ann. Inst. Henri Poincar\'e}, Section Physique Th\'eorique,
{\bf 19} (1), 1-103 (1973).

\bibitem{fr2} J. Fr\"ohlich,  \newblock  Existence of dressed one electron states in a class of persistent models,  {\em Fortschritte der Physik} {\bf 22}, 159-198 (1974).


\bibitem{FroehlichGriesemerSchlein1}
J.~Fr\"ohlich, M.~Griesemer and B.~Schlein,
\newblock Asymptotic electromagnetic fields in models of quantum-mechanical matter
interacting with the quantized radiation field,
\newblock {\em  Advances in
Mathematics} 164, Issue: 2, 349-398, 2001.

\bibitem{FroehlichGriesemerSchlein2}
J.~Fr\"ohlich, M.~Griesemer and B.~Schlein,
\newblock Asymptotic completeness for Rayleigh scattering,
\newblock {\em Ann. Henri Poincar\'e} 3, no. 1,
107--170, 2002.

\bibitem{FroehlichGriesemerSchlein3}
J.~Fr\"ohlich, M.~Griesemer and B.~Schlein,
\newblock Asymptotic completeness for Compton scattering,
\newblock {\em Comm. Math. Phys.} 252, no. 1-3,
415--476, 2004.

\bibitem{fgs1}
J.~Fr\"ohlich, M.~Griesemer and I.M.~Sigal,
\newblock Spectral theory for the standard model
of non-relativisitc QED, {\em Comm. Math. Phys.}, 283, no 3,
613-646, 2008.

\bibitem{fgs2}
J.~Fr\"ohlich, M.~Griesemer and I.M.~Sigal,
\newblock Spectral renormalization group analysis, {\em Rev. Math. Phys.} 21(4) (2009) 511-548, e-print, ArXiv, 2008.

\bibitem{fgs3}
J.~Fr\"ohlich, M.~Griesemer and I.M.~Sigal,
\newblock Spectral renormalization group and limiting absorption principle
for the standard model of non-relativisitc QED,  {\em Rev. Math. Phys.} 23(2) (2011), e-print, ArXiv, 2010. 

\bibitem{FroehlichPizzo}
J. Fr\"ohlich, A. Pizzo, The renormalized electron mass in nonrelativistic QED,
 {\em Commun. Math. Phys.}

\bibitem{ghps} C. G\'erard, F. Hiroshima, A. Panati, and A. Suzuki, Infrared problem for the Nelson model on static space-times, {\em Commun. Math. Phys.}, 2011, To appear.

\bibitem{GergescuGerardMoeller2}
V.~Gergescu, C.~G\'erard, and J.S.~M\o ller,
\newblock Spectral Theory of massless Pauli-Fierz models,
\newblock {\em Commun. Math. Phys.}, 249:29--78, 2004.

\bibitem{Griesemer}
M. Griesemer,
\newblock Exponential decay and ionization thresholds in non-relativistic quantum electrodynamics, \newblock {\em J. Funct. Anal.}, 210(2):321--340, 2004.

\bibitem{GriesemerHasler1}
M.~Griesemer and D.~Hasler,
\newblock On the smooth {F}eshbach-{S}chur map,
\newblock {\em J. Funct. Anal.}, 254(9):2329--2335, 2008.

\bibitem{GriesemerHasler2}
M.~Griesemer and D.~Hasler,
\newblock Analytic perturbation theory and renormalization analysis of matter
  coupled to quantized radiation,
\newblock arXiv:0801.4458.



\bibitem{GriesemerLiebLoss}
M. Griesemer, E.H. Lieb and M. Loss,
\newblock Ground states in non-relativistic quantum electrodynamics, \newblock {\em Invent. Math.}  145, no. 3, 557--595, 2001.

\bibitem{GriesemerZenk}
M. Griesemer and H. Zenk, On the Atomic Photoeffect in Non-relativistic QED, Comm. Math. Phys.



\bibitem{GustafsonSigal} S. Gustafson and I.M. Sigal, \newblock {\em  Mathematical Concepts of Quantum Mechanics.} 2nd edition. \newblock Springer 2006.

\bibitem{H1}
Ch. Hainzl, \newblock  Enhanced binding through coupling to a photon field. Mathematical results in quantum mechanics, Contemp, Math., 307, Amer. Math. Soc., Providence,
RI (2002), 149–154.

\bibitem{H2}
Ch. Hainzl, \newblock One non-relativistic particle coupled to a
photon field, \newblock {\em Ann. Henri Poincar\'e} 4, no. 2,
217 -- 237, 2003.


\bibitem{HHS} Ch. Hainzl, M. Hirokawa, H. Spohn, \newblock  Binding energy for hydrogen-like atoms in the Nelson model without cutoffs, \newblock {\em  J. Funct. Anal.} 220, no. 2, 424–459 (2005).


\bibitem{HS}
Ch. Hainzl, R. Seiringer, Mass renormalization and energy level shift in non-relativistic QED, \newblock {\em Adv. Theor. Math. Phys.} 6 (5) 847 -- 871,  2003.

\bibitem{HVV}
Ch.~Hainzl, V.~Vougalter, S.~Vugalter,
\newblock  Enhanced binding in non-relativistic QED,
\newblock {\em Commun. Math. Phys.} 233, no. 1, 13--26, 2003.

\bibitem{HaslerHerbst2007}
D. Hasler and I. Herbst, \newblock  Absence of ground states for a class of translation invariant models of non-relativistic QED,  {\em Commun. Math. Phys.}  279,  no. 3, 769--787, 2008. ArXiv

\bibitem{HaslerHerbstHuber}
D. Hasler, I. Herbst and M.Huber,
\newblock  On the lifetime of quasi-stationary states in
non-relativistic QED,  {\em Ann. Henri Poincar\'e}  9,  no. 5, 1005--1028,   2008, ArXiv:0709.3856.


\bibitem{HaslerHerbst1}
D. Hasler and I. Herbst,
\newblock  Analytic perturbation theory and renormalization analysis of matter coupled to quantized radiation,    2008, ArXiv:0801.4458v1.


\bibitem{HaslerHerbst2}
D. Hasler and I. Herbst,
\newblock  Smoothness and analyticity of perturbation expansions in QED,  2010, ArXiv:1007.0969v1.


\bibitem{HaslerHerbst3}
D. Hasler and I. Herbst,
\newblock  Convergent expansions in non-relativistic QED,  2010, ArXiv:1005.3522v1.

\bibitem{HillePhillips}  \newblock E. Hille and R.S.Phillips,  {\em  Functional Analalysis and Semi-groups}. AMS 1957.


\bibitem{Hirokawa2}
M. Hirokawa, \newblock  Recent developments in mathematical methods
for models in non-relativistic quantum electrodynamics,
\newblock {\em  A garden of quanta,}  209--242, World Sci.
Publishing, River Edge, NJ, 2003.

\bibitem{Hiroshima1}
F. Hiroshima, \newblock  Ground states of a model in nonrelativistic
quantum electrodynamics. I, \newblock {\em J. Math. Phys.} 40
(1999), no. 12, 6209--6222.

\bibitem{Hiroshima2}
F. Hiroshima, \newblock  Ground states of a model in nonrelativistic
quantum electrodynamics. II, J. Math. Phys. 41 (2000), no. 2,
661--674.

\bibitem{Hiroshima2a}
F. Hiroshima, \newblock Essential self-adjointness of translation-invariant quantum field models for arbitrary coupling constants, Commun. Math. Phys. 211 (2000), 585–613.


\bibitem{Hiroshima3}
F. Hiroshima, \newblock  Ground states and spectrum of quantum
electrodynamics of nonrelativistic particles, \newblock {\em Trans.
Amer. Math.} Soc. 353 (2001), no. 11, 4497--4528 (electronic).

\bibitem{Hiroshima4}
F. Hiroshima, \newblock  Self-adjointness of the Pauli-Fierz
Hamiltonian for arbitrary values of coupling constants, \newblock
{\em Ann. Henri Poincar\'e} 3 (2002), no. 1, 171--201.

\bibitem{Hiroshima5}
F. Hiroshima, \newblock  Nonrelativistic QED at large momentum of
photons,\newblock {\em A garden of quanta},  167--196, World Sci.
Publishing, River Edge, NJ, 2003.

\bibitem{Hiroshima6}
F. Hiroshima, \newblock  Localization of the number of photons of
ground states in nonrelativistic QED, \newblock {\em Rev. Math.
Phys}. 15 (2003), no. 3, 271--312.

\bibitem{Hiroshima7}
F. Hiroshima, \newblock  Analysis of ground states of atoms
interacting with a quantized radiation field,  \newblock {\em Topics
in the theory of Schr\"odinger operators}, World Sci. Publishing,
River Edge, NJ, 2004, 145--272.


\bibitem{HiroshimaSasaki}
 F. Hiroshima and  I. Sasaki,
\newblock On the ionization energy of the
semi-relativistic Pauli-Fierz model for a single particle,
\newblock { Preprint}, arXiv:1003.1661v4 (2010)



\bibitem{HiroshimaSpohn}
F. Hiroshima and H. Spohn, \newblock  Ground state degeneracy of
the Pauli-Fierz Hamiltonian with spin. \newblock {\em Adv. Theor.
Math. Phys.} 5 (2001), no. 6, 1091--1104.

\bibitem{HiroshimaSpohn1}
F. Hiroshima and H. Spohn,. \newblock  Enhanced binding through coupling to a quantum field, \newblock {\em Ann. Inst. Henri Poincar\'e} 2(6) (2001), no. 6, 1159--1187.

\bibitem{HiroshimaSpohn2}
F. Hiroshima, H. Spohn, \newblock Mass renormalization in nonrelativistic quantum electrodynamics, \newblock {\em  J. Math. Phys.} 46 (4),
2005.



\bibitem{HuebnerSpohn1}
M.~H\"ubner and H.~Spohn,
\newblock  Radiative decay: nonperturbative approaches, \newblock {\em Rev. Math. Phys.} 7, no. 3, 363--387, 1995.

\bibitem{Hunziker}
W.~Hunziker, \newblock  Resonances, metastable states and
exponential decay laws in perturbation theory, \newblock {\em  Comm.
Math. Phys.} 132, no. 1, 177--188, 1990.

\bibitem{HunzikerSigal}
W.~Hunziker and I.M.~Sigal,
\newblock The quantum $N$-body problem, \newblock {\em  J. Math. Phys.} 41, no. 6, 3448--3510, 2000.


\bibitem{King} Ch. King, Resonant decay of a two state atom
interacting with a massless non-relativistic quantized scalar field,
Comm. Math. Phys. ,165(3), 569-594, (1994).




\bibitem{KMS} M. K\"onenberg, O. Matte, E. Stockmeyer, Existence of ground states of hydrogen-like atoms in relativistic QED I: The semi-relativistic
Pauli-Fierz operator,  arXiv:0912.4223.

\bibitem{LiebLoss1}  E.H. Lieb and M. Loss, Self-energy of electrons in non-perturbative QED, in: Conference Moshe Flato 1999, vol. I,
Dijon, in:  \newblock {\em Math. Phys. Stud.} vol. 21, Kluwer Acad. Publ., Dordrecht, 2000, pp. 327-344.

\bibitem{LiebLoss2}  E.H. Lieb and M. Loss, A bound on binding energies and mass renormalization in models of quantum electrodynamics,
 \newblock {\em J. Statist. Phys.} 108, 2002.


\bibitem{LiebLoss3} E.H. Lieb and M.~Loss, Existence of atoms and molecules in non-relativistic quantum electrodynamics,  \newblock {\em Adv. Theor. Math. Phys.} {\bf 7}, 667-710
(2003).  arXiv math-ph/0307046.

\bibitem{LMS} M. Loss, T. Miyao and H. Spohn, Lowest energy states in nonrelativistic QED: Atoms and ions in motion,
 \newblock {\em J. Funct. Anal.} 243, Issue 2, 353-393, 2007.

\bibitem{LMS2} M. Loss, T. Miyao and H. Spohn,  Kramers degeneracy theorem in nonrelativistic QED, 2008,  arXiv:0809.4471.



\bibitem{MatteStockmeyer}
O. Matte,  E. Stockmeyer,
\newblock Exponential localization for a hydrogen-like atom in relativistic
  quantum electrodynamics,
\newblock { Comm. Math. Phys.} \textbf{295}, 551--583 (2010)

\bibitem{MMS1}
M. Merkli, M. M\"uck, I.M. Sigal,
\newblock Instability of Equilibrium States for Coupled Heat Reservoirs at Different Temperatures,
\newblock {\em J. Funct. Anal.}, 243, no. 1, 87-120 (2007).


\bibitem{MMS2}
M. Merkli, M. M\"uck, I.M. Sigal,
\newblock Nonequilibrium stationary states and their stability,
\newblock {\em Ann. Inst. H. Poincar\'e}, 243, no. 1, 87-120 (2008).

\bibitem{MiyaoSpohn}
T. Miyao, H. Spohn, \newblock Spectral analysis of the semi-relativistic Pauli-Fierz Hamiltonian, \newblock {\em   J.~Funct. Anal.}, 256, 2123--2156 (2009), arxiv,
2008.


\bibitem{MueThesis}
M. M\"uck,
\newblock {\em Thermal Relaxation for Particle Systems in Interaction with Several Bosonic Heat Reservoirs}.
\newblock Ph.D. Dissertation, Department of Mathematics, Johannes-Gutenberg University, Mainz, July 2004, ISBN 3-8334-1866-4


\bibitem{Mueck}
M.~M\"uck, \newblock Construction of metastable states in quantum
electrodynamics, \newblock {\em Rev. Math. Phys.} 16, no. 1, 1--28,
2004.



\bibitem{PauliFierz} W. Pauli and M. Fierz, Zur Theorie der Emission langwelliger Lichtquanten, \newblock {\em Il Nuovo Cimento} 15 (1938) 167-188.   


\bibitem{Pizzo2003}
A.~Pizzo,
\newblock One-particle (improper) States in Nelson's massless model,
\newblock {\em Annales
Henri Poincar\'e} 4, Issue: 3, 439 - 486, June, 2003.

\bibitem{Pizzo2005}
A.~Pizzo,
\newblock Scattering of an infraparticle: The one particle sector in
Nelson's massless model,
\newblock {\em Annales Henri Poincar\'e} 6, Issue: 3, 553 - 606 ,
2005.




\bibitem{RS}
M.~Reed and B. Simon, \newblock {\em Methods of Modern Mathematical
Physics, volumes III, IV.} 
Academic Press, 1978.

\bibitem{sakurai}
J. J. Sakurai, \newblock {\em Advanced Quantum Mechanics},
Addison-Wesley, 1987.


\bibitem{Sasaki2006}
 I. Sasaki,
\newblock Ground state of a model in relativistic quantum electrodynamics
with a fixed total momentum,
\newblock { Preprint}, arXiv:math-ph/0606029v4 (2006)


\bibitem{Schur}
J. ~Schur,
\newblock \"Uber Potenzreihen die im Inneren des Einheitskreises beschr\"ankt sind,
\newblock {\em J. reine u. angewandte Mathematik}, 205--232, 1917.


\bibitem{sig}
I.M.~Sigal, \newblock Ground state and resonances in the standard
model of the non-relativistic quantum electrodynamics, \newblock {\em J.  Stat. Physics},  134  (2009),  no. 5-6, 899--939. arXiv 0806.3297.

\bibitem{Spohn}
H. Spohn, \newblock {\em Dynamics of Charged Particles and their
Radiation Field}, Cambridge University Press, Cambridge, 2004.


\bibitem{Stockmeyer}
 E. Stockmeyer,
\newblock On the non-relativistic limit of a model in quantum electrodynamics,
{ Preprint}, arXiv:0905.1006v1 (2009)

\end{thebibliography}
\end{document}